\documentclass[11pt]{article}
\usepackage[hmargin={2.5cm,2.5cm},vmargin={2.3cm,2.3cm}]{geometry}

\usepackage{amsmath,amssymb,amsthm}

\usepackage{pgf}
\usepackage{tikz}
\usetikzlibrary{arrows.meta}

\usepackage{MnSymbol}

\usepackage[skip=2pt,font=small]{caption}

\usepackage{hyperref}
\hypersetup{
colorlinks,
linkcolor={blue!60!black},
citecolor={blue!90!black},
urlcolor={blue!90!black}}

\usepackage{footnote}

\usepackage{cite}

\usepackage[titletoc,title]{appendix}

\newtheorem{theorem}{Theorem}[section]

\newtheorem{cor}[theorem]{Corollary}
\newtheorem{prop}[theorem]{Proposition}

\newcommand{\tocless}[2]{#1{#2}}

\newcommand{\ds}{\displaystyle}

\renewcommand{\author}[1]{\large\rm #1\\ \bigskip}
\newcommand{\address}[1]{{\normalsize\it #1\\}\bigskip}
\renewcommand{\title}[1]{\bigskip\bigskip\Large\bf #1\bigskip\bigskip\\}

\newcommand{\Bigpsi}[3]{\phantom{\Psi}_2 \kern -.05em
\Psi_2\left(\genfrac{}{}{0pt}{}{#1}{#2}\biggl|#3\right)}

\newcommand{\bea}{\begin{eqnarray}}
\newcommand{\eea}{\end{eqnarray}}

\newcommand{\ii}{\mathsf{i}}
\newcommand{\s}{S}

\newcommand{\ow}{\overline{W}}
\newcommand{\oW}{\overline{W}}
\newcommand{\oV}{\overline{V}}
\newcommand{\vphi}{\varphi}
\newcommand{\ovphi}{\overline{\varphi}}
\newcommand{\ovphib}{\overline{\phi}}

\newcommand{\iW}{W}

\newcommand{\lag}{{\mathcal L}}
\newcommand{\ol}{\overline{\mathcal{L}}}

\newcommand{\iC}{C}

\newcommand{\Log}{{\textrm{Log}\hspace{1pt}}}

\newcommand{\dilog}{{\textrm{Li}_2}}

\newcommand{\lie}{{\textrm{Li}_2}}
\newcommand{\olie}{{\textrm{Li}_2}}

\newcommand{\q}{{\mathsf q}}

\newcommand{\iS}{S}

\newcommand{\bb}{\mathsf{b}}

\newcommand{\olam}{{\overline{\Lambda}}}

\newcommand{\spn}{{\sigma}}

\newcommand{\astr}[1]{{\mathcal{A}_\medstar(\alpha_1,\alpha_3\,|\,#1,x_1,x_2,x_3)}}

\def\EXP{\textrm{{\large e}}}

\def\re{\mathop{\hbox{\rm Re}}\nolimits}
\def\im{\mathop{\hbox{\rm Im}}\nolimits}

\numberwithin{equation}{subsection}

\begin{document}

\vglue 1.5cm

\begin{center}

\title{Integrable quad equations derived from the quantum Yang-Baxter equation}
\author{Andrew P.~Kels}
\address{Institute of Physics, University of Tokyo,\\
 Komaba, Tokyo 153-8902, Japan}

\end{center}

\vskip 0.0cm 
\begin{abstract}

This paper presents an explicit correspondence between two different types of integrable equations; the quantum Yang-Baxter equation in its star-triangle relation form, and the classical 3D-consistent quad equations in the Adler-Bobenko-Suris (ABS) classification.  Each of the 3D-consistent ABS quad equations of $H$-type, are respectively derived from the quasi-classical expansion of a counterpart star-triangle relation.  Through these derivations it is seen that the star-triangle relation provides a natural path integral quantization of an ABS equation.  The interpretation of the different star-triangle relations is also given in terms of (hyperbolic/rational/classical) hypergeometric integrals, revealing the hypergeometric structure that links the two different types of integrable systems.  Many new limiting relations that exist between the star-triangle relations/hypergeometric integrals are proven for each case. 

\end{abstract}
\tableofcontents




\section{Introduction}

This paper is concerned with developing an explicit connection between two important but distinct principles of integrability, that are typically considered exclusively of each other.  The first principle is the quantum Yang-Baxter equation, a condition of integrability for two-dimensional models of statistical mechanics, that allows an exact solution of the model to be determined in the infinite lattice limit through the transfer matrix technique that was introduced by Baxter \cite{Baxter:1972hz,Baxter:1982zz}.  The second principle is the 3D-consistency condition, a condition of integrability for two-dimensional systems of quad (lattice) equations, where existence of traditional integrabililty characteristics such as zero curvature representations and B\"{a}cklund-Darboux transformations, follows from the consistency of the equations when considered in three or more dimensions \cite{nijhoffwalker,BobSurQuadGraphs,ABS}.

The key to the connection between these two principles of integrability, lies in the quasi-classical expansion of solutions to the quantum Yang-Baxter equation \cite{Bazhanov:2007mh,Bazhanov:2010kz,Bazhanov:2016ajm,Kels:2017fyt}.  The purpose of this paper is to complete a correspondence that has been presented in \cite{Bazhanov:2016ajm}, where each of the 3D-consistent $Q$-type quad equations in the Adler-Bobenko-Suris  (ABS)  classification \cite{ABS}, were shown to arise in the quasi-classical expansion of different respective solutions of the star-triangle relation (STR) form of the Yang-Baxter equation.  This paper extends the latter correspondence to the entire ABS list, through the use of explicit new solutions of an asymmetric form of the STR.  

Specifically, for some complex valued spin variables $\spn_1,\spn_2,\spn_3$, and spectral variables $\theta_1,\theta_3$, the asymmetric form of the STR takes the form
\begin{align}
\label{YBEintro}
\ds\int_\mathbb{R} d\spn_0 \,\rho(\theta_1,\theta_3\,|\,\spn_1,\spn_2,\spn_3;\spn_0;\hbar)= V(\theta_1\,|\,\spn_2,\spn_3;\hbar)\,\oV(\theta_1+\theta_3\,|\,\spn_1,\spn_3;\hbar)\, W(\theta_3\,|\,\spn_2,\spn_1;\hbar)\,,
\end{align}
where $\hbar$ is a real valued parameter, and the integrand is
\begin{align}
\label{quad3}
\rho(\theta_1,\theta_3\,|\,\spn_1,\spn_2,\spn_3;\spn_0;\hbar)=\oV(\theta_1\,|\,\spn_1,\spn_0;\hbar)\, V(\theta_1+\theta_3\,|\,\spn_2,\spn_0;\hbar)\,\oW(\theta_3\,|\,\spn_0,\spn_3;\hbar)\,.
\end{align}
The $V(\theta\,|\,\spn_i,\spn_j)$, $\oV(\theta\,|\,\spn_i,\spn_j)$, $W(\theta\,|\,\spn_i,\spn_j)$, $\oW(\theta\,|\,\spn_i,\spn_j)$ are functions known as Boltzmann weights, which for \eqref{YBEintro} will typically be complex valued.  The key to obtaining the 3D-consistent quad equations, is in the left hand side of \eqref{YBEintro} in the limit of the parameter $\hbar\to0$.  For a change of variables of the form $\sigma_i\to\sigma'_i= f_i(x_i)$, $\theta_j\to\theta'_j= g_j(\alpha_j)$, an asymptotic expansion of the integrand \eqref{quad3} for $\hbar\to0$ can always be found, that takes the form
\begin{align}
\label{quad2}
\Log\rho(\theta_1',\theta_3'\,|\,\spn_1',\spn_2',\spn_3';\spn_0';\hbar)=\hbar^{-1}\astr{x_0}+O(1)\,,
\end{align}
where
\begin{align}
\label{astrintro}
\astr{x_0}=\olam(\alpha_1\,|\,x_0,x_1)+\Lambda(\alpha_1+\alpha_3\,|\,x_0,x_2)+\ol(\alpha_3\,|\,x_0,x_3)\,.
\end{align}
The Lagrangian functions $\olam$, $\Lambda$, $\ol$, coincide with the $O(\hbar^{-1})$ terms of the expansion of the Boltzmann weights $\oV$, $V$, $\oW$, respectively.  The saddle point equation for \eqref{YBEintro}, is then given by the {three-leg equation}
\begin{align}
\label{quad1}
\frac{\partial\astr{x_0}}{\partial x_0}=\ovphib(\alpha_1\,|\,x_0,x_1)+\phi(\alpha_1+\alpha_3\,|\,x_0,x_2)+\ovphi(\alpha_3\,|\,x_0,x_3)=0\,,
\end{align}
for some $\ovphib$, $\phi$, $\ovphi$, which coincide with the respective derivatives of $\olam$, $\Lambda$, $\ol$.  For explicit cases found in this paper, the saddle point equation \eqref{quad1} can be transformed via a specific change of variables into an equation of the form
\begin{align}
\label{quad0}
Q(x,u,y,v;\alpha,\beta)=0\,,
\end{align}
where the function $Q(x,u,y,v;\alpha,\beta)$ is a polynomial that is linear in each of the $x$, $u$, $y$, $v$ (known as the ``affine-linear'' property), and $Q(x,u,y,v;\alpha,\beta)$ satisfies the {\it 3D-consistency} integrability condition \cite{nijhoffwalker,BobSurQuadGraphs,ABS}.  

The fact that the equation \eqref{quad0}, can have an affine-linear form and is 3D-consistent is not obvious, since the original form of the Boltzmann weights, and resulting equations \eqref{quad2}, \eqref{quad1}, are generally non-trivial (elliptic/hyperbolic/rational/algebraic) functions of the variables $\spn_0$, $\spn_1$, $\spn_2$, $\spn_3$, and $x_0$, $x_1$, $x_2$, $x_3$ respectively.  Specifically, it will be seen that all such equations \eqref{quad0} that are obtained from a STR of the form \eqref{YBEintro}, may be identified as $H$-type 3D-consistent affine-linear equations in the ABS classification \cite{ABS,ABS2}.

This is a rather interesting and unexpected connection, particularly since the latter classification was originally derived independently of the existence of any Yang-Baxter structure.  In this paper, the derivation of 3D-consistent quad equations \eqref{quad0} from the STR \eqref{YBEintro}, as outlined above, will be referred to as {\it Yang-Baxter/3D-consistency correspondence}.

There is another important application of the star-triangle relation \eqref{YBEintro}, besides the aforementioned application to integrability.  Namely, explicit solutions of \eqref{YBEintro} can be shown to be equivalent to transformation formulas for hypergeometric functions.  Some examples that are new in this paper include the derivation of the $H1_{(\varepsilon=0)}$ equation from a STR which is equivalent to a sum of two classical Euler beta functions, and the derivation of the equations $H2_{(\varepsilon=1)}$, $H2_{(\varepsilon=0)}$, $H1_{(\varepsilon=1)}$, from STRs that are respectively equivalent to different hypergeometric integral formulas of Barnes \cite{Barnes1908,Barnes1910}.  This result suggests that hypergeometric integrals are a universal structure behind these types of integrable systems.  The summary of the correspondence found between the STR, hypergeometric integrals, and 3D-consistent equations, is given in Appendix \ref{app:hyper}.

The structure of the paper is as follows.  In Section \ref{sec:overview}, the Yang-Baxter/3D-consistency correspondence is discussed in further detail, and in Sections \ref{sec:hyplim}, \ref{sec:ratlim}, \ref{sec:alglim}, explicit derivations are given of the Yang-Baxter/3D-consistency correspondence for all $Q$ and $H$-type ABS equations at the hyperbolic, rational, and algebraic levels respectively.  Finally, the connection to hypergeometric functions is summarised in Appendix \ref{app:hyper}, and also some non-integrable examples, which have a slightly different form from the STR \eqref{YBEintro}, are considered in Appendix \ref{app:notYBE}.

\section{Yang-Baxter/3D-consistency correspondence}
\label{sec:overview}

\subsection{Star-triangle relation}

To begin with the star-triangle relation will be described  in more detail.  A usual form of the star-triangle relation for exactly solved lattice models of statistical mechanics, can be written as \cite{Baxter:1982zz,Bax02rip,PerkYBEs}
\begin{align}
\label{YBE}
\begin{split}
\ds\int_\mathbb{R} d\spn_0 \,  S(\spn_0)\,\oW_{qr}(\spn_1,\spn_0)\, W_{pr}(\spn_2,\spn_0)\,\oW_{pq}(\spn_0,\spn_3)\phantom{\,.}\\
\ds=R_{pqr} W_{qr}(\spn_2,\spn_3)\,\oW_{pr}(\spn_1,\spn_3)\, W_{pq}(\spn_2,\spn_1)\,.
\end{split}
\end{align}
In \eqref{YBE} there are three independent, ``spin'' variables $\spn_1,\spn_2,\spn_3$, and three rapidity variables $p,q,r$.  Typically these variables are taken to be real valued, but for the purposes of this paper the variables can take general complex values.  The two functions $W_{pq}(\spn_i,\spn_j)$, $\oW_{pq}(\spn_i,\spn_j)$, are the Boltzmann weights.  In this paper, the Boltzmann weights will generally be seen to be complex valued, and consequently the Boltzmann weights do not describe ``physical'' interactions of an associated lattice model.  The $S(\spn_i)$, may be interpreted as another type of Boltzmann weight, that depends on the value of the individual spin $\spn_i$, and is independent of the rapidity variables.  It is sometimes the case that $S(\spn_i)=1$, otherwise it is a non-trivial function of $\spn_i$.  The $R_{pqr}$ is a spin independent normalisation factor, that depends on the values of $p,q,r$.

The star-triangle relation \eqref{YBE} also has a natural graphical interpretation.  If the spin variables $\spn_i$, are assigned to vertices, and the Boltzmann weights $W_{pq}(\spn_i,\spn_j)$, $\oW_{pq}(\spn_i,\spn_j)$, are assigned to edges connecting two vertices $\spn_i$, $\spn_j$, then \eqref{YBE} has the graphical form shown in Figure \ref{YBErapidityfig}.

Unless otherwise stated, the symmetry
\begin{align}
\label{spinsymmetry}
W_{pq}(\spn_i,\spn_j)=W_{pq}(\spn_j,\spn_i)\,,\qquad \oW_{pq}(\spn_i,\spn_j)=\oW_{pq}(\spn_j,\spn_i)\,,
\end{align}
is not assumed, and thus the ordering of variables in \eqref{YBE} will need to be taken into consideration.  This also means that a second star-triangle relation where the two spin variables in the argument of each Boltzmann weight of \eqref{YBE} are exchanged (this second star-triangle relation graphically corresponds to reversing the orientation of each rapidity line in Figure \ref{YBErapidityfig}), in general is not satisfied for the cases considered here.

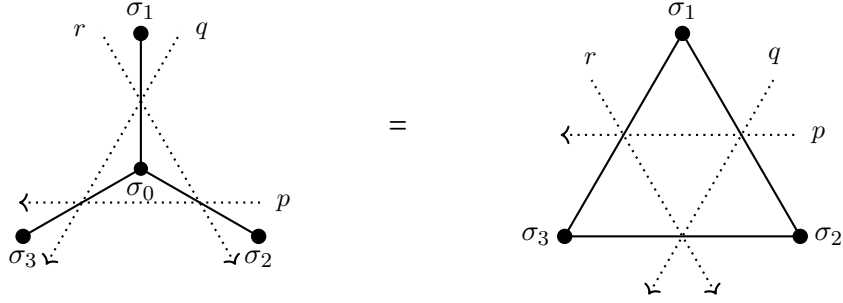
\begin{figure}[h]
\centering

\begin{tikzpicture}[scale=1.8]

\draw[-,thick] (-2,0)--(-2,1);
\draw[-,thick] (-2,0)--(-2.87,-0.5);
\draw[-,thick] (-2,0)--(-1.13,-0.5);
\draw[<-,thick,dotted] (-2.9,-0.25)--(-1.1,-0.25);
\fill[white!] (-1.1,-0.25) circle (0.5pt)
node[right=1.5pt]{\color{black}\small $p$};
\draw[<-,thick,dotted] (-2.7,-0.71)--(-1.7,1.02);
\fill[white!] (-1.7,1.02) circle (0.5pt)
node[right=1.5pt]{\color{black}\small $q$};
\draw[<-,thick,dotted] (-1.3,-0.71)--(-2.3,1.02);
\fill[white!] (-2.3,1.02) circle (0.5pt)
node[left=1.5pt]{\color{black}\small $r$};
\fill (-2,0) circle (1.5pt)
node[below=1.5pt]{\color{black} $\spn_0$};
\filldraw[fill=black,draw=black] (-2,1) circle (1.5pt)
node[above=1.5pt] {\color{black} $\spn_1$};
\filldraw[fill=black,draw=black] (-2.87,-0.5) circle (1.5pt)
node[below=1.5pt] {\color{black} $\spn_3$};
\filldraw[fill=black,draw=black] (-1.13,-0.5) circle (1.5pt)
node[below=1.5pt] {\color{black} $\spn_2$};

\fill[white!] (0.05,0.3) circle (0.01pt)
node[left=0.05pt] {\color{black}\Large $=$};

\draw[-,thick] (2,1)--(1.13,-0.5);
\draw[-,thick] (1.13,-0.5)--(2.87,-0.5);
\draw[-,thick] (2.87,-0.5)--(2,1);
\draw[<-,thick,,dotted] (1.1,0.25)--(2.85,0.25);
\fill[white!] (2.85,0.25) circle (0.5pt)
node[right=1.5pt]{\color{black}\small $p$};
\draw[<-,thick,dotted] (1.75,-0.93)--(2.68,0.67);
\fill[white!] (2.68,0.67) circle (0.5pt)
node[above=1.5pt]{\color{black}\small $q$};
\draw[<-,thick,dotted] (2.25,-0.93)--(1.32,0.67);
\fill[white!] (1.32,0.67) circle (0.5pt)
node[above=1.5pt]{\color{black}\small $r$};
\filldraw[fill=black,draw=black] (2,1) circle (1.5pt)
node[above=1.5pt]{\color{black} $\spn_1$};
\filldraw[fill=black,draw=black] (1.13,-0.5) circle (1.5pt)
node[left=1.5pt]{\color{black} $\spn_3$};
\filldraw[fill=black,draw=black] (2.87,-0.5) circle (1.5pt)
node[right=1.5pt]{\color{black} $\spn_2$};
\end{tikzpicture}
\caption{Graphical representation of the star-triangle relation form of the quantum Yang-Baxter equation \eqref{YBE}.}
\label{YBErapidityfig}
\end{figure}

The directed dotted lines in Figure \ref{YBErapidityfig}, are associated with the rapidity variables that appear in the star-triangle relation \eqref{YBE}.  The crossing of the directed rapidity lines over the solid edges in Figure \ref{YBErapidityfig}, distinguishes the two types of edges with the associated Boltzmann weights pictured in Figure \ref{2boltzmannweights}.  This crossing of directed rapidity lines also distinguishes the ordering of spins that appear in the arguments of the Boltzmann weights $W_{pq}(\spn_i,\spn_j)$, $\oW_{pq}(\spn_i,\spn_j)$.  Such edges and their Boltzmann weights are the building blocks of an integrable lattice model of statistical mechanics, where the Boltzmann weights characterise the interaction energies between nearest neighbour spins.  The star-triangle relation \eqref{YBE} is the condition of integrability for such a lattice model, from which one may find a set of commuting transfer matrices that may be used to solve for the partition function per site in the thermodynamic limit \cite{Baxter:1982zz}.

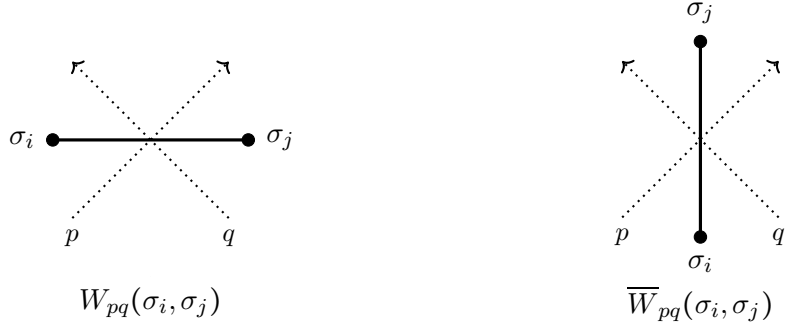
\begin{figure}[hbt]
\centering
\begin{tikzpicture}[scale=2.6]

\draw[-,very thick] (-0.5,2)--(0.5,2);
\draw[->,thick,dotted] (0.4,1.6)--(-0.4,2.4);
\fill[white!] (0.4,1.6) circle (0.01pt)
node[below=0.5pt]{\color{black}\small $q$};
\draw[->,thick,dotted] (-0.4,1.6)--(0.4,2.4);
\fill[white!] (-0.4,1.6) circle (0.01pt)
node[below=0.5pt]{\color{black}\small $p$};
\filldraw[fill=black,draw=black] (-0.5,2) circle (0.9pt)
node[left=3pt]{\color{black} $\spn_i$};
\filldraw[fill=black,draw=black] (0.5,2) circle (0.9pt)
node[right=3pt]{\color{black} $\spn_j$};

\fill (0,1.3) circle(0.01pt)
node[below=0.05pt]{\color{black} $ W_{pq}(\spn_i,\spn_j)$};

\begin{scope}[xshift=80pt,yshift=57pt]
\draw[-,very thick] (0,-0.5)--(0,0.5);
\draw[->,thick,dotted] (-0.4,-0.4)--(0.4,0.4);
\fill[white!] (-0.4,-0.4) circle (0.01pt)
node[below=0.5pt]{\color{black}\small $p$};
\draw[->,thick,dotted] (0.4,-0.4)--(-0.4,0.4);
\fill[white!] (0.4,-0.4) circle (0.01pt)
node[below=0.5pt]{\color{black}\small $q$};
\filldraw[fill=black,draw=black] (0,-0.5) circle (0.9pt)
node[below=3pt]{\color{black} $\spn_i$};
\filldraw[fill=black,draw=black] (0,0.5) circle (0.9pt)
node[above=3pt]{\color{black} $\spn_j$};

\fill (0,-0.7) circle(0.01pt)
node[below=0.05pt]{\color{black} $\oW_{pq}(\spn_i,\spn_j)$};
\end{scope}
\end{tikzpicture}
\caption{The first type of edge on the left, is associated with a Boltzmann weight $W_{pq}(\spn_i,\spn_j)$, and the second type of edge on the right, is associated with the Boltzmann weight $\oW_{pq}(\spn_i,\spn_j)$.}
\label{2boltzmannweights}
\end{figure}

\subsubsection{Symmetric form (quantum Q-type equation)}

For the cases considered in this paper, some symmetries may be used to simplify the form of the star-triangle relation \eqref{YBE}.  First, the Boltzmann weights $W_{pq}(\spn_i,\spn_j)$, $\oW_{pq}(\spn_i,\spn_j)$, only depend on the difference of rapidity variables $p-q$, consequently the star-triangle relation will be written in terms of the redefined Boltzmann weights $W(\theta\,|\,\spn_i,\spn_j)$, $\oW(\theta\,|\,\spn_i,\spn_j)$, as
\begin{align}
W(\theta\,|\,\spn_i,\spn_j):=W_{pq}(\spn_i,\spn_j)\,,\qquad \oW(\theta\,|\,\spn_i,\spn_j):=\oW_{pq}(\spn_i,\spn_j)\,,
\end{align}
where $\theta=p-q$, is a rapidity difference (spectral) variable.

Then setting
\begin{align}
\label{spectraldef}
\theta_1=q-r\,,\qquad \theta_3=p-q\,,
\end{align}
it follows that the star-triangle relation \eqref{YBE} may be written in the following equivalent form
\begin{align}
\label{YBEsym}
\begin{split}
\ds\int_{\mathbb{R}} d\spn_0\,\s(\spn_0)\,\oW(\theta_1\,|\,\spn_1,\spn_0)\,W(\theta_1+\theta_3\,|\,\spn_2,\spn_0)\,\oW(\theta_3\,|\,\spn_3,\spn_0)\phantom{\,,} \\
\ds=R(\theta_1,\theta_3)\,W(\theta_1\,|\,\spn_3,\spn_2)\,\oW(\theta_1+\theta_3\,|\,\spn_1,\spn_3)\,W(\theta_3\,|\,\spn_1,\spn_2)\,.
\end{split}
\end{align}
The star-triangle relation \eqref{YBEsym} depends on two independent, real-valued, rapidity difference variables $\theta_1$, $\theta_3$ (as well as the same three independent spin variables $\spn_1,\spn_2,\spn_3$).  In terms of the variables \eqref{spectraldef}, a graphical representation of the star-triangle relation \eqref{YBEsym} is shown in Figure \ref{YBEfig}.

\begin{figure}[h]
\centering

\begin{tikzpicture}[scale=1.9]

\draw[white] (-2.0,0.5) circle (0.1pt)
node[right=2pt]{\color{black}\small $\oW(\theta_1\,|\,\spn_1,\spn_0)$};
\draw[white] (-2.95,-0.3) circle (0.1pt)
node[above=2pt]{\color{black}\small $\oW(\theta_3\,|\,\spn_0,\spn_3)$};
\draw[white] (-0.85,-0.3) circle (0.1pt)
node[above=2pt]{\color{black}\small $W(\theta_1+\theta_3\,|\,\spn_2,\spn_0)$};
\draw[-,thick] (-2,0)--(-2,1);
\draw[-,thick] (-2,0)--(-2.87,-0.5);
\draw[-,thick] (-2,0)--(-1.13,-0.5);
\fill (-2,0) circle (1.5pt)
node[below=1.5pt]{\color{black} $\spn_0$};
\filldraw[fill=black,draw=black] (-2,1) circle (1.5pt)
node[above=1.5pt] {\color{black} $\spn_1$};
\filldraw[fill=black,draw=black] (-2.87,-0.5) circle (1.5pt)
node[below=1.5pt] {\color{black} $\spn_3$};
\filldraw[fill=black,draw=black] (-1.13,-0.5) circle (1.5pt)
node[below=1.5pt] {\color{black} $\spn_2$};

\fill[white!] (0.2,0.3) circle (0.01pt)
node[left=0.05pt] {\color{black}\Large $=$};

\begin{scope}[xshift=10pt]
\draw[white] (2.0,-0.5) circle (0.1pt)
node[below=2pt]{\color{black}\small $W(\theta_1\,|\,\spn_2,\spn_3)$};
\draw[white] (3.0,0.30) circle (0.1pt)
node[above=2pt]{\color{black}\small $W(\theta_3\,|\,\spn_1,\spn_3)$};
\draw[white] (0.80,0.30) circle (0.1pt)
node[above=2pt]{\color{black}\small $\oW(\theta_1+\theta_3\,|\,\spn_1,\spn_3)$};
\draw[-,thick] (2,1)--(1.13,-0.5);
\draw[-,thick] (1.13,-0.5)--(2.87,-0.5);
\draw[-,thick] (2.87,-0.5)--(2,1);
\filldraw[fill=black,draw=black] (2,1) circle (1.5pt)
node[above=1.5pt]{\color{black} $\spn_1$};
\filldraw[fill=black,draw=black] (1.13,-0.5) circle (1.5pt)
node[left=1.5pt]{\color{black} $\spn_3$};
\filldraw[fill=black,draw=black] (2.87,-0.5) circle (1.5pt)
node[right=1.5pt]{\color{black} $\spn_2$};
\end{scope}
\end{tikzpicture}
\caption{Graphical representation of the star-triangle relation \eqref{YBEsym}.}
\label{YBEfig}
\end{figure}
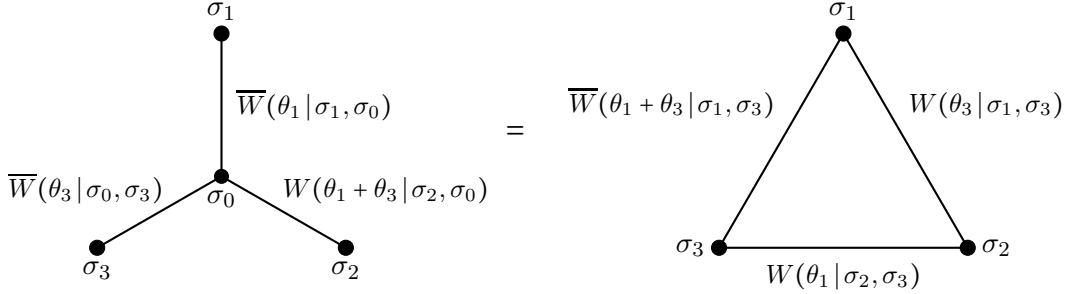

The star-triangle relation \eqref{YBEsym}, is one of two main forms that will be referred to throughout this paper, and is a typical form that is associated with the integrable lattice models of statistical mechanics \cite{Baxter:1982zz}.   It is also worth noting that there exists another symmetry known as the crossing symmetry, for which $W(\theta\,|\,\spn_i,\spn_j)$, and $\oW(\theta\,|\,\spn_i,\spn_j)$, are related by
\begin{align}
\label{crossingsym}
\oW(\theta\,|\,\spn_i,\spn_j)=W(\eta-\theta\,|\,\spn_i,\spn_j)\,,
\end{align}
where $\eta$ is known as the crossing parameter.  It follows that the star-triangle relation \eqref{YBEsym}, where \eqref{crossingsym} is satisfied, can be written in the more symmetric form
\begin{align}
\label{YBEtotalsym}
\begin{split}
\ds\int_{\mathbb{R}} d\spn_0\,\s(\spn_0)\,\oW(\theta_1\,|\,\spn_1,\spn_0)\,\oW(\theta_2\,|\,\spn_2,\spn_0)\,\oW(\theta_3\,|\,\spn_0,\spn_3)\phantom{\qquad\,,} \\
\ds=R(\theta_1,\theta_2,\theta_3)\,W(\theta_1\,|\,\spn_2,\spn_3)\,W(\theta_2\,|\,\spn_1,\spn_3)\,W(\theta_3\,|\,\spn_2,\spn_1)\,,
\end{split}
\end{align}
with the condition that $\eta=\theta_1+\theta_2+\theta_3$.  The crossing symmetry \eqref{crossingsym} (along with the spin reflection symmetry \eqref{spinsymmetry}) is satisfied for the majority of star-triangle relations of the form \eqref{YBEsym}, with the exception of the rational cases in Section \ref{sec:ratlim} (corresponding to integrable quad equations $Q2$ and $Q1_{(\delta=1)}$).    The star-triangle relation of the form \eqref{YBEsym} was previously shown \cite{Bazhanov:2016ajm} to correspond to quantum counterparts of the $Q$-type equations in the ABS list \cite{ABS}.

\subsubsection{Asymmetric form (quantum $H$-type equation)}

The second form of the star-triangle relation considered in this paper (introduced in \eqref{YBEintro}), takes the form
\begin{align}
\label{YBEasym}
\begin{split}
\ds\int_{\mathbb{R}} d\spn_0\,\s(\spn_0)\,\oV(\theta_1\,|\,\spn_1,\spn_0)\,V(\theta_1+\theta_3\,|\,\spn_2,\spn_0)\,\oW(\theta_3\,|\,\spn_0,\spn_3)\phantom{\,,} \\
\ds=R(\theta_1,\theta_3)\,V(\theta_1\,|\,\spn_2,\spn_3)\,\oV(\theta_1+\theta_3\,|\,\spn_1,\spn_3)\,W(\theta_3\,|\,\spn_2,\spn_1)\,. 
\end{split}
\end{align}
This form of the star-triangle relation will be seen to be associated to the asymmetric $H$-type ABS equations \cite{ABS2}.  The star-triangle relations of the form \eqref{YBEasym}, may be obtained as certain asymmetric limits of the star-triangle relation \eqref{YBEsym} (or \eqref{YBEtotalsym}), and examples of such limits will be proven for each example considered in this paper.  In all cases, the Boltzmann weights $W(\theta\,|\,\spn_i,\spn_j)$, $\oW(\theta\,|\,\spn_i,\spn_j)$ in \eqref{YBEasym}, are equivalent to a Boltzmann weight coming from a star-triangle relation of the form \eqref{YBEsym}.\footnote{The classical analogue of this property, is that the three-leg form of an $H$-type ABS equation, always has a ``leg'' function in common with a $Q$-type ABS equation.}  The star-triangle relation \eqref{YBEasym} has the graphical representation given in Figure \ref{YBEasymfig}, where to contrast with Figure \ref{YBEfig}, the Boltzmann weights $V(\theta\,|\,\spn_i,\spn_j)$, $\oV(\theta\,|\,\spn_i,\spn_j)$, are associated to double edges that connect two vertices.

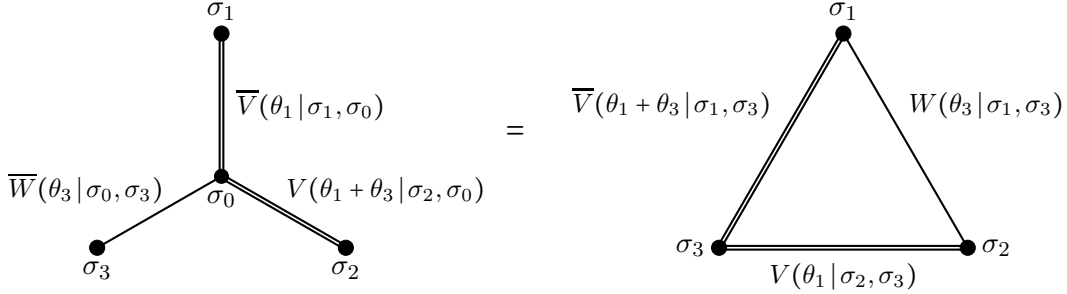
\begin{figure}[h]
\centering

\begin{tikzpicture}[scale=1.9]

\draw[white] (-2.0,0.5) circle (0.1pt)
node[right=2pt]{\color{black}\small $\oV(\theta_1\,|\,\spn_1,\spn_0)$};
\draw[white] (-2.95,-0.3) circle (0.1pt)
node[above=2pt]{\color{black}\small $\oW(\theta_3\,|\,\spn_0,\spn_3)$};
\draw[white] (-0.85,-0.3) circle (0.1pt)
node[above=2pt]{\color{black}\small $V(\theta_1+\theta_3\,|\,\spn_2,\spn_0)$};
\draw[-,double,thick] (-2,0)--(-2,1);
\draw[-,thick] (-2,0)--(-2.87,-0.5);
\draw[-,double,thick] (-2,0)--(-1.13,-0.5);
\fill (-2,0) circle (1.5pt)
node[below=1.5pt]{\color{black} $\spn_0$};
\filldraw[fill=black,draw=black] (-2,1) circle (1.5pt)
node[above=1.5pt] {\color{black} $\spn_1$};
\filldraw[fill=black,draw=black] (-2.87,-0.5) circle (1.5pt)
node[below=1.5pt] {\color{black} $\spn_3$};
\filldraw[fill=black,draw=black] (-1.13,-0.5) circle (1.5pt)
node[below=1.5pt] {\color{black} $\spn_2$};

\fill[white!] (0.2,0.3) circle (0.01pt)
node[left=0.05pt] {\color{black}\Large $=$};

\begin{scope}[xshift=10pt]
\draw[white] (2.0,-0.5) circle (0.1pt)
node[below=2pt]{\color{black}\small $V(\theta_1\,|\,\spn_2,\spn_3)$};
\draw[white] (3.0,0.30) circle (0.1pt)
node[above=2pt]{\color{black}\small $W(\theta_3\,|\,\spn_1,\spn_3)$};
\draw[white] (0.80,0.30) circle (0.1pt)
node[above=2pt]{\color{black}\small $\oV(\theta_1+\theta_3\,|\,\spn_1,\spn_3)$};
\draw[-,double,thick] (2,1)--(1.13,-0.5);
\draw[-,double,thick] (1.13,-0.5)--(2.87,-0.5);
\draw[-,thick] (2.87,-0.5)--(2,1);
\filldraw[fill=black,draw=black] (2,1) circle (1.5pt)
node[above=1.5pt]{\color{black} $\spn_1$};
\filldraw[fill=black,draw=black] (1.13,-0.5) circle (1.5pt)
node[left=1.5pt]{\color{black} $\spn_3$};
\filldraw[fill=black,draw=black] (2.87,-0.5) circle (1.5pt)
node[right=1.5pt]{\color{black} $\spn_2$};
\end{scope}
\end{tikzpicture}
\caption{Graphical representation of the asymmetric form of the star-triangle relation \eqref{YBEasym}.}
\label{YBEasymfig}
\end{figure}

\subsection{3D-consistent equations from the quasi-classical expansion}

\subsubsection{Classical $Q$-type equations}

Recall the symmetric star-triangle relation given in \eqref{YBEsym}.  Let $f_{i,\hbar}(x_i)$, $i=0,1,2,3$, and $g_{j,\hbar}(\alpha)$, $j=1,3$, respectively denote suitable scaling and translations of the variables $x_i$, $\theta_i$, which depend on a parameter $\hbar>0$, such that with a change of variables
\begin{align}
\label{covqcl}
\spn_i= f_{i,\hbar}(x_i)\,,\quad i=0,1,2,3\,,\qquad \theta_j=g_{j,\hbar}(\alpha_j)\,,\quad j=1,3\,,
\end{align}
the leading asymptotics of the left hand side of the star-triangle relation \eqref{YBEsym} takes the form
\begin{align}
\label{qclybe}
\int \frac{dx_0}{\sqrt{\hbar}} \exp\left(\hbar^{-1}(\mathcal{A}_\medstar(x_0,x_1,x_2,x_3;\alpha_1,\alpha_3)+O(1)\right)\,,
\end{align}
where
\begin{align}
\label{adefsym}
\mathcal{A}_\medstar(x_0,x_1,x_2,x_3;\alpha_1,\alpha_3)&=C(x_0)+\ol(\alpha\,|\,x_1,x_0)+\lag(\alpha+\beta\,|\,x_2,x_0)+\ol(\alpha_3\,|\,x_3,x_0).
\end{align}
Up to possibly some factors which have no overall contribution to the star-triangle relation \eqref{YBEsym}, the functions $\lag(\alpha\,|\,x_i,x_j)$, and $\ol(\alpha\,|\,x_i,x_j)$, correspond to the $O(\hbar^{-1})$ asymptotics of $W(\theta\,|\,\spn_i,\spn_j)$, and $\oW(\theta\,|\,\spn_i,\spn_j)$, respectively, while $C(x)$ corresponds to the $O(\hbar^{-1})$ asymptotics of the factor $S(\sigma)$.

The saddle point equation of \eqref{qclybe} is given by 
\begin{align}
\label{3legsym}
\left.\frac{\partial}{\partial x}\mathcal{A}_\medstar(x,x_1,x_2,x_3;\alpha_1,\alpha_3)\right|_{x=x_0}=0\,.
\end{align}
Then for example, by defining 
\begin{align}
\frac{\partial C(x_j)}{\partial x_j}+\frac{\partial\mathcal{L}(\alpha\,|\,x_i,x_j)}{\partial x_j}=\varphi(\alpha\,|\,x_i,x_j)\,,\qquad \frac{\partial\ol(\alpha\,|\,x_i,x_j)}{\partial x_j}=\overline{\varphi}(\alpha\,|\,x_i,x_j)\,,
\end{align}
the saddle point equation \eqref{3legsym} takes a typical ``three-leg'' form 
\begin{align}
\label{3legsym2}
\ds\left.\frac{\partial}{\partial x}\mathcal{A}_\medstar(x,x_1,x_2,x_3;\alpha_1,\alpha_3)\right|_{x=x_0}=\overline{\varphi}(\alpha_1\,|\,x_1,x_0)+\varphi(\alpha_1+\alpha_3\,|\,x_2,x_0)+\overline{\varphi}(\alpha_3\,|\,x_0,x_3)=0 \,.
\end{align}
Note that here the derivative $\partial\,\iC(x)/\partial x$, is absorbed into the definition of $\varphi(x,y;\alpha)$, however it may also in principle be absorbed into $\ovphi(x,y;\alpha)$, or a combination of both.  It is only the overall expression given in \eqref{3legsym2} that matters.  Defining $\alpha_2=\alpha_1+\alpha_3$, the functions $\varphi(\alpha\,|\,x_i,x_j)$, and $\overline{\varphi}(\alpha\,|\,x_i,x_j)$ of the saddle point equation \eqref{3legsym}, may be associated to edges and vertices of a quadrilateral, as depicted in Figure \ref{3legfigq}.

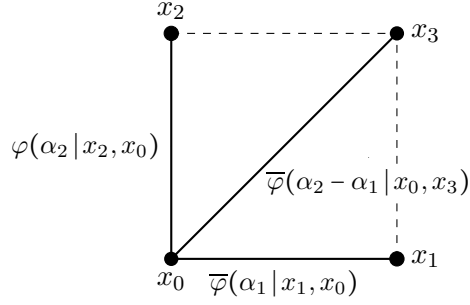
\begin{figure}[tbh]
\centering
\begin{tikzpicture}[scale=1.5]
\draw[thick,-] (3,-1)--(3,1);
\draw[thick,-] (3,-1)--(5,1);
\draw[thick,-] (5,-1)--(3,-1);
\draw[-,dashed] (3,1)--(5,1);
\draw[-,dashed] (5,-1)--(5,1);
\fill (3,0) circle (0.1pt)
node[left=0.5pt]{\color{black}\small $\varphi(\alpha_2\,|\,x_2,x_0)$};
\fill (4,-1) circle (0.1pt)
node[below=0.5pt]{\color{black}\small $\overline{\varphi}(\alpha_1\,|\,x_1,x_0)$};
\fill [white] (4.9,-0.5) rectangle (5.1,-0.2);
\fill (4.75,-0.1) circle (0.1pt)
node[below=0.5pt]{\color{black}\small $\overline{\varphi}(\alpha_2-\alpha_1\,|\,x_0,x_3)$};
\fill (3,-1) circle (1.8pt)
node[below=1.5pt]{\color{black} $x_0$};
\filldraw[fill=black,draw=black] (3,1) circle (1.8pt)
node[above=1.5pt]{\color{black} $x_2$};
\fill (5,1) circle (1.8pt)
node[right=1.5pt]{\color{black} $x_3$};
\filldraw[fill=black,draw=black] (5,-1) circle (1.8pt)
node[right=1.5pt]{\color{black} $x_1$};
\end{tikzpicture}
\caption{Q-type three-leg equation \eqref{3legsym} arising as the saddle point equation for the star-triangle relation \eqref{YBEsym}.}
\label{3legfigq}
\end{figure}

For all solutions of the star-triangle relation \eqref{YBEsym} considered in this paper, there exists a change of variables of the form
\begin{align}
\label{ptrans}
x=h_0(x_0),\quad u=h_1(x_1),\quad y=h_2(x_2),\quad v=h_3(x_3),\quad \alpha=h_a(\alpha_1,\alpha_3),\quad \beta=h_b(\alpha_1,\alpha_3)\,,
\end{align}
such that the three-leg saddle point equation \eqref{3legsym} may be seen to correspond to the equation
\begin{align}
\label{quadequation}
Q(x,u,y,v;\alpha,\beta)=0\,,
\end{align}
where the function $Q(x,u,y,v;\alpha,\beta)$  is a polynomial that is linear in each of the variables $x,u,y,v$ respectively (known as the affine-linear property).  For the algebraic case that is considered in Section \ref{sec:alglim}, the equations \eqref{3legsym} and \eqref{quadequation} are equivalent, up to the corresponding change of variables of the form \eqref{ptrans}.  For each of the remaining cases, the correspondence between \eqref{3legsym} and \eqref{quadequation} is obtained by first writing \eqref{3legsym} in the form
\begin{align}
\label{mult3leg}
\left.\frac{\partial}{\partial x}\mathcal{A}_\medstar(x,x_1,x_2,x_3;\alpha_1,\alpha_3)\right|_{x=x_0}\hspace{-0.4cm}=\Log\left(\overline{\Phi}(\alpha_1\,|\,x_1,x_0)\,\Phi(\alpha_1+\alpha_3\,|\,x_2,x_0)\,\overline{\Phi}(\alpha_3\,|\,x_3,x_0)\right) +2\pi\ii k\,,
\end{align}
for some $k\in\mathbb{Z}$, where $\Log(z)$ denotes the principal branch of the complex logarithm.  Then it may be seen that up to the change of variables of the form \eqref{ptrans}, the equation $\overline{\Phi}(\alpha_1\,|\,x_1,x_0)\,\Phi(\alpha_1+\alpha_3\,|\,x_2,x_0)\,\overline{\Phi}(\alpha_3\,|\,x_3,x_0)=1$ is equivalent to the quad equation \eqref{quadequation}.  Note that for such cases, \eqref{3legsym}, and \eqref{quadequation}, are obviously only equivalent in connected subspaces of the 6-dimensional complex parameter space of the variables $x_0,x_1,x_2,x_3,\alpha_1,\alpha_3$, where $k=0$ in \eqref{mult3leg}.  

An important notion of integrability for equations of the type \eqref{quadequation}, is known as 3D-consistency (also known as consistency-around-a-cube, or multi-dimensional consistency).  The property of 3D-consistency may be stated as follows:

{\bf 3D-consistency condition.} \cite{nijhoffwalker,BobSurQuadGraphs,ABS} Associate the six respective equations
\begin{align}
\label{3dequations}
\begin{gathered}
Q(x_0,x_1,x_2,x_{12},\alpha,\beta)=0\,, \\
Q(x_0,x_1,x_3,x_{13},\alpha,\gamma)=0\,, \\
Q(x_0,x_2,x_3,x_{23},\beta,\gamma)=0\,, \\
Q(x_{13},x_3,x_{123},x_{23},\alpha,\beta)=0\,, \\
Q(x_{12},x_2,x_{123},x_{23},\alpha,\gamma)=0\,, \\
Q(x_{12},x_1,x_{123},x_{13},\beta,\gamma)=0\,, 
\end{gathered}
\end{align}
to the corresponding six faces of a cube, as indicated in Figure \ref{3D}.  Edges that are parallel in Figure \ref{3D}, are always associated with the same variable $\alpha$, $\beta$, or $\gamma$.  Then consider the initial value problem, where $x_0,x_1,x_2,x_3$ and $\alpha,\beta,\gamma$, are known, and $x_{12},x_{13},x_{23},x_{123}$, are to be determined.  Due to the affine-linear property of \eqref{quadequation}, the first three equations in \eqref{3dequations} may be solved uniquely for the variables $x_{12}$, $x_{13}$, $x_{23}$, respectively.  The 3D-consistency condition is that the remaining three equations in \eqref{3dequations}, each must agree for the solution of the remaining variable $x_{123}$.

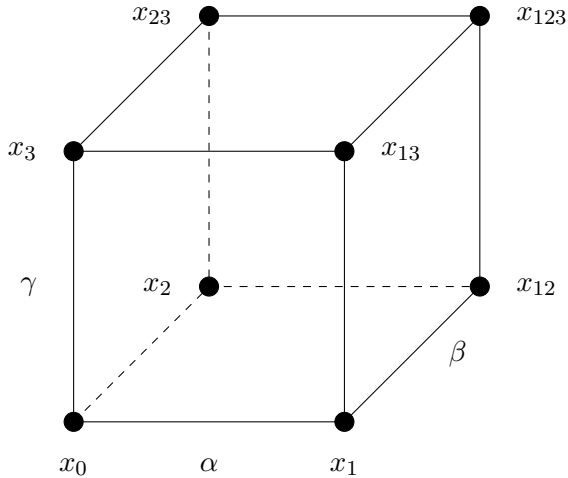
\begin{figure}[htb]
\centering
\begin{tikzpicture}[scale=1.8]
\draw[white!] (-1,0) circle (0.1pt)
node[left=10pt]{\color{black} $\gamma$};
\draw[white!] (0,-1) circle (0.1pt)
node[below=10pt]{\color{black} $\alpha$};
\draw[white!] (1.5,-0.5) circle (0.1pt)
node[right=10pt]{\color{black} $\beta$};
\draw[-] (1,1)--(-1,1)--(-1,-1)--(1,-1)--(1,1);
\draw[-] (2,-0)--(2,2)--(-0,2);
\draw[-,dashed] (-0,2)--(-0,-0)--(2,-0);
\draw[-] (-1,1)--(0,2);
\draw[-] (1,1)--(2,2);
\draw[-] (1,-1)--(2,0);
\draw[-,dashed] (-1,-1)--(0,0);
\filldraw[fill=black,draw=black] (2,2) circle (2pt)
node[right=10pt]{$x_{123}$};
\filldraw[fill=black,draw=black] (0,0) circle (2pt)
node[left=10pt]{$x_{2}$};
\filldraw[fill=black,draw=black] (2,0) circle (2pt)
node[right=10pt]{$x_{12}$};
\filldraw[fill=black,draw=black] (0,2) circle (2pt)
node[left=10pt]{$x_{23}$};
\filldraw[fill=black,draw=black] (1,1) circle (2pt)
node[right=10pt]{$x_{13}$};
\filldraw[fill=black,draw=black] (-1,-1) circle (2pt)
node[below=10pt]{$x_0$};
\filldraw[fill=black,draw=black] (1,-1) circle (2pt)
node[below=10pt]{$x_1$};
\filldraw[fill=black,draw=black] (-1,1) circle (2pt)
node[left=10pt]{$x_3$};
\end{tikzpicture}

\caption{3D-consistency, also known as ``consistency-around-a-cube'', or ``multi-dimensional consistency''.}
\label{3D}
\end{figure}

Remarkably, for all cases considered in this paper, the quad equation \eqref{quadequation} derived from the star-triangle relation \eqref{YBEsym} satisfies 3D-consistency exactly as described above, and corresponds to an equation in the ABS classification \cite{ABS}.  The resulting correspondence between the solutions of the star-triangle relation \eqref{YBEsym}, and the symmetric 3D-consistent ABS equations of $Q$-type, is given by the following Theorem \cite{Bazhanov:2007mh,Bazhanov:2010kz,Bazhanov:2016ajm}:

\begin{theorem}
\label{Qtheorem}
There exists a respective solution $W(\theta\,|\,\spn_i,\spn_j)$, $\oW(\theta\,|\,\spn_i,\spn_j)$, of the star-triangle relation \eqref{YBEsym} for each of the 3D-consistent $Q$-type quad equations \eqref{quadequation} of the ABS list, which has a quasi-classical expansion that can be written in the form \eqref{qclybe}, through which the corresponding $Q$-type quad equation \eqref{quadequation} is obtained from the saddle point three-leg equation \eqref{3legsym} with a change of variables of the form \eqref{ptrans}.
\end{theorem}

The Yang-Baxter/3D-consistency correspondence for $Q$-type equations given by Theorem \ref{Qtheorem} (and which also appear in this paper) is summarised in Table \ref{Qtable}  (note also the interpretation of the star-triangle relations in terms of hypergeometric integrals which is summarised in Table \ref{hyptable}).

\begin{table}[htbp]
\begin{center}
\medskip
\begin{tabular}{c|c|c}
  Boltzmann Weights & Three-Leg Equation & Quad Equation\\
 \hline
 \hline
 \eqref{genfadvol}  & \eqref{Q3d13leg} & $Q3_{(\delta=1)}\phantom{1}$ \eqref{Q3d1} \\ 
   \hline
  \eqref{fadvol} & \eqref{Q3d03leg} & $Q3_{(\delta=0)}\phantom{1}$ \eqref{Q3d0} \\ 
   \hline
  \eqref{q2quan} & \eqref{Q23leg} & $Q2_{\phantom{(\delta=1)}}$ \eqref{Q2} \\ 
 \hline
  \eqref{q1d1quan}  & \eqref{Q1d13leg} & $Q1_{(\delta=1)}$ \eqref{Q1d1} \\ 
   \hline
  \eqref{quanq1d0}  & \eqref{Q1d03leg} & $Q1_{(\delta=0)}$ \eqref{Q1d0}
\end{tabular}
\caption{$Q$-type ABS quad equations \eqref{quadequation} derived from the star-triangle relation \eqref{YBEsym}.}
\label{Qtable}
\end{center}
\end{table}

The Yang-Baxter/3D-consistency correspondence summarised in Theorem \ref{Qtheorem} and Table \ref{Qtable} has appeared previously in other papers {\cite{Bazhanov:2007mh,Bazhanov:2010kz,Bazhanov:2016ajm}}.  These equations will be considered again here, mainly because the star-triangle relations for the $H$-type equations (discussed in the following subsection) will be seen to arise from certain limits of the star-triangle relations \eqref{YBEsym} for the $Q$-type equations. 

Since the main motivation of this paper is to consider the Yang-Baxter/3D-consistency correspondence for the $H$-type equations, the elliptic case will not be considered here, for which there is only the one $Q$-type equation, $Q4$.  The star-triangle relation for $Q4$ also slightly differs from the exact form of \eqref{YBEsym}, in that the integration and associated spin variables are not taken in $\mathbb{R}$, but rather are restricted to a compact subset of $\mathbb{R}$.  The equation $Q4$ has previously been derived from various different solutions of the star-triangle relation, as summarised in Figure \ref{ellipticfig}.  Figure \ref{ellipticfig} illustrates that a single quad equation may correspond to several different forms of the star-triangle relation, which are each associated with distinct integrable lattice models.  In other words the quantization of an integrable quad equation through the Yang-Baxter/3D-consistency correspondence is not unique.

\begin{figure}[tbh]
\centering
\begin{tikzpicture}[scale=1]
\draw[white!] (-7,3) circle (0.01pt)
node[above=1pt]{\color{black} Quantum};
\draw[white!] (7,3) circle (0.01pt);
\draw[white!] (-7,0.25) circle (0.01pt)
node[above=1pt]{\color{black} Classical};
\draw[white!] (7,0.25) circle (0.01pt);

\draw[white!] (2.0,1.7) circle (0.01pt)
node[right=1pt]{\color{black}\small \cite{Bazhanov:2016ajm} $\; N\to\infty$};
\draw[white!] (-2.0,1.7) circle (0.01pt)
node[left=1pt]{\color{black}\small  \cite{Kels:2017vbc} $\q\to\EXP^{\pi\ii/rN}$};
\draw[white!] (0,3) circle (0.01pt)
node[above=1pt]{\color{black} MS}
node[above=30pt]{\color{black}\small  \cite{Kels:2017vbc} $\;\q\to\EXP^{\pi\ii/rN}$ }
node[below=20pt]{\color{black}\small \cite{Bazhanov:2010kz} $\q\to\EXP^{\pi\ii/N}$ };
\draw[white!] (1.8,3.6) circle (0.01pt)
node[below=5pt]{\color{black}\small \cite{Bazhanov:2010kz} $\;\q\to\EXP^{\pi\ii/N}$};
\draw[white!] (-1.8,3.6) circle (0.01pt)
node[below=5pt]{\color{black}\small $r=1$};
\begin{scope}[>=Latex]
\draw[thick,->>] (0.6,3.4)--(3.4,3.4);
\end{scope}
\draw[thick,<-] (-0.6,3.4)--(-3.3,3.4);
\begin{scope}[>=Latex]
\draw[thick,->>] (-3.6,3.6) .. controls (-1.5,4.2) and (1.5,4.2) .. (3.6,3.6);
\end{scope}
\draw[white!] (-4,3) circle (0.01pt)
node[above=1pt]{\color{black} LEGF};
\draw[white!] (4,3) circle (0.01pt)
node[above=1pt]{\color{black} KM};
\draw[thick,-] (0,3)--(0,2.25);
\begin{scope}[>=Latex]
\draw[thick,->] (0,1.75)--(0,1);
\end{scope}
\draw[white!] (0,0.25) circle (0.01pt)
node[above=1pt]{\color{black} $Q4$};
\begin{scope}[>=Latex]
\draw[thick,->] (-4,3)--(-0.5,1);
\end{scope}
\begin{scope}[>=Latex]
\draw[thick,->] (4,3)--(0.5,1);
\end{scope}

\end{tikzpicture}
\caption{Star-triangle relations that reduce to $Q4$ in a quasi-classical limit. Here LEGF, MS, KM, stand for lens elliptic gamma function solution \cite{Kels:2015bda}, master solution \cite{Bazhanov:2010kz}, Kashiwara-Miwa solution \cite{Kashiwara:1986tu}, of star-triangle relations respectively.  Also filled single, and double arrow heads, respectively represent leading ($O(\hbar^{-1})$) and subleading ($O(1)$) order quasi-classical expansions.}
\label{ellipticfig}
\end{figure}
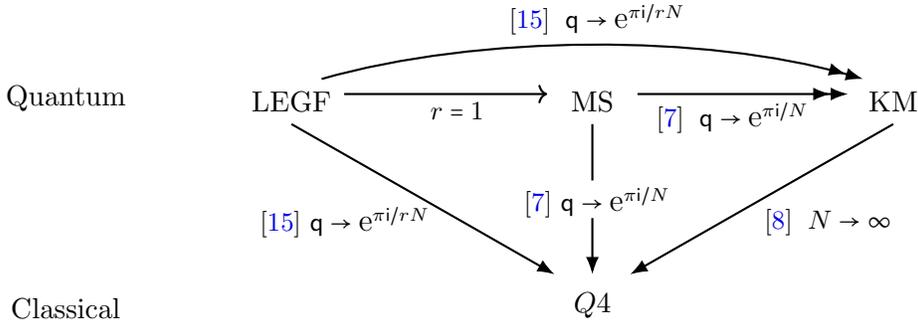

\subsubsection{Classical $H$-type equations}

Next consider the asymmetric form of the star-triangle relation given in \eqref{YBEasym}.  The majority of the previous discussion for \eqref{YBEsym} also holds for \eqref{YBEasym}, with only some small differences.

The quasi-classical expansion \eqref{covqcl} of the star-triangle relation \eqref{YBEasym}, has the same form given in \eqref{qclybe}, but with \eqref{adefsym} replaced with 
\begin{align}
\label{adefasym}
\mathcal{A}_\medstar(x_0,x_1,x_2,x_3;\alpha_1,\alpha_3)&=C(x_0)+\olam(\alpha\,|\,x_1,x_0)+\Lambda(\alpha+\beta\,|\,x_2,x_0)+\ol(\alpha_3\,|\,x_0,x_3)\,.
\end{align}

The functions $\Lambda(\alpha\,|\,x_i,x_j)$, $\olam(\alpha\,|\,x_i,x_j)$, $\ol(\alpha\,|\,x_i,x_j)$, correspond to the $O(\hbar^{-1})$ asymptotics of $V(\theta\,|\,\spn_i,\spn_j)$, $\oV(\theta\,|\,\spn_i,\spn_j)$, $\oW(\theta\,|\,\spn_i,\spn_j)$, respectively (up to factors which have no overall contribution to the star-triangle relation \eqref{YBEasym}), while $C(x)$ corresponds to the $O(\hbar^{-1})$ asymptotics of the factor $S(\spn)$.

Similarly to \eqref{3legsym2}, the saddle point equation \eqref{3legsym} for \eqref{adefasym}, may be written in a typical three-leg form 
\begin{align}
\label{3legasym}
\left.\ds\frac{\partial}{\partial x}\mathcal{A}_\medstar(x,x_1,x_2,x_3;\alpha_1,\alpha_3)\right|_{x=x_0}=\ovphib(\alpha_1\,|\,x_0,x_1)+\phi(\alpha_1+\alpha_3\,|\,x_0,x_2)+\ovphi(\alpha_3\,|\,x_0,x_3)=0 \,,
\end{align}
where $\phi(\alpha\,|\,x_i,x_j)$, $\ovphib(\alpha\,|\,x_i,x_j)$, $\ovphi(\alpha\,|\,x_i,x_j)$, are appropriately defined derivatives of $\Lambda(\alpha\,|\,x_i,x_j)$, $\olam(\alpha\,|\,x_i,x_j)$, $\ol(\alpha\,|\,x_i,x_j)$, and $C(x)$.  Defining $\alpha_2=\alpha_1+\alpha_3$, the functions $\phi(\alpha\,|\,x_i,x_j)$, $\ovphib(\alpha\,|\,x_i,x_j)$, and $\varphi(\alpha\,|\,x_i,x_j)$  in the saddle point equation \eqref{3legasym}, may be associated to edges and vertices of a quadrilateral, as depicted in Figure \ref{3legfigh}.

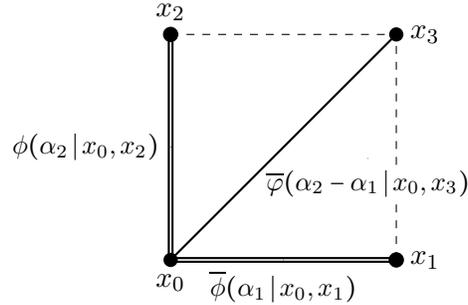
\begin{figure}[tbh]
\centering
\begin{tikzpicture}[scale=1.5]
\draw[double,thick,-] (3,-1)--(3,1);
\draw[-,thick] (3,-1)--(5,1);
\draw[double,thick,-] (5,-1)--(3,-1);
\draw[-,dashed] (3,1)--(5,1);
\draw[-,dashed] (5,-1)--(5,1);
\fill (3,0) circle (0.1pt)
node[left=0.5pt]{\color{black}\small $\phi(\alpha_2\,|\,x_0,x_2)$};
\fill (4,-1) circle (0.1pt)
node[below=0.5pt]{\color{black}\small $\overline{\phi}(\alpha_1\,|\,x_0,x_1)$};
\fill [white] (4.9,-0.5) rectangle (5.1,-0.2);
\fill (4.75,-0.1) circle (0.1pt)
node[below=0.5pt]{\color{black}\small $\overline{\varphi}(\alpha_2-\alpha_1\,|\,x_0,x_3)$};
\fill (3,-1) circle (1.8pt)
node[below=1.5pt]{\color{black} $x_0$};
\filldraw[fill=black,draw=black] (3,1) circle (1.8pt)
node[above=1.5pt]{\color{black} $x_2$};
\fill (5,1) circle (1.8pt)
node[right=1.5pt]{\color{black} $x_3$};
\filldraw[fill=black,draw=black] (5,-1) circle (1.8pt)
node[right=1.5pt]{\color{black} $x_1$};
\end{tikzpicture}
\caption{H-type three-leg equation \eqref{3legasym} arising as the saddle point equation for the star-triangle relation \eqref{YBEasym}.  Double edges are drawn to be consistent with Figure \ref{YBEasymfig}.}
\label{3legfigh}
\end{figure}

For all solutions of the asymmetric star-triangle relation \eqref{YBEasym} considered in this paper, there exists a change of variables of the form \eqref{ptrans}, such that the three-leg saddle point equation \eqref{3legasym} may be written in the form
\begin{align}
\label{quadequation2}
H(x,u,y,v;\alpha,\beta)=0\,,
\end{align}
where the function $H(x,u,y,v;\alpha,\beta)$  is a polynomial that is linear in each of the variables $x,u,y,v$ respectively.  The correspondence between \eqref{3legasym} and \eqref{quadequation2}, is analogous to the correspondence between \eqref{3legsym}, and \eqref{quadequation}.  That is, for a particular algebraic case in Section \ref{sec:alglim} corresponding to both the $H1_{(\varepsilon=0)}$, and $H1_{(\varepsilon=1)}$ quad equations, the equations \eqref{3legasym} and \eqref{quadequation2} are equivalent up to the change of variables of the form \eqref{ptrans}.  For the remaining cases, \eqref{3legasym} may be written as 
\begin{align}
\label{mult3leg2}
\left.\frac{\partial}{\partial x}\mathcal{A}_\medstar(x,x_1,x_2,x_3;\alpha_1,\alpha_3)\right|_{x=x_0}\hspace{-0.4cm}=\Log\left(\overline{\Psi}(\alpha_1\,|\,x_1,x_0)\,\Psi(\alpha_1+\alpha_3\,|\,x_2,x_0)\,\overline{\Phi}(\alpha_3\,|\,x_3,x_0)\right) +2\pi\ii k\,,
\end{align}
for some $k\in\mathbb{Z}$.  Then the equation $\overline{\Psi}(\alpha_1\,|\,x_1,x_0)\,\Psi(\alpha_1+\alpha_3\,|\,x_2,x_0)\,\overline{\Phi}(\alpha_3\,|\,x_3,x_0)=1$ may be seen to be equivalent to the quad equation \eqref{quadequation2}, up to the change of variables of the form \eqref{ptrans}. Consequently, \eqref{3legsym}, and \eqref{quadequation}, are equivalent in the connected subspaces of the parameter space of variables $x_0,x_1,x_2,x_3,\alpha_1,\alpha_3$, where $k=0$ in \eqref{mult3leg2}.

The equations \eqref{quadequation2} derived from \eqref{YBEasym} also satisfy 3D-consistency, exactly as described in the previous subsection for \eqref{quadequation}, and are identified in this paper with the asymmetric $H$-type ABS equations \cite{ABS,ABS2}.  The resulting Yang-Baxter/3D-consistency correspondence is given by the following Theorem:

\begin{theorem}
\label{Htheorem}
There exists a respective solution $W(\theta\,|\,\spn_i,\spn_j)$, $\oW(\theta\,|\,\spn_i,\spn_j)$, $V(\theta\,|\,\spn_i,\spn_j)$, $\oV(\theta\,|\,\spn_i,\spn_j)$ of the star-triangle relation \eqref{YBEasym} for each of the 3D-consistent $H$-type quad equations \eqref{quadequation2} of the ABS list, which has a quasi-classical expansion that can be written in the form \eqref{qclybe} (with \eqref{adefasym}), through which the corresponding $H$-type quad equation \eqref{quadequation2} is obtained from the saddle point three-leg equation \eqref{3legasym} with a change of variables of the form \eqref{ptrans}.
\end{theorem}

The correspondence between solutions of the star-triangle relation \eqref{YBEasym}, and each of the $H$-type ABS quad equations given by Theorem \ref{Htheorem}, is summarised in Table \ref{Htable}.

\begin{table}[htbp]
\begin{center}
\medskip
\begin{tabular}{c|c|c}
  Boltzmann Weights & Three-Leg Equation & Quad Equation\\
 \hline
 \hline
 \eqref{h3d1quan}  & \eqref{H3d13leg} & $H3_{(\delta=1;\,\varepsilon=1)}$ \eqref{H3d1} \\ 
   \hline
  \eqref{h3d1quanalt} & \eqref{H3d13legalt} & $H3_{(\delta=1;\,\varepsilon=1)}$ \eqref{H3d12} \\ 
   \hline
  \eqref{h3d0quan} & \eqref{H3d03leg} & $H3_{(\delta=0,1;\,\varepsilon=1-\delta)}$ \eqref{H3d01} \\ 
 \hline
  \eqref{h3d02quan}  & \eqref{H3d023leg} & $H3_{(\delta=0;\,\varepsilon=0)}$ \eqref{H3d02} \\ 
   \hline
  \eqref{h2d1quan}  & \eqref{H2d13leg} & $H2_{(\varepsilon=1)}$ \eqref{H2d1} \\
  \hline
   \eqref{h2d1quanalt}  & \eqref{H2d13legalt} & $H2_{(\varepsilon=1)}$ \eqref{H2d12} \\ 
   \hline
  \eqref{h2d0quan} & \eqref{H2d03leg} & $H2_{(\varepsilon=0)}$ \eqref{H2d0} \\ 
   \hline
  \eqref{quanh1e0a} & \eqref{h1e0a3leg} & $H1_{(\varepsilon=1)}$ \eqref{H1d1} \\ 
 \hline
  \eqref{quanh1e0b}  & \eqref{h1e0b3leg} & $H1_{(\varepsilon=1)}$ \eqref{H1d12} \\ 
 \hline
  \eqref{quanh1e0a}  & \eqref{h1e0b3leg3} & $H1_{(\varepsilon=0)}$ \eqref{h1e0b}
\end{tabular}
\caption{$H$-type ABS quad equations \eqref{quadequation2} derived from the star-triangle relation \eqref{YBEasym}.}
\label{Htable}
\end{center}
\end{table}

The proof of Theorem \ref{Htheorem} is provided through explicit examples and computations that appear in Sections \ref{sec:hyplim}, \ref{sec:ratlim}, and \ref{sec:alglim}.  Each of the Boltzmann weights corresponding to solutions of the star-triangle relation \eqref{YBEasym} in Table \ref{Htable} are new, except for the case which corresponds to $H3_{(\delta=0,1;\,\varepsilon=1-\delta)}$ which appeared previously in \cite{BAZHANOV2018509}.  Quad equations for which $\varepsilon=1$, may each be derived from two different solutions of the star-triangle relation \eqref{YBEasym}, which respectively give different three-leg equations that come from the quasi-classical expansion, but result in the same quad equations (up to relabelling of the variables).\footnote{This is consistent with a previously study \cite{BollSuris11} of the three-leg equations, where the latter were obtained from 3D-consistent asymmetric $H$-type ABS equations and associated biquadratic polynomials (rather than from the saddle point equations of respective solutions of the star-triangle relations).}  Also both of the quad equations $H3_{(\delta=1,\varepsilon=0)}$, $H3_{(\delta=0,\varepsilon=1)}$, may be obtained from the same solution of the star-triangle relation, and same three-leg equation, by using different respective changes of variables of the form \eqref{ptrans}. Thus they are referred to here as a single quad equation $H3_{(\delta=0,1;\,\varepsilon=1-\delta)}$.  This can be expected, since one of the quad equations $H3_{(\delta=1,\varepsilon=0)}$, $H3_{(\delta=0,\varepsilon=1)}$, may be transformed into the other equation by a change straightforward change of variables, and vice versa, however it appears that this is the first time that such a symmetry between these equations has been noticed.  Finally the quad equations $H1_{(\delta=1)}$, and $H1_{(\delta=0)}$, may be obtained from the same solution of the star-triangle relation \eqref{qh1e0a}, by using different respective choices of quasi-classical expansion.

\section{Hyperbolic cases}
\label{sec:hyplim}

\subsection{Hyperbolic gamma function}

The central function at the hyperbolic level is known as the hyperbolic gamma function (also non-compact quantum dilogarithm, or double sine function),\footnote{The name typically refers to the convention used for the function \eqref{HGF}.  The convention used here appeared in \cite{Ruijsenaars:1997:FOA}, and the different conventions are related by simple changes of variables,  {\it e.g.}, as outlined in \cite{BultThesis,Rains2009,Spiridonov:2010em}.} and is defined by
\begin{align}
\label{HGF}
\Gamma_h(z;\bb)=\exp\left\{\int_{[0,\infty)}\frac{dx}{x}\left(\frac{\ii z}{x}-\frac{\sinh(2\ii zx)}{2\sinh(x\bb)\sinh(x/\bb)}\right)\right\},\qquad |\,\mbox{Im}\,(z)\,|\,<\,\mbox{Re}\,(\eta)\,,
\end{align}
where the crossing parameter is
\begin{align}
\label{hypeta}
\eta=\frac{\bb+\bb^{-1}}{2}\,.
\end{align}
For the purposes here, the modular parameter $\bb$ will take positive values
\begin{align}
\label{bbdef}
\bb>0\,,
\end{align}
such that the crossing parameter \eqref{hypeta} is also real valued and positive.

The hyperbolic gamma function \eqref{HGF} satisfies the inversion relation
\begin{align}
\label{HGFidents}
\Gamma_h(z;\bb)=\frac{1}{\Gamma_h(-z;\bb)}\,,
\end{align}
and difference equations
\begin{align}
\label{HGFdif}
\displaystyle
\frac{\Gamma_h(z-\ii\bb;\bb)}{\Gamma_h(z;\bb)}=2\cosh(\pi(2z-\ii\bb^{-1})/(2\bb))\,,\quad\frac{\Gamma_h(z-\ii\bb^{-1};\bb)}{\Gamma_h(z;\bb)}=2\cosh(\pi(2z-\ii\bb)\bb/2)\,.
\end{align}
Through the use of the latter equations, the hyperbolic gamma function \eqref{HGF} may be analytically continued to $\mathbb{C}-\ii\bb\mathbb{Z}_{\geq0}-\ii\bb^{-1}\mathbb{Z}_{\geq0}-\{\eta\}$ \cite{Ruijsenaars:1997:FOA}.


Defining the parameter $\hbar$ in terms of $\bb$, as
\begin{align}
\hbar=2\pi \bb^2\,,
\end{align}
the quasi-classical expansion will involve setting
\begin{align}
\label{hypqcl}
\spn_i=\frac{x_i}{\sqrt{2\pi\hbar}}\,,\quad i=0,1,2,3,\qquad\theta_j=\frac{\alpha_j}{\sqrt{2\pi\hbar}}\,,\quad j=1,3\,,
\end{align}
in the star-triangle relations \eqref{YBEsym}, \eqref{YBEasym}, and considering the limit
\begin{align}
\hbar\to0^+\,.
\end{align}
In this limit the leading asymptotics of \eqref{HGF} are given by \cite{FaddeevKashaev}
\begin{align}
\label{HGFqcl}
\Log\Gamma_h(z(2\pi\bb)^{-1};\bb)=\frac{\gamma_h(z)}{\ii\hbar}+O(\hbar)\,,\qquad \im(z)<\pi\,,
\end{align}
where $\gamma_h(z)$ is defined by
\begin{align}
\gamma_h(z)=\dilog(-\EXP^{z})+\frac{\pi^2}{12}-\frac{z^2}{4}\,,
\end{align}
and $\lie(z)$ is the dilogarithm function, defined for $\mathbb{C}-[1,\infty)$ by
\begin{align}
\dilog(z)=-\int^z_0dx\,\frac{\Log(1-x)}{x}\,.
\end{align}

The asymptotics of the hyperbolic gamma function \eqref{HGF} for large $z$ will also be needed. Specifically, the hyperbolic gamma function \eqref{HGF} satisfies \cite{Ruijsenaars1999,Rains2009} 
\begin{align}
\label{HGFinflim}
\Log\Gamma_h(z\pm\kappa;\bb)=\mp\frac{\pi\ii}{2}B_{2,2}(\ii z+\eta;\bb,\bb^{-1})+O(\EXP^{-2\pi \beta |z\pm\kappa|})\,,\qquad\kappa\to\infty\,,
\end{align}
where $0<\beta<\min(\bb,\bb^{-1})$, and $B_{2,2}(z;\omega_1,\omega_2)$ is a multiple Bernoulli polynomial \cite{Narukawa2004247}
\begin{align}
\label{b22def}
B_{2,2}(z;\omega_1,\omega_2)=\frac{z^2}{\omega_1\omega_2}-\frac{(\omega_1+\omega_2)z}{\omega_1\omega_2}+\frac{\omega_1^2+\omega^2_2+3\omega_1\omega_2}{6\omega_1\omega_2}\,.
\end{align}

The asymptotics of the hyperbolic gamma function $\varphi(z)$ in the limit $|z|\to\infty$, allows for many different asymmetric degenerations of the star-triangle relation \eqref{YBEsym}, from which different forms of the star-triangle relations in either \eqref{YBEsym}, or \eqref{YBEasym} can be obtained.  Specifically, different solutions of the star-triangle relation \eqref{YBEsym}, or \eqref{YBEasym}, are obtained, by shifting some of the variables $\spn_0,\spn_1,\spn_2,\spn_3,\theta_1,\theta_3$ to infinity, in the form
\begin{align}
\label{Kshift}
\spn_i\rightarrow \spn_i+n_i\,\kappa,\qquad \theta_i\rightarrow\theta_i+m_i\,\kappa,\qquad \kappa\rightarrow\infty\,,
\end{align}
where $n_i,m_i\in\mathbb{Z}$, such that the integral \eqref{YBEsym} converges when taking $\kappa\rightarrow\infty$.   From the asymptotics of the hyperbolic gamma function \eqref{HGFinflim}, the shifts of the form \eqref{Kshift} will simply result in some additional exponential factors, and in most of the relevant cases these factors simply cancel out of the resulting star-triangle relation.  In this way different Boltzmann weights satisfying either of the star-triangle relations \eqref{YBEsym}, \eqref{YBEasym}, may be obtained, whose quasi-classical expansions will be seen to give the different ABS equations $Q3_{(\delta=1)}$, $Q3_{(\delta=0)}$ $H3_{(\delta=1;\,\varepsilon=1)}$, $H3_{(\delta=0,1;\,\varepsilon=1-\delta)}$, $H3_{(\delta=0;\,\varepsilon=0)}$ at the hyperbolic level.  Note also that there is no known analogue of the asymptotic formula \eqref{HGFinflim} for the elliptic gamma function (elliptic counterpart of \eqref{HGF}), and it is not known how to obtain such degenerations amongst the formulas appearing at the elliptic level, if it is at all possible.  This is a reason why no solutions of the asymmetric type star-triangle relation \eqref{YBEasym} (and no respective classical $H$-type ABS equations) have been found at the elliptic level.

There are also many other limits of the form \eqref{Kshift} which do not result in identities of the form of the star-triangle relations \eqref{YBEsym}, \eqref{YBEasym}.  Such cases are not considered to be integrable, and some examples are given in Appendix \ref{app:notYBE}.  
The resulting classical quad equations obtained from the saddle-point equation are still affine-linear, however do not have the 3D-consistency property.  Such examples suggest that the 3D-consistent equations will only arise from the specific form of the star-triangle relations given in \eqref{YBEsym}, \eqref{YBEasym}.  The equations in Appendix \ref{app:notYBE} may still be of interest however, for example the triangle identity pictured in Figure \ref{trianglefig} describes a deformation of an edge of the Faddeev-Volkov models \cite{Bazhanov:2007mh,Spiridonov:2010em}, into two different edges, which is an example of Z-invariance, a property closely associated to integrability of the lattice model \cite{Baxter:1978xr,Kels:2017fyt}.  

\subsection{\texorpdfstring{$Q3_{(\delta=1)}$}{Q3(delta=1)} case}

\tocless\subsubsection{Star-triangle relation}

Let the spins $\spn_1,\spn_2,\spn_3$ and spectral parameters $\theta_1,\theta_3$ take values
\begin{align}
\label{hypvals}
\spn_i\in\mathbb{R}\,,\quad i=1,2,3,\qquad 0<\theta_1,\theta_3,\theta_1+\theta_3<\eta\,,
\end{align}
where $\eta$ is the crossing parameter defined in \eqref{hypeta}.

The star-triangle relation is given by \cite{Spiridonov:2010em}
\begin{align}
\label{qq3d1}
\begin{split}
\int_{\mathbb{R}}d\spn_0\,S(\spn_0)\,\oW(\theta_1\,|\,\spn_1,\spn_0)\,W(\theta_1+\theta_3\,|\,\spn_2,\spn_0)\,\oW(\theta_3\,|\,\spn_3,\spn_0)\phantom{\,.} \\
=R(\theta_1,\theta_3)\,W(\theta_1\,|\,\spn_2,\spn_3)\,\oW(\theta_1+\theta_3\,|\,\spn_1,\spn_3)\,W(\theta_3\,|\,\spn_1,\spn_2)\,,
\end{split}
\end{align}
where the Boltzmann weights are
\begin{align}
\label{genfadvol}
\begin{split}
\iW(\theta\,|\,\spn_i,\spn_j)&=\ds\frac{\Gamma_h(\spn_i+\spn_j+\ii\theta;\bb)}{\Gamma_h(\spn_i+\spn_j-\ii\theta;\bb)}\frac{\Gamma_h(\spn_i-\spn_j+\ii\theta;\bb)}{\Gamma_h(\spn_i-\spn_j-\ii\theta;\bb)}\,, \\[0.1cm]
\oW(\theta\,|\,\spn_i,\spn_j)&=\ds W(\eta-\theta\,|\,\spn_i,\spn_j)\,,
\end{split}
\end{align}
and
\begin{align}
\label{sgenfadvol}
\begin{split}
S(\spn)&=\ds\frac{1}{2}\,\Gamma_h(2\spn-\ii\eta;\bb)\,\Gamma_h(-2\spn-\ii\eta;\bb) \\[0.1cm]
&=\ds2\sinh(2\pi \spn\bb)\sinh(2\pi \spn\bb^{-1})\,.
\end{split}
\end{align}
The normalisation is given by
\begin{align}
R(\theta_1,\theta_3)=\frac{\Gamma_h(\ii(\eta-2\theta_1);\bb)\,\Gamma_h(\ii(\eta-2\theta_3);\bb)}{\Gamma_h(\ii(\eta-2(\theta_1+\theta_3));\bb)}\,.
\end{align}
From \eqref{HGFinflim}, the asymptotics of the integrand of \eqref{qq3d1} are
\begin{align}
\label{q3d1xinf}
O(\EXP^{-4\pi\eta\, \left|\spn_0\right|})\,,\qquad \spn_0\to\pm\infty\,,
\end{align}
and particularly the integral in \eqref{qq3d1} is absolutely convergent for the values \eqref{hypvals}.

The star-triangle relation \eqref{qq3d1} was originally obtained \cite{Spiridonov:2010em} as a hyperbolic limit of the elliptic master solution of the star-triangle relation \cite{Bazhanov:2010kz}. The connection to $Q3_{(\delta=1)}$ was previously given in \cite{Bazhanov:2016ajm}.  In the hypergeometric function theory, \eqref{qq3d1} is equivalent to a univariate hyperbolic beta integral formula \cite{STOKMAN2005119}.  The star-triangle relation \eqref{qq3d1}, also arises as the $r=1$ case of the {\it lens hyperbolic gamma function} solution to the star-triangle relation, and corresponding hyperbolic beta sum/integral formula \cite{GahramanovKels}.  

For future reference, the integral on the left hand side of \eqref{qq3d1} will be denoted by
\begin{align}
\label{q3d1int}
I_{14}(\theta_1,\theta_3\,|\,\spn_1,\spn_2,\spn_3):=\int_\mathcal{C}d\spn\,\rho_{14}(\theta_1,\theta_3\,|\,\spn_1,\spn_2,\spn_3;\spn)\,,
\end{align}
where
\begin{align}
\label{q3d1rho}
\rho_{14}(\theta_1,\theta_3\,|\,\spn_1,\spn_2,\spn_3;\spn):=S(\spn)\,\oW(\theta_1\,|\,\spn_1,\spn)\,W(\theta_1+\theta_3\,|\,\spn_2,\spn)\,\oW(\theta_3\,|\,\spn_3,\spn)\,,
\end{align}
and $\mathcal{C}$ is a deformation of $\mathbb{R}$ that separates the points $\ii\theta_1\pm \spn_1+\ii\bb\mathbb{Z}_{\geq0}+\ii\bb^{-1}\mathbb{Z}_{\geq0}$, $\ii\theta_3\pm \spn_3+\ii\bb\mathbb{Z}_{\geq0}+\ii\bb^{-1}\mathbb{Z}_{\geq0}$, $\ii(\eta-\theta_1-\theta_3)\pm \spn_2+\ii\bb\mathbb{Z}_{\geq0}+\ii\bb^{-1}\mathbb{Z}_{\geq0}$, from their negatives.  The integral \eqref{q3d1int} is an analytic continuation of \eqref{qq3d1}, for general complex valued variables $\spn_1,\spn_2,\spn_3$, $\theta_1,\theta_3$.  If
\begin{align}
\label{q3d1rvals}
\re(\theta_1)\pm \im(\spn_1)\,, \re(\theta_3)\pm \im(\spn_3)\,, \eta-\re(\theta_1+\theta_3)\pm \im(\spn_2)>0\,,
\end{align}
the contour $\mathcal{C}$ can be chosen to be $\mathbb{R}$.

\tocless\subsubsection{Classical integrable equations}

In the quasi-classical limit \eqref{hypqcl}, the crossing parameter becomes
\begin{align}
\eta\to\eta_0=\pi\,,
\end{align}
so that the classical variables $x_1,x_2,x_3$, $\alpha_1,\alpha_3$, take values
\begin{align}
\label{hypcvals}
x_1,x_2,x_3\in\mathbb{R}\,,\qquad 0<\alpha_1,\alpha_3,\alpha_1+\alpha_3<\pi\,.
\end{align}
However the variables $x_1,x_2,x_3$, $\alpha_1,\alpha_3$, can in principle take more general complex values, inherited from the analytic continuation that was described above for \eqref{q3d1int}.

The Lagrangian function for this case is defined by
\begin{align}
\label{lagq3d1}
\begin{split}
\ds\mathcal{L}(\alpha\,|\,x_i,x_j)&=\ds\lie(-\EXP^{x_i+x_j+\ii\alpha})+\lie(-\EXP^{x_i-x_j+\ii\alpha})+\lie(-\EXP^{-x_i+x_j+\ii\alpha})+\lie(-\EXP^{-x_i-x_j+\ii\alpha}) \\
&\quad-2\,\lie(-\EXP^{\ii\alpha})+x_i^2+x_j^2-\frac{\alpha^2}{2}+\frac{\pi^2}{6}\,.
\end{split}
\end{align}
This satisfies the following relations
\begin{align}
\mathcal{L}(\alpha\,|\,x_i,x_j)=\mathcal{L}(\alpha\,|\,x_j,x_i)\,,\qquad\mathcal{L}(\alpha\,|\,x_i,x_j)=-\mathcal{L}(-\alpha\,|\,x_i,x_j)\,.
\end{align}

Using the asymptotics \eqref{HGFqcl}, the leading order $O(\hbar^{-1})$ quasi-classical expansion \eqref{hypqcl} of the integrand \eqref{q3d1rho} is
\begin{align}
\label{q3d1qcl}
\begin{split}
\frac{\alpha_1^2+\alpha_3^2+(\pi-\alpha_1-\alpha_3)^2-4\lie(\EXP^{-\ii\alpha_1})-4\lie(\EXP^{-\ii\alpha_3})-4\lie(-\EXP^{\ii(\alpha_1+\alpha_3)})}{2\ii\hbar}+ \\[0.1cm]
\Log \rho_{14}(\tfrac{\alpha_1}{\sqrt{2\pi\hbar}},\tfrac{\alpha_3}{\sqrt{2\pi\hbar}}\,|\,\tfrac{x_1}{\sqrt{2\pi\hbar}},\tfrac{x_2}{\sqrt{2\pi\hbar}}\tfrac{x_3}{\sqrt{2\pi\hbar}};\tfrac{x_0}{\sqrt{2\pi\hbar}})=(\ii\hbar)^{-1}\astr{x_0}+O(1)\,,
\end{split}
\end{align}
where
\begin{align}
\astr{x_0}=C(x_0)+\ol(\alpha_1\,|\,x_0,x_1)+\lag(\alpha_1+\alpha_3\,|\,x_0,x_2)+\ol(\alpha_3\,|\,x_0,x_3)\,,
\end{align}
and
\begin{align}
\begin{gathered}
\ol(\alpha\,|\,x_i,x_j)=\lag(\pi-\alpha\,|\,x_i,x_j)\,, \\
\ds\iC(x)=2\pi\ii x\,\mbox{sgn}(\re(x))\,.
\end{gathered}
\end{align}

The saddle point three-leg equation \eqref{3legsym} is then given by
\begin{align}
\label{Q3d13leg}
\left.\frac{\partial \astr{x}}{\partial x}\right|_{x=x_0}\hspace{-0.3cm}=\ovphi(\alpha_1\,|\,x_1,x_0)+\varphi(\alpha_1+\alpha_3\,|\,x_2,x_0)+\ovphi(\alpha_3\,|\,x_3,x_0)=0\,,
\end{align}
where $\varphi(\alpha\,|\,x_i,x_j)$ is defined by
\begin{align}
\begin{split}
\varphi(\alpha\,|\,x_i,x_j)&=\Log(1+\EXP^{-x_i-x_j+\ii\alpha})+\Log(1+\EXP^{x_i-x_j+\ii\alpha}) \\
&\;-\Log(1+\EXP^{x_i+x_j+\ii\alpha})-\Log(1+\EXP^{-x_i+x_j+\ii\alpha})+2x_j+\ii\alpha\,\mbox{sgn}(\re(x_j))\,,
\end{split}
\end{align}
and
\begin{align}
\ovphi(\alpha\,|\,x_i,x_j)=\vphi(\pi-\alpha\,|\,x_i,x_j)\,.
\end{align}
Note that the function $\varphi(\alpha\,|\,x_i,x_j)$ is not symmetric upon the exchange of $x_i\leftrightarrow x_j$, and thus the ordering of pairs of the variables $x_0,x_1,x_2,x_3$, appearing in arguments of $\varphi(\alpha\,|\,x_i,x_j)$, in \eqref{Q3d13leg}, needs to be taken into consideration.

The equation \eqref{Q3d13leg} is a three-leg form of $Q3_{(\delta=1)}$, arising as the equation for the saddle point of the star-triangle relation \eqref{qq3d1} in the limit \eqref{hypqcl}.  This was also previously obtained in \cite{Bazhanov:2016ajm}.  With the following change of variables
\begin{align}
\begin{gathered}
x=-\cosh(x_0),\qquad u=\cosh(x_1),\qquad y=\cosh(x_2),\qquad v=\cosh(x_3), \\[0.1cm]
\alpha=\EXP^{-\ii\alpha_1},\qquad\beta=-\EXP^{-\ii(\alpha_1+\alpha_3)},
\end{gathered}
\end{align}
the exponential of the three-leg equation \eqref{Q3d13leg} may be written in the form
\begin{align}
Q(x,u,y,v;\alpha,\beta)=0\,,
\end{align}
where 
\begin{align}
\label{Q3d1}
\begin{split}
Q(x,u,y,v;\alpha,\beta)=&(\alpha^2-1)\beta(xu+yv)-(\beta^2-1)\alpha(xy+uv)+(\alpha^2-\beta^2)(xv+uy)\phantom{=0\,.} \\
&+(\alpha^2-1)(\beta^2-1)(\alpha^2-\beta^2)(4\alpha\beta)^{-1}\,.
\end{split}
\end{align}
This is identified as $Q3_{(\delta=1)}$ \cite{ABS}, which is an affine-linear quad equation that satisfies the 3D-consistency condition exactly as defined in Section \ref{sec:overview}.

\subsection{\texorpdfstring{$Q3_{(\delta=0)}$}{Q3(delta=0)} case}

\tocless\subsubsection{Star-triangle relation}

Let the variables $\spn_1,\spn_2,\spn_3$, and $\theta_1,\theta_3$, take values \eqref{hypvals}.  The star-triangle relation is given in this case by \cite{Bazhanov:2007mh}
\begin{align}
\label{qq3d0}
\begin{split}
\int_{\mathbb{R}}d\spn_0\,\oW(\theta_1\,|\,\spn_1,\spn_0)\,W(\theta_1+\theta_3\,|\,\spn_2,\spn_0)\,\oW(\theta_3\,|\,\spn_3,\spn_0)\phantom{\qquad\,.} \\
=R(\theta_1,\theta_3)\,W(\theta_1\,|\,\spn_2,\spn_3)\,\oW(\theta_1+\theta_3\,|\,\spn_1,\spn_3)\,W(\theta_3\,|\,\spn_1,\spn_2)\,,
\end{split}
\end{align}
where the Boltzmann weights are
\begin{align}
\label{fadvol}
\begin{split}
\iW(\theta\,|\,\spn_i,\spn_j)&=\frac{\Gamma_h(\spn_i-\spn_j+\ii\theta;\bb)}{\Gamma_h(\spn_i-\spn_j-\ii\theta;\bb)}\,, \\
\oW(\theta\,|\,\spn_i,\spn_j)&=W(\eta-\theta\,|\,\spn_i,\spn_j)\,,
\end{split}
\end{align}
$\eta$ is defined in \eqref{hypeta}, and the normalisation is given by
\begin{align}
R(\theta_1,\theta_3)=\frac{\Gamma_h(\ii(\eta-2\theta_1);\bb)\,\Gamma_h(\ii(\eta-2\theta_3);\bb)}{\Gamma_h(\ii(\eta-2(\theta_1+\theta_3));\bb)}\,.
\end{align}
From \eqref{HGFinflim}, the asymptotics of the integrand of \eqref{qq3d0} are
\begin{align}
\label{q3d0xinf}
O(\EXP^{-4\pi\eta\, \left|\spn_0\right|})\,,\qquad \spn_0\to\pm\infty\,,
\end{align}
and particularly the integral in \eqref{qq3d0} is absolutely convergent for the values \eqref{hypvals}.

Equation \eqref{qq3d0} is the star-triangle relation for the Faddeev-Volkov model and was previously shown to be connected to $Q3_{(\delta=0)}$ and the Hirota difference equation \cite{Bazhanov:2007mh,Bazhanov:2007vg}.  Up to a change of variables, the star-triangle relation \eqref{qq3d0} is equivalent to a hyperbolic analogue of the nonterminating Saalsch\"{u}tz formula (equivalently hyperbolic analogue of Barnes's second lemma), that was derived as a limit of an 8-parameter hyperbolic hypergeometric function \cite{BultThesis}.

Now define
\begin{align}
\label{q3d0int}
I_{6,Q}(\theta_1,\theta_3\,|\,\spn_1,\spn_2,\spn_3)=\int_{\mathcal{C}}d\spn\,\rho_{6,Q}(\theta_1,\theta_3\,|\,\spn_1,\spn_2,\spn_3;\spn)\,,
\end{align}
where
\begin{align}
\label{q3d0rho}
\rho_{6,Q}(\theta_1,\theta_3\,|\,\spn_1,\spn_2,\spn_3;\spn)=\oW(\theta_1\,|\,\spn_1,\spn)\,W(\theta_1+\theta_3\,|\,\spn_2,\spn)\,\oW(\theta_3\,|\,\spn_3,\spn)\,,
\end{align}
and $\mathcal{C}$ is a deformation of $\mathbb{R}$ that separates the points $\ii\theta_1+\spn_1+\ii\bb\mathbb{Z}_{\geq0}+\ii\bb^{-1}\mathbb{Z}_{\geq0}$, $\ii\theta_3+\spn_3+\ii\bb\mathbb{Z}_{\geq0}+\ii\bb^{-1}\mathbb{Z}_{\geq0}$, $\ii(\eta-\theta_1-\theta_3)+\spn_2+\ii\bb\mathbb{Z}_{\geq0}+\ii\bb^{-1}\mathbb{Z}_{\geq0}$, from the points $-\ii\theta_1+\spn_1-\ii\bb\mathbb{Z}_{\geq0}-\ii\bb^{-1}\mathbb{Z}_{\geq0}$, $-\ii\theta_3+\spn_3-\ii\bb\mathbb{Z}_{\geq0}-\ii\bb^{-1}\mathbb{Z}_{\geq0}$, $-\ii(\eta-\theta_1-\theta_3)+\spn_2-\ii\bb\mathbb{Z}_{\geq0}-\ii\bb^{-1}\mathbb{Z}_{\geq0}$.    For the values \eqref{q3d1rvals}, the contour $\mathcal{C}$ may be chosen to be $\mathbb{R}$.

The star-triangle relation \eqref{qq3d0}, may be obtained from the following limit of \eqref{qq3d1} \cite{Spiridonov:2010em}:

\begin{prop} \label{q3d0prop}
For the values \eqref{q3d1rvals},
\begin{align}
\label{q3d0lim}
\begin{split}
\lim_{\kappa\to\infty}\EXP^{4\pi\eta\kappa+2\pi\eta(\spn_1+\spn_3)-2\pi\theta_1(\spn_1-\spn_2)-2\pi\theta_3(\spn_3-\spn_2)}I_{14}(\theta_1,\theta_3\,|\,\spn_1+\kappa,\spn_2+\kappa,\spn_3+\kappa)\phantom{\,,} \\
=I_{6,Q}(\theta_1,\theta_3\,|\,\spn_1,\spn_2,\spn_3)\,,
\end{split}
\end{align}
where $I_{14}(\theta_1,\theta_3\,|\,\spn_1,\spn_2,\spn_3)$ is defined in \eqref{q3d1int}.

\end{prop}

\begin{proof}
For the values \eqref{q3d1rvals}, the contours of both $I_{14}$, and $I_{6,Q}$, may be chosen to be $\mathbb{R}$.  Following a change of integration variable of $\sigma\to\sigma+\kappa$ in \eqref{q3d1int}, by the asymptotics \eqref{q3d1xinf}, \eqref{q3d0xinf}, and the asymptotic formula \eqref{HGFinflim}, the combination of the transformed integrand \eqref{q3d1rho} and the exponential factor on the left hand side of \eqref{q3d0lim} is uniformly bounded on $\mathbb{R}$, and the result follows by dominated convergence.
\end{proof}

The star-triangle relation \eqref{qq3d0} follows from Proposition \ref{q3d0prop}, by using \eqref{HGFinflim} to take the same limit of the right hand side of \eqref{qq3d1}.

\tocless\subsubsection{Classical integrable equations}

Let the classical variables $x_1,x_2,x_3$, $\alpha_1,\alpha_3$, take values in \eqref{hypcvals}.  The Lagrangian function for this case is defined by
\begin{align}
\label{lagq3d0}
\ds\mathcal{L}(\alpha\,|\,x_i,x_j)=\lie(-\EXP^{x_i-x_j+\ii\alpha})+\lie(-\EXP^{x_j-x_i+\ii\alpha})-2\,\lie(-\EXP^{\ii\alpha})+\frac{(x_i-x_j)^2}{2}\,.
\end{align}
This satisfies the following relations
\begin{align}
\mathcal{L}(\alpha\,|\,x_i,x_j)=\mathcal{L}(\alpha\,|\,x_j,x_i)\,,\quad\mathcal{L}(\alpha\,|\,x_i,x_j)=-\mathcal{L}(-\alpha\,|\,x_i,x_j)\,.
\end{align}

Using the asymptotics \eqref{HGFqcl}, the leading order $O(\hbar^{-1})$ quasi-classical expansion \eqref{hypqcl} of the integrand \eqref{q3d0rho} is
\begin{align}
\label{q3d0qcl}
\begin{split}
\frac{\alpha_1^2+\alpha_3^2+(\pi-\alpha_1-\alpha_3)^2-4\lie(\EXP^{-\ii\alpha_1})-4\lie(\EXP^{-\ii\alpha_3})-4\lie(-\EXP^{\ii(\alpha_1+\alpha_3)})}{2\ii\hbar}+ \\
\Log \rho_{6,Q}(\tfrac{\alpha_1}{\sqrt{2\pi\hbar}},\tfrac{\alpha_3}{\sqrt{2\pi\hbar}}\,|\,\tfrac{x_1}{\sqrt{2\pi\hbar}},\tfrac{x_2}{\sqrt{2\pi\hbar}}\tfrac{x_3}{\sqrt{2\pi\hbar}};\tfrac{x_0}{\sqrt{2\pi\hbar}})=(\ii\hbar)^{-1}\astr{x_0}+O(1)\,,
\end{split}
\end{align}
where
\begin{align}
\astr{x_0}=\ol(\alpha_1\,|\,x_1,x_0)+\lag(\alpha_1+\alpha_3\,|\,x_2,x_0)+\ol(\alpha_3\,|\,x_3,x_0)\,,
\end{align}
and
\begin{align}
\ol(\alpha\,|\,x_i,x_j)=\lag(\pi-\alpha\,|\,x_i,x_j)\,.
\end{align}

The saddle point three-leg equation \eqref{3legsym} is then given by
\begin{align}
\label{Q3d03leg}
\left.\frac{\partial \astr{x}}{\partial x}\right|_{x=x_0}\hspace{-0.3cm}=\ovphi(\alpha_1\,|\,x_1,x_0)+\varphi(\alpha_1+\alpha_3\,|\,x_2,x_0)+\ovphi(\alpha_3\,|\,x_3,x_0)=0\,,
\end{align}
where $\varphi(\alpha\,|\,x_i,x_j)$, is defined by
\begin{align}
\varphi(\alpha\,|\,x_i,x_j)=\Log(1+\EXP^{x_i-x_j+\ii\alpha})-\Log(1+\EXP^{x_j-x_i+\ii\alpha})-x_i+x_j\,,
\end{align}
and
\begin{align}
\ovphi(\alpha\,|\,x_i,x_j)=\vphi(\eta-\alpha\,|\,x_i,x_j)\,.
\end{align}

The equation \eqref{Q3d03leg} is a three-leg form of $Q3_{(\delta=0)}$, arising as the equation for the saddle point of the star-triangle relation \eqref{qq3d0} in the limit \eqref{hypqcl}.  This was also previously obtained in \cite{Bazhanov:2007mh,Bazhanov:2016ajm}.  With the following change of variables
\begin{align}
x=-\EXP^{x_0},\qquad u=\EXP^{x_1},\qquad y=\EXP^{x_2},\qquad v=\EXP^{x_3},\qquad\alpha=\EXP^{-\ii\alpha_1},\qquad\beta=-\EXP^{-\ii(\alpha_1+\alpha_3)},
\end{align}
the exponential of the three-leg equation \eqref{Q3d03leg} may be written in the form
\begin{align}
Q(x,u,y,v;\alpha,\beta)=0\,,
\end{align}
where 
\begin{align}
\label{Q3d0}
Q(x,u,y,v;\alpha,\beta)=(\alpha^2-1)\beta(xu+yv)-(\beta^2-1)\alpha(xy+uv)+(\alpha^2-\beta^2)(xv+uy)\,.
\end{align}
This is identified as $Q3_{(\delta=0)}$ \cite{ABS}, an affine-linear quad equation that satisfies the 3D-consistency condition exactly as defined in Section \ref{sec:overview}.

\subsection{\texorpdfstring{$H3_{(\delta=1;\,\varepsilon=1)}$}{H3(delta=1,epsilon=1)} case}

\tocless\subsubsection{Star-triangle relation}

Let the variables $\spn_1,\spn_2,\spn_3$, and $\theta_1,\theta_3$, take values \eqref{hypvals}.  The star-triangle relation is given in this case by
\begin{align}
\label{qh3d1}
\begin{split}
\ds\int_{\mathbb{R}}d\spn_0\;\iS(\spn_0)\,\oV({\theta_1}\,|\,\spn_1,\spn_0)\,V({\theta_1+\theta_3}\,|\,\spn_2,\spn_0)\,\oW({\theta_3}\,|\,\spn_0,\spn_3)\phantom{\,.} \ds \\
\ds =R(\theta_3)\,V({\theta_1}\,|\,\spn_2,\spn_3)\,\oV(\theta_1+\theta_3\,|\,\spn_1,\spn_3)\,W({\theta_3}\,|\,\spn_2,\spn_1)\,,
\end{split}
\end{align}
where the Boltzmann weights are
\begin{align}
\label{h3d1quan}
\begin{split}
V(\theta\,|\,\spn_i,\spn_j)&=\ds\Gamma_h(\spn_i+\spn_j+\ii\theta;\bb)\,\Gamma_h(\spn_i-\spn_j+\ii\theta;\bb)\,, \\[0.1cm]
\oV(\theta\,|\,\spn_i,\spn_j)&=\ds V(\eta-\theta\,|\,-\spn_i,\spn_j)\,, \\[0.1cm]
\oW(\theta\,|\,\spn_i,\spn_j)&=\frac{\Gamma_h(\spn_i+\spn_j+\ii(\eta-\theta);\bb)\,\Gamma_h(\spn_i-\spn_j+\ii(\eta-\theta);\bb)}{\Gamma_h(\spn_i+\spn_j-\ii(\eta-\theta);\bb)\,\Gamma_h(\spn_i-\spn_j-\ii(\eta-\theta);\bb)}\,, \\[0.1cm]
W(\theta\,|\,\spn_i,\spn_j)&=\ds\frac{\Gamma_h(\spn_i-\spn_j+\ii\theta;\bb)}{\Gamma_h(\spn_i-\spn_j-\ii\theta;\bb)}\,,
\end{split}
\end{align}
$\eta$ is defined in \eqref{hypeta}, and
\begin{align}
\label{sh3d1quan}
\begin{split}
\iS(\spn)&=2\sinh(2\pi \spn\bb)\sinh(2\pi \spn\bb^{-1})\,,  \\[0.1cm]
R(\theta)&=\Gamma_h(\ii(\eta-2\theta);\bb)\,.
\end{split}
\end{align}
From \eqref{HGFinflim}, the asymptotics of the integrand of \eqref{qh3d1} are
\begin{align}
\label{h3d1xinf}
O(\EXP^{-2\pi(\eta-\theta_3-\ii(\spn_1-\spn_2))\left|\spn_0\right|})\,,\qquad \spn_0\to\pm\infty\,,
\end{align}
and particularly the integral in \eqref{qh3d1} is absolutely convergent for the values \eqref{hypvals}.

The Boltzmann weights $\oW(\theta\,|\,\spn_i,\spn_j)$ and $W(\theta\,|\,\spn_i,\spn_j)$ appear also in \eqref{genfadvol} ($Q3_{(\delta=1)}$ case), and \eqref{fadvol} ($Q3_{(\delta=0)}$ case), respectively, while the Boltzmann weight $V(\theta\,|\,\spn_i,\spn_j)$ is not symmetric in the exchange of spins $\spn_i\leftrightarrow \spn_j$, and generally has non-vanishing imaginary component.

The star-triangle relation \eqref{qh3d1} did not appear before, but up to a change of variables is equivalent to a hyperbolic beta integral of Askey-Wilson type \cite{STOKMAN2005119,Ruijsenaars2003}.

Define
\begin{align}
\label{h3d1int}
I_{10}(\theta_1,\theta_3\,|\,\spn_1,\spn_2,\spn_3)=\int_{\mathcal{C}}d\spn\,\rho_{10}(\theta_1,\theta_3\,|\,\spn_1,\spn_2,\spn_3;\spn)\,,
\end{align}
where
\begin{align}
\label{h3d1rho}
\rho_{10}(\theta_1,\theta_3\,|\,\spn_1,\spn_2,\spn_3;\spn)=S(\spn)\,\oV({\theta_1}\,|\,\spn_1,\spn)\,V({\theta_1+\theta_3}\,|\,\spn_2,\spn)\,\oW({\theta_3}\,|\,\spn_0,\spn)\,,
\end{align}
and $\mathcal{C}$ is a deformation of $\mathbb{R}$ separating the points $\ii\theta_1+\spn_1 +\ii\bb\mathbb{Z}_{\geq0}+\ii\bb^{-1}\mathbb{Z}_{\geq0}$, $\ii\theta_3\pm \spn_3 +\ii\bb\mathbb{Z}_{\geq0}+\ii\bb^{-1}\mathbb{Z}_{\geq0}$, $\ii(\eta-\theta_1-\theta_3)-\spn_2 +\ii\bb\mathbb{Z}_{\geq0}+\ii\bb^{-1}\mathbb{Z}_{\geq0}$, from their negatives.    If
\begin{align}
\label{h3d1rvals}
\re(\theta_1)+ \im(\spn_1)\,, \re(\theta_3)\pm \im(\spn_3)\,, \eta-\re(\theta_1+\theta_3)- \im(\spn_2)>0\,,
\end{align}
the contour $\mathcal{C}$ can be chosen to be $\mathbb{R}$.

The star-triangle relation \eqref{qh3d1}, may be obtained from the following limit of \eqref{qq3d1}:

\begin{prop} \label{h3d1prop}
For the values \eqref{q3d1rvals},
\begin{align}
\label{h3d1lim}
\begin{split}
\lim_{\kappa\to\infty}&\left(\EXP^{4\pi(\eta+\theta_3-\ii(\spn_1-\spn_2))\kappa+\ii\pi(\eta(\eta-2(\theta_1+\ii \spn_1))-\theta_3(\theta_3+2\theta_1)-\spn_1^2+\spn_2^2+2\ii\theta_1(\spn_1-\spn_2)-2\ii\theta_3x_2)}\right. \\
&\left.\phantom{\EXP^{2^2}}\times I_{14}(\theta_1+\ii\kappa,\theta_3\,|\,\spn_1+\kappa,\spn_2+\kappa,\spn_3)\right)=I_{10}(\theta_1,\theta_3\,|\,\spn_1,\spn_2,\spn_3)\,,
\end{split}
\end{align}
where $I_{14}(\theta_1,\theta_3\,|\,\spn_1,\spn_2,\spn_3)$ is defined in \eqref{q3d1int}.

\end{prop}

The proof is analogous to Proposition \ref{q3d0prop}, without needing the change of variables of the integrand.  By analytic continuation the domain of values may be extended to the values that are permitted by the contour, as given by the conditions stated below \eqref{h3d1rho}.  
%
The star-triangle relation \eqref{qh3d1} follows from Proposition \ref{h3d1prop}, by using \eqref{HGFinflim} to take the same limit of the right hand side of \eqref{qq3d1}.

\tocless\subsubsection{Classical integrable equations}

Let the classical variables $x_1,x_2,x_3$, $\alpha_1,\alpha_3$, take values in \eqref{hypcvals}.  The Lagrangian functions for this case are defined by
\begin{align}
\label{lagh3d1}
\begin{split}
\Lambda(\alpha\,|\,x_i,x_j)&=\ds\lie(-\EXP^{x_i+ x_j+\ii\alpha})+\lie(-\EXP^{x_i- x_j+\ii\alpha})+\frac{(x_i+\ii\alpha)^2+x_j^2}{2}\,, \\[0.0cm]
\ol(\alpha\,|\,x_i,x_j)&=\olie(\EXP^{x_i+x_j-\ii\alpha})+\olie(\EXP^{x_i-x_j-\ii\alpha})+\olie(\EXP^{-x_i+x_j-\ii\alpha})+\olie(\EXP^{-x_i-x_j-\ii\alpha}) \\
&\quad -2\,\olie(\EXP^{-\ii\alpha})+x_i^2+x_j^2-\frac{(\pi-\alpha)^2}{2}+\frac{\pi^2}{6} \,.
\end{split}
\end{align}
The $\ol(\alpha\,|\,x_i,x_j)$ is equivalent to $\mathcal{L}(\pi-\alpha\,|\,x_i,x_j)$ from \eqref{lagq3d1} ($Q3_{(\delta=1)}$ case), and thus satisfies
\begin{align}
\ol(\alpha\,|\,x_i,x_j)=\ol(\alpha\,|\,x_j,x_i)\,,\qquad\ol(\pi+\alpha\,|\,x_i,x_j)=-\ol(\pi-\alpha\,|\,x_i,x_j)\,.
\end{align}
However $\Lambda(\alpha\,|\,x_i,x_j)$ does not satisfy such symmetries, and particularly the ordering of $x_i,x_j$ in $\Lambda(\alpha\,|\,x_i,x_j)$ needs to be taken into account.

Using the asymptotics \eqref{HGFqcl}, the leading order $O(\hbar^{-1})$ quasi-classical expansion \eqref{hypqcl} of the integrand \eqref{h3d1rho} is
\begin{align}
\label{h3d1qcl}
\begin{split}
\frac{\alpha_3^2-2\pi\alpha_3-4\lie(\EXP^{-\ii\alpha_3})}{2\ii\hbar}+\Log \rho_{10}(\tfrac{\alpha_1}{\sqrt{2\pi\hbar}},\tfrac{\alpha_3}{\sqrt{2\pi\hbar}}\,|\,\tfrac{x_1}{\sqrt{2\pi\hbar}},\tfrac{x_2}{\sqrt{2\pi\hbar}}\tfrac{x_3}{\sqrt{2\pi\hbar}};\tfrac{x_0}{\sqrt{2\pi\hbar}})\phantom{\,,} \\[0.1cm]
=(\ii\hbar)^{-1}\astr{x_0}+O(1)\,,
\end{split}
\end{align}
where
\begin{align}
\astr{x_0}=C(x_0)+\olam(\alpha_1\,|\,x_1,x_0)+\Lambda(\alpha_1+\alpha_3\,|\,x_2,x_0)+\ol(\alpha_3\,|\,x_0,x_3)\,, 
\end{align}
and
\begin{align}
\begin{gathered}
\olam(\alpha\,|\,x_i,x_j)=\Lambda(\pi-\alpha\,|\,-x_i,x_j)\,, \\
C(x)=2\pi\ii x\,\mbox{sgn}(\re(x))\,.
\end{gathered}
\end{align}

The saddle point three-leg equation \eqref{3legasym} is then given by
\begin{align}
\label{H3d13leg}
\left.\frac{\partial \astr{x}}{\partial x}\right|_{x=x_0}\hspace{-0.3cm}=\ovphib(\alpha_1\,|\,x_1,x_0)+ \phi(\alpha_1+\alpha_3\,|\,x_2,x_0)+\vphi(\alpha_3\,|\,x_3,x_0)=0\,,
\end{align}
where $\phi(\alpha\,|\,x_i,x_j)$, and $\vphi(\alpha\,|\,x_i,x_j)$, are defined by
\begin{align}
\begin{split}
\phi(\alpha\,|\,x_i,x_j)&=\Log(1+\EXP^{x_i-x_j+\ii\alpha})-\Log(1+\EXP^{x_i+x_j+\ii\alpha})+2x_j+\pi\ii\,\mbox{sgn}(\re(x_j))\,, \\
\vphi(\alpha\,|\,x_i,x_j)&=\Log(1-\EXP^{x_i-x_j-\ii\alpha})+\Log(1-\EXP^{-x_i-x_j-\ii\alpha})-\Log(1-\EXP^{x_i+x_j-\ii\alpha})-\Log(1-\EXP^{x_j-x_i-\ii\alpha}) \,,
\end{split}
\end{align}
and
\begin{align}
\ovphib(\alpha\,|\,x_i,x_j)=\phi(\pi-\alpha\,|\,-x_i,x_j)\,.
\end{align}

Equation \eqref{H3d13leg} is a three-leg form of $H3_{(\delta=1;\,\varepsilon=1)}$, arising as the equation for the saddle point of the star-triangle relation \eqref{qh3d1} in the limit \eqref{hypqcl}.  With the following change of variables
\begin{align}
x=-\cosh(x_0),\;\; v=\cosh(x_3),\;\; u=\EXP^{-x_1},\;\; y=\EXP^{x_2},\;\;\alpha=\EXP^{-\ii\alpha_1},\;\;\beta=-\EXP^{-\ii(\alpha_1+\alpha_3)},
\end{align}
the exponential of the three-leg equation \eqref{H3d13leg} may be written in the form
\begin{align}
H(x,u,y,v;\alpha,\beta)=0\,,
\end{align}
where 
\begin{align}
\label{H3d1}
H(x,u,y,v;\alpha,\beta)=2(uv+xy)\beta - 2(yv+xu)\alpha + uy(\beta^2-\alpha^2)(\alpha\beta)^{-1} - (\alpha^2-\beta^2)\,.
\end{align}
This is identified as $H3_{(\delta=1;\,\varepsilon=1)}$ \cite{ABS2}, an affine-linear quad equation that satisfies the 3D-consistency condition exactly as defined in Section \ref{sec:overview}.

\subsection{\texorpdfstring{$H3_{(\delta=1;\,\varepsilon=1)}$}{H3(delta=1,epsilon=1)} case (alternate form)}

\tocless\subsubsection{Star-triangle relation}

Let the variables $\spn_1,\spn_2,\spn_3$, and $\theta_1,\theta_3$, take values \eqref{hypvals}.  The star-triangle relation is given in this case by
\begin{align}
\label{qh3d1alt}
\begin{split}
\ds\int_{\mathbb{R}}d\spn_0\,\oV({\theta_1}\,|\,\spn_1,\spn_0)\,V({\theta_1+\theta_3}\,|\,\spn_2,\spn_0)\,\oW({\theta_3}\,|\,\spn_0,\spn_3)\phantom{\quad\,.} \\
\ds=R(\theta_3)\,V({\theta_1}\,|\,\spn_2,\spn_3)\,\oV({\theta_1+\theta_3}\,|\,\spn_1,\spn_3)\,W({\theta_3}\,|\,\spn_2,\spn_1)\,,
\end{split}
\end{align}
where the Boltzmann weights are
\begin{align}
\label{h3d1quanalt}
\begin{split}
V(\theta\,|\,\spn_i,\spn_j)&=\ds\Gamma_h(\spn_i+\spn_j+\ii\theta)\,\Gamma_h(\spn_j-\spn_i+\ii\theta;\bb)\,, \\[0.1cm]
\oV(\theta\,|\,\spn_i,\spn_j)&=\ds V(\eta-\theta\,|\,\spn_i,-\spn_j)\,, \\[0.1cm]
\oW(\theta\,|\,\spn_i,\spn_j)&=\ds\frac{\Gamma_h(\spn_i-\spn_j+\ii(\eta-\theta);\bb)}{\Gamma_h(\spn_i-\spn_j-\ii(\eta-\theta);\bb)}\,, \\[0.1cm]
W(\theta\,|\,\spn_i,\spn_j)&=\ds\frac{\Gamma_h(\spn_i+\spn_j+\ii\theta;\bb)\,\Gamma_h(\spn_i-\spn_j+\ii\theta;\bb)}{\Gamma_h(\spn_i+ \spn_j-\ii\theta;\bb)\,\Gamma_h(\spn_i-\spn_j-\ii\theta;\bb)}\,,
\end{split}
\end{align}
$\eta$ is defined in \eqref{hypeta}, and
\begin{align}
\begin{split}
R(\theta)&=\ds\Gamma_h(\ii(\eta-2\theta);\bb)\,.
\end{split}
\end{align}
From \eqref{HGFinflim}, the asymptotics of the integrand of \eqref{qh3d1alt} are
\begin{align}
\label{h3d12xinf}
O(\EXP^{-4\pi\eta\,\left|\spn_0\right|})\,,\qquad \spn_0\to\pm\infty\,,
\end{align}
and particularly the integral in \eqref{qh3d1alt} is absolutely convergent for the values \eqref{hypvals}.

Similarly to the preceding case of $H3_{(\delta=1;\,\varepsilon=1)}$ in \eqref{h3d1quan}, $\oW(\theta\,|\,\spn_i,\spn_j)$, and $W(\theta\,|\,\spn_i,\spn_j)$ appear for \eqref{genfadvol} ($Q3_{(\delta=1)}$ case), and \eqref{fadvol} ($Q3_{(\delta=0)}$ case), respectively, while the Boltzmann weight $V(\theta\,|\,\spn_i,\spn_j)$ is not symmetric in the exchange of spins, and generally has a non-vanishing complex imaginary component.  There are only sign differences between the $V(\theta\,|\,\spn_i,\spn_j)$ here and the $V(\theta\,|\,\spn_i,\spn_j)$ of the preceding case \eqref{h3d1quan}, while in the star-triangle relation \eqref{qh3d1alt}, the Boltzmann weights $\oW(\theta\,|\,\spn_i,\spn_j)$ and $W(\theta\,|\,\spn_i,\spn_j)$ are exchanged compared to \eqref{qh3d1}; this results in a different three-leg form from \eqref{H3d13leg}, which however corresponds to the same quad equation $H3_{(\delta=1;\,\varepsilon=1)}$ (up to relabelling of variables). 

The star-triangle relation \eqref{qh3d1alt} did not appear before, but up to a change of variables is equivalent to a hyperbolic analogue of the nonterminating Saalsch\"{u}tz formula \cite{BultThesis} (equivalently hyperbolic analogue of Barnes's second lemma).  This is the same hypergeometric equation that corresponds to \eqref{qq3d0} ($Q3_{(\delta=0)}$ case), but with a different change of variables.

Define
\begin{align}
\label{h3d12int}
I_{6,H}(\theta_1,\theta_3\,|\,\spn_1,\spn_2,\spn_3)=\int_{\mathcal{C}}d\spn\,\rho_{6,H}(\theta_1,\theta_3\,|\,\spn_1,\spn_2,\spn_3;\spn)\,,
\end{align}
where
\begin{align}
\label{h3d12rho}
\rho_{6,H}(\theta_1,\theta_3\,|\,\spn_1,\spn_2,\spn_3;\spn)=\oV({\theta_1}\,|\,\spn_1,\spn)\,V({\theta_1+\theta_3}\,|\,\spn_2,\spn)\,\oW({\theta_3}\,|\,\spn,\spn_3)\,,
\end{align}
and the contour $\mathcal{C}$ is a deformation of $\mathbb{R}$, separating the points $\ii(\eta-\theta_1-\theta_3)\pm \spn_2+\ii\bb\mathbb{Z}_{\geq0}+\ii\bb^{-1}\mathbb{Z}_{\geq0}$, $\ii\theta_3+\spn_3+\ii\bb\mathbb{Z}_{\geq0}+\ii\bb^{-1}\mathbb{Z}_{\geq0}$, from the points $-\ii\theta_1\pm \spn_1-\ii\bb\mathbb{Z}_{\geq0}-\ii\bb^{-1}\mathbb{Z}_{\geq0}$, $-\ii\theta_3-\spn_3-\ii\bb\mathbb{Z}_{\geq0}-\ii\bb^{-1}\mathbb{Z}_{\geq0}$.    If
\begin{align}
\label{h3d12rvals}
\re(\theta_1)\pm \im(\spn_1)\,, \re(\theta_3)+ \im(\spn_3)\,, \eta-\re(\theta_1+\theta_3)\pm \im(\spn_2)>0\,,
\end{align}
the contour $\mathcal{C}$ can be chosen to be $\mathbb{R}$.

The star-triangle relation \eqref{qh3d1alt}, may be obtained from the following limit of \eqref{qq3d1}:

\begin{prop} \label{h3d12prop}
For the values \eqref{q3d1rvals},
\begin{align}
\label{h3d12lim}
\begin{split}
\lim_{\kappa\to\infty}\EXP^{4\pi\eta\kappa+\ii\pi(\eta(\eta-2(\theta_1+\ii \spn_3))-\theta_3(\theta_3+2(\theta_1-\ii \spn_3))-\spn_1^2+\spn_2^2)}I_{14}(\theta_1+\ii\kappa,\theta_3\,|\,\spn_1,\spn_2,\spn_3+\kappa)\phantom{\,,} \\
=I_{6,H}(\theta_1,\theta_3\,|\,\spn_1,\spn_2,\spn_3)\,,
\end{split}
\end{align}
where $I_{14}(\theta_1,\theta_3\,|\,\spn_1,\spn_2,\spn_3)$ is defined in \eqref{q3d1int}.

\end{prop}

The steps of the proof of Proposition \ref{h3d12prop} are analogous to the respective steps for Proposition \ref{q3d0prop}.  The star-triangle relation \eqref{qh3d1alt} follows from Proposition \ref{h3d12prop}, by using \eqref{HGFinflim} to take the same limit of the right hand side of \eqref{qq3d1}.

\tocless\subsubsection{Classical integrable equations}

Let the classical variables $x_1,x_2,x_3$, $\alpha_1,\alpha_3$, take values in \eqref{hypcvals}.  The Lagrangian functions for this case are defined by
\begin{align}
\label{lagh3d1alt}
\begin{split}
\Lambda(\alpha\,|\,x_i,x_j)&=\ds\lie(-\EXP^{x_i+x_j+\ii\alpha})+\lie(-\EXP^{x_j-x_i+\ii\alpha})+\frac{(x_j+\ii\alpha)^2}{2}\,, \\
\ol(\alpha\,|\,x_i,x_j)&=\ds\olie(\EXP^{x_i-x_j-\ii\alpha})+\olie(\EXP^{x_j-x_i-\ii\alpha})-2\,\olie(\EXP^{-\ii\alpha})+\frac{(x_i-x_j)^2}{2}\,.
\end{split}
\end{align}
The $\ol(\alpha\,|\,x_i,x_j)$ is equivalent to $\mathcal{L}(\pi-\alpha\,|\,x_i,x_j)$ from \eqref{lagq3d0} ($Q3_{(\delta=0)}$ case) , and thus satisfies
\begin{align}
\ds\ol(\alpha\,|\,x_i,x_j)=\ol(\alpha\,|\,x_j,x_i)\,,\qquad\ol(\alpha\,|\,x_i,x_j)=-\ol(-\alpha\,|\,x_i,x_j)\,.
\end{align}
The $\Lambda(\alpha\,|\,x_i,x_j)$ does not satisfy the above symmetries, and particularly the ordering of $x_i,x_j$ in $\Lambda(\alpha\,|\,x_i,x_j)$ needs to be taken into account.

Using the asymptotics \eqref{HGFqcl}, the leading order $O(\hbar^{-1})$ quasi-classical expansion \eqref{hypqcl} of the integrand \eqref{h3d12rho} is
\begin{align}
\label{h3d12qcl}
\begin{split}
\frac{\alpha_3^2-2\pi\alpha_3-4\lie(\EXP^{-\ii\alpha_3})-x_1^2-x_2^2}{2\ii\hbar}+\Log \rho_{6,H}(\tfrac{\alpha_1}{\sqrt{2\pi\hbar}},\tfrac{\alpha_3}{\sqrt{2\pi\hbar}}\,|\,\tfrac{x_1}{\sqrt{2\pi\hbar}},\tfrac{x_2}{\sqrt{2\pi\hbar}}\tfrac{x_3}{\sqrt{2\pi\hbar}};\tfrac{x_0}{\sqrt{2\pi\hbar}})\phantom{\,,} \\[0.1cm]
=(\ii\hbar)^{-1}\astr{x_0}+O(1)\,,
\end{split}
\end{align}
where
\begin{align}
\astr{x_0}=\olam(\alpha_1\,|\,x_1,x_0)+\Lambda(\alpha_1+\alpha_3\,|\,x_2,x_0)+\ol(\alpha_3\,|\,x_0,x_3)\,,
\end{align}
and
\begin{align}
\olam(\alpha\,|\,x_i,x_j)=\Lambda(\pi-\alpha\,|\,x_i,-x_j)\,.
\end{align}

The saddle point three-leg equation \eqref{3legasym} is then given by
\begin{align}
\label{H3d13legalt}
\left.\frac{\partial \astr{x}}{\partial x}\right|_{x=x_0}\hspace{-0.3cm}=\ovphib(\alpha_1\,|\,x_1,x_0)+\phi(\alpha_1+\alpha_3\,|\,x_2,x_0)+\vphi(\alpha_3\,|\,x_3,x_0)=0\,,
\end{align}
where $\phi(\alpha\,|\,x_i,x_j)$, and $\vphi(\alpha\,|\,x_i,x_j)$, are defined by
\begin{align}
\begin{split}
\phi(\alpha\,|\,x_i,x_j)&=-\Log(1+\EXP^{x_i+x_j+\ii\alpha})-\Log(1+\EXP^{x_j-x_i+\ii\alpha})+2x_j+\ii\alpha\,, \\
\vphi(\alpha\,|\,x_i,x_j)&=-\Log(1-\EXP^{x_j-x_i-\ii\alpha})+\Log(1-\EXP^{x_i-x_j-\ii\alpha})-x_i-x_j\,,
\end{split}
\end{align}
and
\begin{align}
\ovphib(\alpha\,|\,x_i,x_j)=-\phi(\pi-\alpha\,|\,x_i,-x_j)\,.
\end{align}

Equation \eqref{H3d13legalt} is a three-leg form of $H3_{(\delta=1;\,\varepsilon=1)}$ (different from \eqref{H3d13leg}), arising as the equation for the saddle point of the star-triangle relation \eqref{qh3d1alt} in the limit \eqref{hypqcl}.  With the following change of variables
\begin{align}
u=\cosh(x_1),\quad y=\cosh(x_2),\quad x=-\EXP^{-x_0},\quad v=\EXP^{x_3},\quad\alpha=\EXP^{-\ii\alpha_1},\quad\beta=\EXP^{-\ii(\alpha_1+\alpha_3)},
\end{align}
the exponential of the three-leg equation \eqref{H3d13legalt} may be written in the form
\begin{align}
H(x,u,y,v;\alpha,\beta)=0\,,
\end{align}
where 
\begin{align}
\label{H3d12}
H(x,u,y,v;\alpha,\beta)=2(uv+xy)\beta - 2(yv+xu)\alpha + xv(\beta^2-\alpha^2)(\alpha\beta)^{-1} - (\alpha^2-\beta^2)\,.
\end{align}
This is identified as $H3_{(\delta=1;\,\varepsilon=1)}$ \cite{ABS2}, an affine-linear quad equation that satisfies the 3D-consistency condition exactly as defined in Section \ref{sec:overview}.

\subsection{\texorpdfstring{$H3_{(\delta=0,1;\,\varepsilon=1-\delta)}$}{H3(delta=0,1,epsilon=1-delta)} case}

\tocless\subsubsection{Star-triangle relation}

Let the variables $\spn_1,\spn_2,\spn_3$, and $\theta_1,\theta_3$, take values \eqref{hypvals}.  The star-triangle relation is given in this case by
\begin{align}
\label{qh3d0}
\begin{split}
\ds\int_{\mathbb{R}}d\spn_0\,S(\spn_0)\,\oV(\theta_1\,|\,\spn_1,\spn_0)\,V({\theta_1+\theta_3}\,|\,\spn_2,\spn_0)\,W(\eta-\theta_3\,|\,\spn_0,\spn_3)\phantom{\,,} \\
\ds=R(\theta_1,\theta_3)\,V({\theta_1}\,|\,\spn_2,\spn_3)\,\oV(\theta_1+\theta_3\,|\,\spn_1,\spn_3)\,W({\theta_3}\,|\,\spn_2,\spn_1)\,,
\end{split}
\end{align}
where the Boltzmann weights are
\begin{align}
\label{h3d0quan}
\begin{split}
V(\theta\,|\,\spn_i,\spn_j)&=\ds\EXP^{\overline{B}(\ii\theta+\spn_i-\spn_j)+\overline{B}(\ii\theta+\spn_j-\spn_i)+\overline{B}(\ii\theta-\spn_i-\spn_j)}\Gamma_h(\spn_i+\spn_j+\ii\theta;\bb)\,, \\[0.1cm]
\oV(\theta\,|\,\spn_i,\spn_j)&=\ds \frac{1}{V(-(\eta-\theta)\,|\,\spn_i,\spn_j)}\,, \\[0.1cm]
W(\theta\,|\,\spn_i,\spn_j)&=\ds\EXP^{-2\pi\theta(\spn_i+\spn_j)}\frac{\Gamma_h(\spn_i-\spn_j+\ii\theta;\bb)}{\Gamma_h(\spn_i-\spn_j-\ii\theta;\bb)}\,,
\end{split}
\end{align}
$\eta$ is defined in \eqref{hypeta},
\begin{align}
\label{obdef}
\overline{B}(z):=\frac{\pi\ii}{2}B_{2,2}(\ii z+\eta;\bb,\bb^{-1})\,,
\end{align}
where $B_{2,2}(z;\omega_1,\omega_2)$ is defined in \eqref{b22def}, and
\begin{align}
\label{h3d0quans}
\begin{split}
S(\spn)&=\ds\EXP^{4\pi \eta \spn}\,, \\
R(\theta_1,\theta_3)&=\ds\EXP^{\overline{B}(\ii(\eta-2\theta_1))-\overline{B}(\ii(\eta-2(\theta_1+\theta_3)))}\Gamma_h(\ii(\eta-2\theta_3);\bb)\,.
\end{split}
\end{align}
In \eqref{h3d0quan}, the Boltzmann weight $W(\theta\,|\,\spn_i,\spn_j)$ is the same as for \eqref{fadvol} ($Q3_{(\delta=0)}$ case), up to the additional exponential factor.

From \eqref{HGFinflim}, the asymptotics of the integrand of \eqref{qh3d0} are
\begin{align}
\label{h3d0xinf}
\begin{split}
O(\EXP^{4\pi\eta \spn_0})\,,&\qquad \spn_0\to-\infty\,, \\
O(\EXP^{-2\pi(\eta-\theta_3+\ii(\spn_1-\spn_2)) \spn_0})\,,&\qquad \spn_0\to+\infty\,,
\end{split}
\end{align}
and particularly the integral in \eqref{qh3d0} is absolutely convergent for the values \eqref{hypvals}.  

Note that the Boltzmann weights in \eqref{h3d0quan} have additional factors of $\overline{B}(z)$ defined in \eqref{obdef}, whereas in the previous cases it always happened that these additional exponential terms (which come from the asymptotic relation \eqref{HGFinflim}) completely cancelled out of the star-triangle relation.

The star-triangle relation \eqref{qh3d0} previously appeared in the Appendix of \cite{BAZHANOV2018509}, as a limiting case of \eqref{qq3d0}.  It seems that the star-triangle relation \eqref{qh3d0} did not appear before in the hypergeometric integral theory, but it may be interpreted as a hyperbolic analogue of an integral formula of Barnes, known as Barnes's first lemma \cite{Barnes1908} (for the star-triangle relation corresponding to the latter, see \eqref{qh2d0} for the $H2_{(\varepsilon=0)}$ case).

Define
\begin{align}
\label{h3d0int}
I_{4}(\theta_1,\theta_3\,|\,\spn_1,\spn_2,\spn_3)=\int_{\mathcal{C}}d\spn\,\rho_4(\theta_1,\theta_3\,|\,\spn_1,\spn_2,\spn_3;\spn)\,,
\end{align}
where
\begin{align}
\label{h3d0rho}
\rho_4(\theta_1,\theta_3\,|\,\spn_1,\spn_2,\spn_3;\spn)=S(\spn)\,\oV({\theta_1}\,|\,\spn_1,\spn)\,V({\theta_1+\theta_3}\,|\,\spn_2,\spn)\,\oW({\theta_3}\,|\,\spn,\spn_3)\,,
\end{align}
and the contour $\mathcal{C}$ is a deformation of $\mathbb{R}$, separating the points $\ii\theta_3+\spn_3+\ii\bb\mathbb{Z}_{\geq0}+\ii\bb^{-1}\mathbb{Z}_{\geq0}$, $\ii(\eta-\theta_1-\theta_3)-\spn_2+\ii\bb\mathbb{Z}_{\geq0}+\ii\bb^{-1}\mathbb{Z}_{\geq0}$, from the points $-\ii\theta_3+\spn_3-\ii\bb\mathbb{Z}_{\geq0}-\ii\bb^{-1}\mathbb{Z}_{\geq0}$, $-\ii\theta_1-\spn_1-\ii\bb\mathbb{Z}_{\geq0}-\ii\bb^{-1}\mathbb{Z}_{\geq0}$.      If
\begin{align}
\label{h3d0rvals}
\re(\theta_1)+ \im(\spn_1)\,, \re(\theta_3)\pm \im(\spn_3)\,, \eta-\re(\theta_1+\theta_3)- \im(\spn_2)>0\,,
\end{align}
the contour $\mathcal{C}$ can be chosen to be $\mathbb{R}$.	

The star-triangle relation \eqref{qh3d0}, may be obtained from the following limit of \eqref{qq3d1}:

\begin{prop} \label{h3d0prop}
For the values \eqref{q3d1rvals},
\begin{align}
\label{h3d0lim}
\lim_{\kappa\to\infty}\EXP^{4\pi(2\eta+\theta_3-\ii(\spn_1-\spn_2))\kappa}I_{14}(\theta_1+2\ii\kappa,\theta_3\,|\,\spn_1+\kappa,\spn_2+\kappa,\spn_3+\kappa)=I_{4}(\theta_1,\theta_3\,|\,\spn_1,\spn_2,\spn_3)\,,
\end{align}
where $I_{14}(\theta_1,\theta_3\,|\,\spn_1,\spn_2,\spn_3)$ is defined in \eqref{q3d1int}.

\end{prop}

The steps of the proof of Proposition \ref{h3d0prop}, are analogous to the respective steps for Proposition \ref{q3d0prop}.  The star-triangle relation \eqref{qh3d0} follows from Proposition \ref{h3d0prop}, by using \eqref{HGFinflim} to take the same limit of the right hand side of \eqref{qq3d1}.

\tocless\subsubsection{Classical integrable equations}

Let the classical variables $x_1,x_2,x_3$, $\alpha_1,\alpha_3$, take values in \eqref{hypcvals}.  The Lagrangian functions for this case are defined by
\begin{align}
\label{lagh3d0}
\begin{split}
\Lambda(\alpha\,|\,x_i,x_j)&=\ds \lie(-\EXP^{x_i+x_j+\ii\alpha})+\ii(x_i+x_j)\alpha+x_i^2+x_j^2-\frac{\alpha^2}{2}+\frac{\pi^2}{6}\,, \\[0.1cm]
\lag(\alpha\,|\,x_i,x_j)&=\ds\lie(-\EXP^{x_i-x_j+\ii\alpha})+\lie(-\EXP^{x_j-x_i+\ii\alpha})-2\,\lie(-\EXP^{\ii\alpha})+\frac{(x_i-x_j)^2}{2}\,. 
\end{split}
\end{align}
These functions are all symmetric in $x_i,x_j$, satisfying
\begin{align}
\Lambda(\alpha\,|\,x_i,x_j)=\Lambda(\alpha\,|\,x_j,x_i)\,,\quad\olam(\alpha\,|\,x_i,x_j)=\olam(\alpha\,|\,x_j,x_i)\,,\quad\lag(\alpha\,|\,x_i,x_j)=\lag(\alpha\,|\,x_j,x_i)\,,
\end{align}
while $\lag(\alpha\,|\,x_i,x_j)$ (related to \eqref{lagq3d0}) also satisfies
\begin{align}
\lag(\alpha\,|\,x_i,x_j)=-\lag(-\alpha\,|\,x_i,x_j)\,.
\end{align}

Using the asymptotics \eqref{HGFqcl}, the leading order $O(\hbar^{-1})$ quasi-classical expansion \eqref{hypqcl} of the integrand \eqref{h3d0rho} is
\begin{align}
\label{h3d0qcl}
\begin{split}
\frac{\ii x_1(\pi-\alpha_1)+\ii x_2(\alpha_1+\alpha_3)+\ii(x_3+\ii\alpha_3)(\pi-\alpha_3)+\alpha_1(\pi+\alpha_3)-2\lie(\EXP^{-\ii\alpha_3})-\frac{\pi^2}{6}}{\ii\hbar}+\phantom{\,,} \\[0.1cm]
\Log \rho_4(\tfrac{\alpha_1}{\sqrt{2\pi\hbar}},\tfrac{\alpha_3}{\sqrt{2\pi\hbar}}\,|\,\tfrac{x_1}{\sqrt{2\pi\hbar}},\tfrac{x_2}{\sqrt{2\pi\hbar}}\tfrac{x_3}{\sqrt{2\pi\hbar}};\tfrac{x_0}{\sqrt{2\pi\hbar}})=(\ii\hbar)^{-1}\astr{x_0}+O(1)\,,
\end{split}
\end{align}
where
\begin{align}
\astr{x_0}=\olam(\alpha_1\,|\,x_1,x_0)+\Lambda(\alpha_1+\alpha_3\,|\,x_2,x_0)+\ol(\alpha_3\,|\,x_3,x_0)\,,
\end{align}
and
\begin{align}
\olam(\alpha\,|\,x_i,x_j)=-\Lambda(-(\pi-\alpha)\,|\,x_i,x_j)\,,\qquad\ol(\alpha\,|\,x_i,x_j)=\lag(\pi-\alpha\,|\,x_i,x_j)\,.
\end{align}

The saddle point three-leg equation \eqref{3legasym} is then given by
\begin{align}
\label{H3d03leg}
\left.\frac{\partial \astr{x}}{\partial x}\right|_{x=x_0}\hspace{-0.3cm}=\ovphib(\alpha_1\,|\,x_1,x_0)+\phi(\alpha_1+\alpha_3\,|\,x_2,x_0)+\vphi(\alpha_3\,|\,x_3,x_0)=0\,,
\end{align}
where $\phi(\alpha\,|\,x_i,x_j)$, and $\vphi(\alpha\,|\,x_i,x_j)$, are defined by
\begin{align}
\begin{split}
\phi(\alpha\,|\,x_i,x_j)&=-\Log(1+\EXP^{x_i+x_j+\ii\alpha})+x_j\,, \\
\vphi(\alpha\,|\,x_i,x_j)&=-\Log(1-\EXP^{x_j-x_i-\ii\alpha})+\Log(1-\EXP^{x_i-x_j-\ii\alpha})-x_i+\ii\alpha\,,
\end{split}
\end{align}
and
\begin{align}
\ovphib(\alpha\,|\,x_i,x_j)=-\phi(\pi-\alpha\,|\,-x_i,-x_j)+x_i+\ii\alpha\,.
\end{align}

Equation \eqref{H3d03leg} represents a three-leg form for either $H3_{(\delta=1;\,\varepsilon=0)}$, or $H3_{(\delta=0;\,\varepsilon=1)}$, arising as the equation for the saddle point of the star-triangle relation \eqref{qh3d0} in the limit \eqref{hypqcl}.  With the following change of variables
\begin{align}
x=-\EXP^{x_0},\qquad u=\EXP^{x_1},\qquad y=\EXP^{x_2},\qquad v=\EXP^{x_3},\qquad\alpha=\EXP^{-\ii\alpha_1},\qquad\beta=-\EXP^{-\ii(\alpha_1+\alpha_3)},
\end{align}
the exponential of the three-leg equation \eqref{H3d03leg} may be written in the form
\begin{align}
\label{qref}
H(x,u,y,v;\alpha,\beta)=0\,,
\end{align}
where 
\begin{align}
\label{H3d01}
H(x,u,y,v;\alpha,\beta)=(uv+xy)\beta-(yv+xu)\alpha-(\alpha^2-\beta^2)\,.
\end{align}
Similarly, with the following change of variables
\begin{align}
x=-\EXP^{x_0},\qquad u=\EXP^{-x_1},\qquad y=\EXP^{-x_2},\qquad v=\EXP^{x_3},\qquad\alpha=\EXP^{\ii\alpha_1},\qquad\beta=-\EXP^{\ii(\alpha_1+\alpha_3)},
\end{align}
the exponential of the three-leg equation \eqref{H3d03leg} may be written in the form \eqref{qref}, where
\begin{align}
\label{H3d01b}
H(x,u,y,v;\alpha,\beta)=(uv+xy)\beta-(yv+xu)\alpha+uy(\beta^2-\alpha^2)(\alpha\beta)^{-1}\,.
\end{align}
The above equations \eqref{qref}, with \eqref{H3d01}, and \eqref{H3d01b}, are respectively identified as $H3_{(\delta=1;\,\varepsilon=0)}$, and $H3_{(\delta=1;\,\varepsilon=0)}$ \cite{ABS2}, which are affine-linear quad equations that satisfy the 3D-consistency condition exactly as defined in Section \ref{sec:overview}.  Thus the same three-leg saddle point equation \eqref{H3d03leg}, corresponds to both of the respective cases of $H3_{(\delta=1;\,\varepsilon=0)}$, and $H3_{(\delta=0;\,\varepsilon=1)}$.  This can be expected, since equation \eqref{qref} with one of \eqref{H3d01}, \eqref{H3d01b}, is straightforwardly transformed into the other equation, and vice versa, simply by setting $u\to u^{-1}$, $y\to y^{-1}$, $\alpha\to\alpha^{-1}$, $\beta\to\beta^{-1}$.

\subsection{\texorpdfstring{$H3_{(\delta=0;\,\varepsilon=0)}$}{H3(delta=0,epsilon=0)} case}

\tocless\subsubsection{Star-triangle relation}

Let the variables $\spn_1,\spn_2,\spn_3$, and $\theta_1,\theta_3$, take values \eqref{hypvals}. The star-triangle relation is given in this case by
\begin{align}
\label{qh3d02}
\begin{split}
\ds\int_{\mathbb{R}}d\spn_0\,\oV({\theta_1}\,|\,\spn_1,\spn_0)\,V({\theta_1+\theta_3}\,|\,\spn_2,\spn_0)\,W({\eta-\theta_3}\,|\,\spn_3,\spn_0)\phantom{\,,} \\
\ds=R(\theta_3)\,V({\theta_1}\,|\,\spn_2,\spn_3)\,\oV(\theta_1+\theta_3\,|\,\spn_1,\spn_3)\,W({\theta_3}\,|\,\spn_2,\spn_1)\,,
\end{split}
\end{align}
where the Boltzmann weights are
\begin{align}
\label{h3d02quan}
\begin{split}
V(\theta\,|\,\spn_i,\spn_j)&=\ds\EXP^{2\pi\ii \spn_i\spn_j}, \\[0.1cm]
\oV(\theta\,|\,\spn_i,\spn_j)&=\ds \frac{1}{V(\theta\,|\,\spn_i,\spn_j)}\,, \\[0.1cm]
W(\theta\,|\,\spn_i,\spn_j)&=\ds\frac{\Gamma_h(\spn_i-\spn_j+\ii\theta;\bb)}{\Gamma_h(\spn_i-\spn_j-\ii\theta;\bb)}\,,
\end{split}
\end{align}
$\eta$ is defined in \eqref{hypeta}, and
\begin{align}
\begin{split}
R(\theta_3)&=\ds\Gamma_h(\ii(\eta-2\theta_3);\bb)\,,
\end{split}
\end{align}

From \eqref{HGFinflim}, the asymptotics of the integrand of \eqref{qh3d02} are
\begin{align}
\label{h3d02xinf}
O(\EXP^{-2\pi(\eta-\theta_3)\, \left|\spn_0\right|-2\pi\ii(\spn_1-\spn_2)\spn_0})\,,&\qquad \spn_0\to\pm\infty\,,
\end{align}
and particularly the integral in \eqref{qh3d02} is absolutely convergent for the values \eqref{hypvals}. 

In \eqref{h3d02quan}, the Boltzmann weight $W(\theta\,|\,\spn_i,\spn_j)$ is the same as the Boltzmann weight for \eqref{fadvol} ($Q3_{(\delta=0)}$ case).  The Boltzmann weight $V(\theta\,|\,\spn_i,\spn_j)$ clearly has no $\theta$ dependency, thus \eqref{qh3d02} is independent of $\theta_1$. The Boltzmann weight $V(\theta\,|\,\spn_i,\spn_j)$ is also obviously symmetric upon the exchange $\spn_i\leftrightarrow \spn_j$.

The star-triangle relation \eqref{qh3d02} has not appeared before, however it is interestingly equivalent (up to a change of variables) to a ``self-duality'' property for the Boltzmann weight $W(\theta\,|\,\spn_i,\spn_j)$ of the Faddeev-Volkov model \cite{Bazhanov:2007mh,Bazhanov:2007vg}.  It also appears not to have previously been considered in the hypergeometric integral theory, however it can be considered to be a hyperbolic analogue of a special case of Barnes's integral formula for the hypergeometric function $\!\!~_2F_1$ \cite{Barnes1908} (the latter is equivalent to the alternate form of the star-triangle relation corresponding to $H1_{(\varepsilon=1)}$ in \eqref{qh1e0a}).  

Define
\begin{align}
\label{h3d02int}
I_{2}(\theta_1,\theta_3\,|\,\spn_1,\spn_2,\spn_3)=\int_{\mathcal{C}}d\spn\,\rho_2(\theta_1,\theta_3\,|\,\spn_1,\spn_2,\spn_3;\spn)\,,
\end{align}
where
\begin{align}
\label{h3d02rho}
\rho_2(\theta_1,\theta_3\,|\,\spn_1,\spn_2,\spn_3;\spn)=\oV({\theta_1}\,|\,\spn_1,\spn)\,V({\theta_1+\theta_3}\,|\,\spn_2,\spn)\,\oW({\theta_3}\,|\,\spn,\spn_3)\,,
\end{align}
and the contour $\mathcal{C}$ is a deformation of $\mathbb{R}$, which separates the points $\ii\theta_3+\spn_3+\ii\bb\mathbb{Z}_{\geq0}+\ii\bb^{-1}\mathbb{Z}_{\geq0}$, from the points $-\ii\theta_3+\spn_3-\ii\bb\mathbb{Z}_{\geq0}-\ii\bb^{-1}\mathbb{Z}_{\geq0}$.     If
\begin{align}
\label{h3d02rvals}
\re(\theta_3)\pm \im(\spn_3)>0\,,
\end{align}
the contour $\mathcal{C}$ can be chosen to be $\mathbb{R}$.	

The star-triangle relation \eqref{qh3d02}, may be obtained from the following limit of \eqref{qq3d1}:

\begin{prop} \label{h3d02prop}
For the values \eqref{q3d1rvals},
\begin{align}
\label{h3d02lim}
\begin{split}
\lim_{\kappa\to\infty}\left(\EXP^{2\pi(3\eta+\theta_3-\ii(\spn_1-\spn_2))\kappa+\ii\pi(\eta(\eta-2(\theta_1+\ii(\spn_1+\spn_3)))-\theta_3(\theta_3+2\theta_1)-\spn_1^2+\spn_2^2+\ii\theta_1(\spn_1-\spn_2)+\ii\theta_3(\spn_3-\spn_2))}\right.\phantom{\,,} \\
\left.\phantom{\EXP^{2}}\times I_{14}(\theta_1+\ii\kappa,\theta_3\,|\,\spn_1+\kappa,\spn_2+\kappa,\spn_3+\kappa)\right) =I_{2}(\theta_1,\theta_3\,|\,\spn_1,\spn_2,\spn_3)\,,
\end{split}
\end{align}
where $I_{14}(\theta_1,\theta_3\,|\,\spn_1,\spn_2,\spn_3)$ is defined in \eqref{q3d1int}.

\end{prop}

The steps of the proof of Proposition \ref{h3d02prop}, are analogous to the respective steps for Proposition \ref{q3d0prop}.  The star-triangle relation \eqref{qh3d02} follows from Proposition \ref{h3d02prop}, by using \eqref{HGFinflim} to take the same limit of the right hand side of \eqref{qq3d1}.

\tocless\subsubsection{Classical integrable equations}

Let the classical variables $x_1,x_2,x_3$, $\alpha_1,\alpha_3$, take values in \eqref{hypcvals}.  The Lagrangian functions for this case are defined by
\begin{align}
\label{lagh3d02}
\begin{split}
\Lambda(\alpha\,|\,x_i,x_j)&=\ds-x_ix_j\,, \\[0.1cm] 
\lag(\alpha\,|\,x_i,x_j)&=\ds \lie(-\EXP^{x_i-x_j+\ii\alpha})+\lie(-\EXP^{x_j-x_i+\ii\alpha})-2\lie(-\EXP^{\ii\alpha})+\frac{(x_i-x_j)^2}{2}\,.
\end{split}
\end{align}
These functions are each symmetric in $x_i,x_j$, satisfying
\begin{align}
\lag(\alpha\,|\,x_i,x_j)=\mathcal{L}(\alpha\,|\,x_j,x_i)\,,\qquad\Lambda(\alpha\,|\,x_i,x_j)=\Lambda(\alpha\,|\,x_j,x_i)\,,
\end{align}
while $\lag(\alpha\,|\,x_i,x_j)$ (equivalent to \eqref{lagq3d0}) also satisfies
\begin{align}
\lag(\alpha\,|\,x_i,x_j)=-\lag(-\alpha\,|\,x_i,x_j)\,.
\end{align}

Using the asymptotics \eqref{HGFqcl}, the leading order $O(\hbar^{-1})$ quasi-classical expansion \eqref{hypqcl} of the integrand \eqref{h3d02rho} is
\begin{align}
\label{h3d02qcl}
\begin{split}
\frac{(\pi-\alpha_3)^2-\frac{\pi^2}{3}-4\lie(\EXP^{-\ii\alpha_3})}{2\ii\hbar}+\Log \rho_2(\tfrac{\alpha_1}{\sqrt{2\pi\hbar}},\tfrac{\alpha_3}{\sqrt{2\pi\hbar}}\,|\,\tfrac{x_1}{\sqrt{2\pi\hbar}},\tfrac{x_2}{\sqrt{2\pi\hbar}}\tfrac{x_3}{\sqrt{2\pi\hbar}};\tfrac{x_0}{\sqrt{2\pi\hbar}})\phantom{\,,} \\[0.1cm]
=(\ii\hbar)^{-1}\astr{x_0}+O(1)\,,
\end{split}
\end{align}
where
\begin{align}
\astr{x_0}=\olam(\alpha_1\,|\,x_1,x_0)+\Lambda(\alpha_1+\alpha_3\,|\,x_2,x_0)+\ol(\alpha_3\,|\,x_0,x_3)\,,
\end{align}
and
\begin{align}
\olam(\alpha\,|\,x_i,x_j)=-\Lambda(\alpha\,|\,x_i,x_j)\,,\qquad\ol(\alpha\,|\,x_i,x_j)&=\ds\lag(\pi-\alpha\,|\,x_i,x_j)\,.
\end{align}

The saddle point three-leg equation \eqref{3legasym} is then given by
\begin{align}
\label{H3d023leg}
\left.\frac{\partial \astr{x}}{\partial x}\right|_{x=x_0}\hspace{-0.3cm}=\ovphib(\alpha_1\,|\,x_1,x_0)+\phi(\alpha_1+\alpha_3\,|\,x_2,x_0)+\vphi(\alpha_3\,|\,x_3,x_0)=0\,,
\end{align}
where $\phi(\alpha\,|\,x_i,x_j)$, and $\vphi(\alpha\,|\,x_i,x_j)$, are defined by
\begin{align}
\phi(\alpha\,|\,x_i,x_j)=-x_i\,,\quad\vphi(\alpha\,|\,x_i,x_j)=\Log(1-\EXP^{x_i-x_j-\ii\alpha})-\Log(1-\EXP^{x_j-x_i-\ii\alpha})-x_i+x_j\,,
\end{align}
and
\begin{align}
\ovphib(\alpha\,|\,x_i,x_j)=-\phi(\alpha\,|\,x_i,x_j)\,.
\end{align}

Equation \eqref{H3d023leg} is a three-leg form of $H3_{(\delta=0;\,\varepsilon=0)}$, arising as the equation for the saddle point of the star-triangle relation \eqref{qh3d02} in the limit \eqref{hypqcl}.  With the following change of variables
\begin{align}
x=-\EXP^{x_0},\qquad u=\EXP^{x_1},\qquad y=\EXP^{x_2},\qquad v=\EXP^{x_3},\qquad\alpha=\EXP^{-\ii\alpha_1}\,,\qquad\beta=-\EXP^{-\ii(\alpha_1+\alpha_3)},
\end{align}
the exponential of the three-leg equation \eqref{H3d023leg} may be written in the form
\begin{align}
H(x,u,y,v;\alpha,\beta)=0\,,
\end{align}
where 
\begin{align}
\label{H3d02}
H(x,u,y,v;\alpha,\beta)=(uv+xy)\beta-(yv+xu)\alpha\,.
\end{align}
This is identified as $H3_{(\delta=0;\,\varepsilon=0)}$ \cite{ABS2}, an affine-linear quad equation that satisfies the 3D-consistency condition exactly as defined in Section \ref{sec:overview}.

\section{Rational cases}
\label{sec:ratlim}

The rational limit of each solution  of the star-triangle relation of the previous section respectively (except for the $H3_{(\delta=0,\varepsilon=0)}$ case of \eqref{qh3d02}), results in different solutions of the star-triangle relations \eqref{YBEsym}, \eqref{YBEasym}, which have Boltzmann weights given in terms of products of the Euler gamma function $\Gamma(z)$.  Notably the Boltzmann weights of the $Q$-type equations obtained in the rational limit, will no longer satisfy the spin reflection symmetry, or have a physical regime where the Boltzmann weights take positive values.  This is a side-effect of needing to break the symmetry of the integrand in \eqref{YBEsym} in order to have a convergent limit.  The quasi-classical expansion of the star-triangle relations obtained at the rational level in this Section, will specifically result in the $Q2$, $Q1_{(\delta=1)}$, $H2_{(\varepsilon=1)}$, $H2_{(\varepsilon=0)}$, classical integrable quad equations (see Figure \ref{ratlimsfig}).

\subsection{Euler Gamma function}

The central function at the rational level is the classical Euler gamma function $\Gamma(z)$, which may be defined by the integral
\begin{align}
\label{eulergamma}
\Gamma(z)=\int^\infty_0dt\,t^{z-1}\EXP^{-t}\,,\qquad\mbox{Re}(z)>0\,.
\end{align}
The gamma function satisifies the well-known functional equation
\begin{align}
\Gamma(1+z)=z\Gamma(z)\,,
\end{align}
through which \eqref{eulergamma} may be analytically continued to $\mathbb{C}-\mathbb{Z}_{\leq0}$.

The gamma function satisfies the following formula for large $z$
\begin{align}
\label{Stirling}
\Log\Gamma(z)=\gamma(-\ii z)-\frac{\Log(z)}{2}-z+\frac{\Log(2\pi)}{2}+O(z^{-1})\,,\qquad|\,z\,|\rightarrow\infty,\quad|\,\mbox{Arg}(z)\,|<\pi\,,
\end{align}
where the function $\gamma(z)$ is defined by
\begin{align}
\label{gammadef}
\gamma(z)=\ii z\Log(\ii z)\,,\qquad |\,\mbox{Arg}(\ii z)\,|<\pi \,.
\end{align} 
Equation \eqref{Stirling} is known as the Stirling formula.

The quasi-classical expansion will specifically involve scaling the spins and spectral parameters appearing in a star-triangle relation \eqref{YBEsym}, \eqref{YBEasym}, by a parameter $\hbar^{-1}>0$, as follows
\begin{align}
\label{ratqcl}
\spn_i = x_i\hbar^{-1},\;i=0,1,2,3,\qquad\theta_j=\alpha_j\hbar^{-1},\;j=1,3\qquad\hbar\rightarrow0\,.
\end{align}
Then as $\hbar$ approaches zero, the quasi-classical expansion of the Boltzmann weights will be obtained through the use of \eqref{Stirling}.  

As $\bb\to0$, the Euler gamma function \eqref{eulergamma} is related to the hyperbolic gamma function \eqref{HGF}, by \cite{Ruijsenaars1999,Rains2009}
\begin{align}
\label{ratasymp}
\Log\Gamma_h(z\bb\pm\ii (2\bb)^{-1};\bb)=(\ii z\mp\tfrac{1}{2})\Log(2(\pi\bb^2)^{-1})\pm\Log(\bb)\pm\Log\Gamma(\tfrac{1}{2}\pm\ii z)+O(\bb)\,,
\end{align}
where $|\mbox{Arg}(\ii z)\,|<\pi$, and either the $+$ or $-$ sign can be taken.  The asymptotics \eqref{ratasymp} will be used to relate the star-triangle relations of this Section, to the star-triangle relations of the hyperbolic cases appearing in Section \ref{sec:hyplim}.  The specific limit will involve the following change of variables in \eqref{YBEsym}, \eqref{YBEasym}
\begin{align}
\label{ratlim}
\spn_i\to \spn_i\bb\,,\; i=0,1,3\,,\qquad \spn_2\to (\spn_2-\ii\pi)\bb+\ii\eta\,,\qquad \theta_j\to\theta_j\bb\,,\; j=1,3\,,
\end{align}
and taking $\bb\to0$.  The resulting limits between the star-triangle relations at the hyperbolic and rational levels, is summarised in the diagram of Figure \ref{ratlimsfig}.  Note also that there is no analogue of the crossing parameter $\eta$ in \eqref{hypeta}, for the rational cases that are obtained from the limit \eqref{ratlim} (essentially the rational limit \eqref{ratlim} involves taking $\eta\to\infty$).

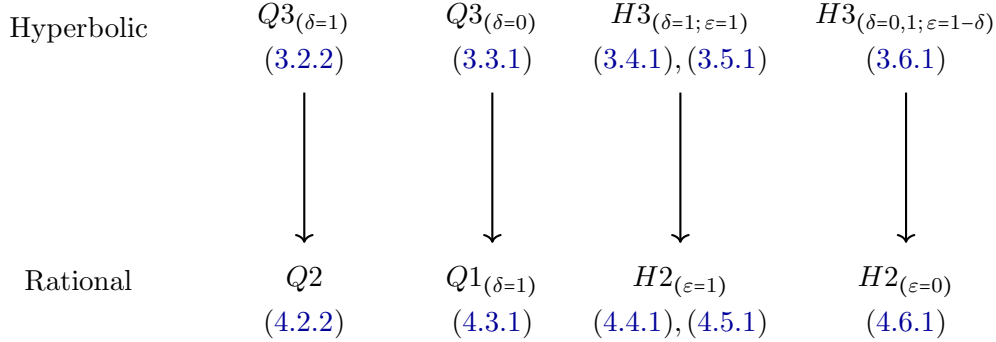
\begin{figure}[tbh]
\centering
\begin{tikzpicture}[scale=1]
\draw[white!] (-7,3.5) circle (0.01pt)
node[above=1pt]{\color{black} Hyperbolic};
\draw[white!] (7,3) circle (0.01pt);
\draw[white!] (-7,0.2) circle (0.01pt)
node[above=1pt]{\color{black} Rational};
\draw[white!] (7,0) circle (0.01pt);

\draw[white!] (-4,3) circle (0.01pt)
node[above=1pt]{\color{black} $\begin{array}{c} Q3_{(\delta=1)} \\[0.1cm] \eqref{qq3d1} \end{array}$};
\draw[white!] (1,3) circle (0.01pt)
node[above=1pt]{\color{black} $\begin{array}{c} H3_{(\delta=1;\,\varepsilon=1)} \\[0.1cm] \eqref{qh3d1},\eqref{qh3d1alt} \end{array}$};
\draw[white!] (4,3) circle (0.01pt)
node[above=1pt]{\color{black} $\begin{array}{c} H3_{(\delta=0,1;\,\varepsilon=1-\delta)} \\[0.1cm] \eqref{qh3d0} \end{array}$};
\draw[white!] (-1.5,3) circle (0.01pt)
node[above=1pt]{\color{black} $\begin{array}{c} Q3_{(\delta=0)} \\[0.1cm] \eqref{qq3d0} \end{array}$};
\draw[thick,->] (1,3)--(1,1);
\draw[thick,->] (4,3)--(4,1);
\draw[thick,->] (-1.5,3)--(-1.5,1);
\draw[thick,->] (-4,3)--(-4,1);
\draw[white!] (-4,-0.5) circle (0.01pt)
node[above=1pt]{\color{black} $\begin{array}{c} Q2 \\[0.1cm] \eqref{qq2} \end{array}$};
\draw[white!] (1,-0.5) circle (0.01pt)
node[above=1pt]{\color{black} $\begin{array}{c} H2_{(\varepsilon=1)} \\[0.1cm] \eqref{qh2d1},\eqref{qh2d1alt} \end{array}$};
\draw[white!] (-1.5,-0.5) circle (0.01pt)
node[above=1pt]{\color{black} $\begin{array}{c} Q1_{(\delta=1)} \\[0.1cm] \eqref{qq1d1} \end{array}$};
\draw[white!] (4,-0.5) circle (0.01pt)
node[above=1pt]{\color{black} $\begin{array}{c} H2_{(\varepsilon=0)} \\[0.1cm] \eqref{qh2d0} \end{array}$};
\end{tikzpicture}
\caption{Star-triangle relations corresponding to respective integrable quad equations $Q2$, $Q1_{(\delta=1)}$, $H2_{(\varepsilon=1)}$, $H2_{(\varepsilon=0)}$, at the rational level, obtained as the limits \eqref{ratlim} of the hyperbolic cases of the previous section.  Note that \eqref{ratlim} is not a convergent limit for the star-triangle relation \eqref{qh3d02} ($H3_{(\delta=0;\,\varepsilon=0)}$ case).}
\label{ratlimsfig}
\end{figure}

\subsection{\texorpdfstring{$Q2$}{Q2} case}

\tocless\subsubsection{Star-triangle relation}

Let the spins $\spn_1,\spn_2,\spn_3$, and spectral parameters $\theta_1,\theta_3$ take values
\begin{align}
\label{ratvals}
\spn_i\in\mathbb{R}\,,\; i=1,2,3\,,\qquad 0<\theta_1,\theta_3,\theta_1+\theta_3<\pi\,,\; j=1,3\,.
\end{align}
The star-triangle relation is given by \cite{Bazhanov:2016ajm}
\begin{align}
\label{qq2}
\begin{split}
\ds\int_{\mathbb{R}}d\spn_0\,\iS(\spn_0)\,\ow(\theta_1\,|\,\spn_1,\spn_0)\,W(\theta_1+\theta_3\,|\,\spn_2,\spn_0)\,\oW(\theta_3\,|\,\spn_0,\spn_3)\phantom{\,,} \\
=R(\theta_1,\theta_3)\,W(\theta_1\,|\,\spn_2,\spn_3)\,\oW(\theta_1+\theta_3\,|\,\spn_1,\spn_3)\,W(\theta_3\,|\,\spn_2,\spn_1)\,.
\end{split}
\end{align}
where the Boltzmann weights are
\begin{align}
\label{q2quan}
\begin{split}
W(\theta\,|\,\spn_i,\spn_j)&=\frac{\Gamma(\pi+\ii(\spn_i+\spn_j)-\theta)\,\Gamma(\pi+\ii(\spn_i-\spn_j)-\theta)}{\Gamma(\pi+\ii(\spn_i+\spn_j)+\theta)\,\Gamma(\pi+\ii(\spn_i-\spn_j)+\theta)}\,, \\[0.1cm]
\ow(\theta\,|\,\spn_i,\spn_j)&=\Gamma(\theta+\ii(\spn_i+\spn_j))\,\Gamma(\theta+\ii(\spn_i-\spn_j))\,\Gamma(\theta-\ii(\spn_i+\spn_j))\,\Gamma(\theta-\ii(\spn_i-\spn_j))\,, 
\end{split}
\end{align}
and
\begin{align}
\label{q2s}
\ds\iS(\spn)=\frac{1}{4\pi\,\Gamma(2\ii \spn)\,\Gamma(-2\ii \spn)}\,,\qquad
R(\theta_1,\theta_3)=\frac{\Gamma(2\theta_1)\,\Gamma(2\theta_3)}{\Gamma(2(\theta_1+\theta_3))}\,.
\end{align}

The Boltzmann weights satisfy the following relations
\begin{align}
W(\theta\,|\,\spn_i,\spn_j)\,W(-\theta\,|\,\spn_i,\spn_j)=1\,,\qquad\oW(\theta\,|\,\spn_i,\spn_j)=\oW(\theta\,|\,\spn_j,\spn_i)\,.
\end{align}
While the Boltzmann weight $\oW(\theta\,|\,\spn_i,\spn_j)$ is real for the values \eqref{ratvals}, the Boltzmann weight $W(\theta\,|\,\spn_i,\spn_j)$ generally has $\im(W(\theta\,|\,\spn_i,\spn_j))\neq0$.

From \eqref{Stirling}, the asymptotics of the integrand of \eqref{qq2} are
\begin{align}
\label{q2xinf}
O(\EXP^{-2\pi\, \left|\spn_0\right|})\,,\qquad \spn_0\to\pm\infty\,,
\end{align}
and particularly the integral in \eqref{qq2} is absolutely convergent for the values \eqref{ratvals}.

Since $\spn_2$ is shifted by $\spn_2+\ii\eta$ in \eqref{ratlim}, the rational limit results in the two different types of Boltzmann weights appearing in \eqref{q2quan}.  The $\pi$ terms that appear in the arguments of $\Gamma(z)$ in $W(\theta\,|\,\spn_i,\spn_j)$, can in fact be replaced by an arbitrary positive number $\zeta>0$, then the star-triangle relation \eqref{qq2} remains satisfied with $0<\theta_1,\theta_3,\theta_1+\theta_3<\zeta$.  This change is equivalent to taking a slightly different form of rational limit \eqref{ratlim}.  Unlike other star-triangle relations \eqref{qq3d1}, \eqref{qq3d0} (and \eqref{qq1d0}) for the $Q$-type equations, the Boltzmann weights $W(\theta\,|\,\spn_i,\spn_j)$ and $\oW(\theta\,|\,\spn_i,\spn_j)$ are not related by the crossing symmetry $\oW(\theta\,|\,\spn_i,\spn_j)=W(\eta-\theta\,|\,\spn_i,\spn_j)$, for a crossing parameter $\eta$.

The star-triangle relation \eqref{qq2} and connection to $Q2$ previously appeared in \cite{Bazhanov:2016ajm}.\footnote{\label{footnotetypo} Note that there is a typo in \cite{Bazhanov:2016ajm} where the $\pi$ terms were omitted from the arguments of $\Gamma(z)$ for the Boltzmann weight $W(\theta\,|\,\spn_i,\spn_j)$ (and similarly for $W(\theta\,|\,\spn_i,\spn_j)$ in \eqref{q1d1quan}).}  Up to a change of variables, the star-triangle relation \eqref{qq2} is equivalent to a gamma function integration formula given by Askey \cite{Askey1989}, which may in turn be interpreted as an integral analogue of the Dougall summation formula.

Define
\begin{align}
\label{q2int}
I_{14,r}(\theta_1,\theta_3\,|\,\spn_1,\spn_2,\spn_3)=\int_{\mathcal{C}}d\spn\rho_{14,r}(\theta_1,\theta_3\,|\,\spn_1,\spn_2,\spn_3;\spn)\,,
\end{align}
where
\begin{align}
\label{q2rho}
\rho_{14,r}(\theta_1,\theta_3\,|\,\spn_1,\spn_2,\spn_3;\spn)\,=S(\spn_0)\,\oW(\theta_1\,|\,\spn_1,\spn_0)\,W(\theta_1+\theta_3\,|\,\spn_2,\spn_0)\,\oW(\theta_3\,|\,\spn_3,\spn_0)\,,
\end{align}
and the contour $\mathcal{C}$ is a deformation of $\mathbb{R}$, that separates the points $\ii\theta_1\pm \spn_1+\ii\mathbb{Z}_{\geq0}$, $\ii\theta_3\pm \spn_3+\ii\mathbb{Z}_{\geq0}$, $\ii(\pi-\theta_1-\theta_3)+\spn_2+\ii\mathbb{Z}_{\geq0}$, from their negatives.  Then the star-triangle relation \eqref{qq2}, may be obtained from the following limit of \eqref{qq3d1} ($Q3_{(\delta=1)}$ case):

\begin{prop} \label{q2prop}
For the values \eqref{q3d1rvals},
\begin{align}
\label{q2lim}
\lim_{\bb\to0^+}&\frac{2(2\pi\bb)^6}{\bb^2}\, I_{14}(\theta_1\bb,\theta_3\bb\,|\,\spn_1\bb,(\spn_2-\ii\pi)\bb+\ii\eta,\spn_3\bb) =I_{14,r}(\theta_1,\theta_3\,|\,\spn_1,\spn_2,\spn_3)\,,
\end{align}
where $I_{14}(\theta_1,\theta_3\,|\,\spn_1,\spn_2,\spn_3)$ is defined in \eqref{q3d1int}.

\end{prop}

\begin{proof}
For the values \eqref{q3d1rvals}, the contour in $I_{14}$ may be chosen to be $\mathbb{R}$.  Following a change of integration variable $\sigma\to\sigma\bb$ in \eqref{q3d1int}, by the asymptotics \eqref{q3d1xinf}, \eqref{q2xinf}, and the asymptotic formula \eqref{ratasymp}, the combination of the factors and the integrand on the left hand side of \eqref{q2lim} is uniformly bounded on $\mathbb{R}$, and the result follows by dominated convergence (by analytic continuation the domain of values may be extended to the values that are permitted by the contour, as given by the conditions below \eqref{q2rho}).
\end{proof}

The star-triangle relation \eqref{qq2} follows from Proposition \ref{q2prop}, by using \eqref{ratasymp} to take the same limit of the right hand side of \eqref{qq3d1}.

\tocless\subsubsection{Classical integrable equations}

The classical variables $x_1,x_2,x_3$, $\alpha_1,\alpha_3$, take the values
\begin{align}
\label{ratcvals}
x_i\in\mathbb{R}\,,\qquad 0<\alpha_1,\alpha_3\,.
\end{align}
The Lagrangian functions for this case are defined in terms of \eqref{gammadef} as
\begin{align}
\label{lagq2}
\begin{split}
\lag(\alpha\,|\,x_i,x_j)&=\gamma(x_i+x_j+\ii\alpha)+\gamma(x_i-x_j+\ii\alpha)-\gamma(x_i+x_j-\ii\alpha)-\gamma(x_i-x_j-\ii\alpha)\,, \\
\ol(\alpha\,|\,x_i,x_j)&=\gamma(x_i+x_j-\ii\alpha)+\gamma(x_i-x_j-\ii\alpha)+\gamma(-x_i+x_j-\ii\alpha)+\gamma(-x_i-x_j-\ii\alpha)-\gamma(-2\ii\alpha) \,.
\end{split}
\end{align}
These functions satisfy
\begin{align}
\lag(\alpha\,|\,x_i,x_j)=-\lag(-\alpha\,|\,x_i,x_j)\,,\qquad\ol(\alpha\,|\,x_i,x_j)=\ol(\alpha\,|\,x_j,x_i)\,.
\end{align}

Using the asymptotics \eqref{Stirling}, the leading order $O(\hbar^{-1})$ quasi-classical expansion \eqref{ratqcl} of the integrand \eqref{q2rho} is
\begin{align}
\label{q2qcl}
\Log \left(\hbar^{-3}\rho_{14,r}(\alpha_1\hbar^{-1},\alpha_3\hbar^{-1}\,|\,x_1\hbar^{-1},x_2\hbar^{-1},x_3\hbar^{-1};x_0\hbar^{-1})\right)=\hbar^{-1}\astr{x_0}+O(1)\,,
\end{align}
where
\begin{align}
\astr{x_0}=\iC(x_0)+\ol(\alpha_1\,|\,x_1,x_0)+\lag(\alpha_1+\alpha_3\,|\,x_2,x_0)+\ol(\alpha_3\,|\,x_0,x_3)\,,
\end{align}
and
\begin{align}
\ds\iC(x)=2\ii x(\Log(-2\ii x)-\Log(2\ii x))\,.
\end{align}

The saddle point three-leg equation \eqref{3legsym} is then given by
\begin{align}
\label{Q23leg}
\frac{1}{\ii}\hspace{-0.05cm}\left.\frac{\partial \astr{x}}{\partial x}\right|_{x=x_0}\hspace{-0.4cm}=\ovphi(\alpha_1\,|\,x_1,x_0)+\varphi(\alpha_1+\alpha_3\,|\,x_2,x_0)+\ovphi(\alpha_3\,|\,x_3,x_0)=0\,,
\end{align}
where $\varphi(\alpha\,|\,x_i,x_j)$, is defined by
\begin{align}
\begin{split}
\varphi(\alpha\,|\,x_i,x_j)=\Log(\ii(x_i+x_j)-\alpha)+\Log(\ii(x_i-x_j)+\alpha)\phantom{\,,} \\
-\Log(\ii(x_i+x_j)+\alpha)-\Log(\ii(x_i-x_j)-\alpha)\,,
\end{split}
\end{align}
and $\ovphi(\alpha\,|\,x_i,x_j)$, is defined by
\begin{align}
\begin{split}
\ovphi(\alpha\,|\,x_i,x_j)=\Log(\alpha+\ii(x_i+x_j))+\Log(\alpha-\ii(x_i-x_j))-\Log(\alpha-\ii(x_i+x_j))\phantom{\,,} \\
-\Log(\alpha+\ii(x_i-x_j)) -\pi\ii\,\textrm{Sign}(\re(x_j)).
\end{split}
\end{align}

The equation \eqref{Q23leg} is a three-leg form of $Q2$, arising as the equation for the saddle point of the star-triangle relation \eqref{qq2} in the limit \eqref{ratqcl}.  This was also previously obtained in \cite{Bazhanov:2016ajm}.  With the following change of variables
\begin{align}
x=x_0^2,\qquad u=x_1^2,\qquad y=x_2^2,\qquad v=x_3^2,\qquad \alpha=-\ii\alpha_1,\qquad \beta=-\ii(\alpha_1+\alpha_3),
\end{align}
the exponential of the three-leg equation \eqref{Q23leg} may be written in the form
\begin{align}
Q(x,u,y,v;\alpha,\beta)=0\,,
\end{align}
where 
\begin{align}
\label{Q2}
Q(x,u,y,v;\alpha,\beta)=\alpha(x-y)(u-v)-\beta(x-u)(y-v)+\alpha\beta(\alpha-\beta)(x+u+y+v-\alpha^2-\beta^2+\alpha\beta).
\end{align}
This is identified as $Q2$ \cite{ABS}, an affine-linear quad equation that satisfies the 3D-consistency condition exactly as defined in Section \ref{sec:overview}.

\subsection{\texorpdfstring{$Q1_{(\delta=1)}$}{Q1(delta=1)} case}

\tocless\subsubsection{Star-triangle relation}

Let the spins $\spn_1,\spn_2,\spn_3$, and spectral parameters $\theta_1,\theta_3$ take values \eqref{ratvals}.
The star-triangle relation is given in this case by \cite{Bazhanov:2016ajm}
\begin{align}
\label{qq1d1}
\begin{split}
\ds\int_{\mathbb{R}}d\spn_0\,\oW(\theta_1\,|\,\spn_1,\spn_0)\,W(\theta_1+\theta_3\,|\,\spn_2,\spn_0)\,\oW(\theta_3\,|\,\spn_0,\spn_3)\qquad\phantom{\,,} \\
=R(\theta_1,\theta_3)\,W(\theta_1\,|\,\spn_2,\spn_3)\,\oW(\theta_1+\theta_3\,|\,\spn_1,\spn_3)\,W(\theta_3\,|\,\spn_2,\spn_1)\,,
\end{split}
\end{align}
where the Boltzmann weights are
\begin{align}
\label{q1d1quan}
\begin{split}
\ds W(\theta\,|\,\spn_i,\spn_j)&=\frac{\Gamma(\pi+\ii(\spn_i-\spn_j)-\theta)}{\Gamma(\pi+\ii(\spn_i-\spn_j)+\theta)}\,, \\
\ds\oW(\theta\,|\,\spn_i,\spn_j)&=\Gamma(\theta+\ii(\spn_i-\spn_j))\,\Gamma(\theta-\ii(\spn_i-\spn_j))\,,
\end{split}
\end{align}
and the normalisation is given by
\begin{align}
R(\theta_1,\theta_3)=2\pi\,\frac{\Gamma(2\theta_1)\Gamma(2\theta_3)}{\Gamma(2(\theta_1+\theta_3))}\,.
\end{align}

The Boltzmann weights satisfy the following relations
\begin{align}
W(\theta\,|\,\spn_i,\spn_j)\,W(-\theta\,|\,\spn_i,\spn_j)=1\,,\qquad\oW(\theta\,|\,\spn_i,\spn_j)=\oW(\theta\,|\,\spn_j,\spn_i)\,.
\end{align}
While the Boltzmann weight $\oW(\theta\,|\,\spn_i,\spn_j)$ is real for the values \eqref{ratvals}, the Boltzmann weight $W(\theta\,|\,\spn_i,\spn_j)$ in general has $\im(W(\theta\,|\,\spn_i,\spn_j))\neq0$.

From \eqref{Stirling}, the asymptotics of the integrand of \eqref{qq1d1} are
\begin{align}
\label{q1d1xinf}
O(\EXP^{-2\pi\, \left|\spn_0\right|})\,,\qquad \spn_0\to\pm\infty\,,
\end{align}
and particularly the integral in \eqref{qq1d1} is absolutely convergent for the values \eqref{ratvals}.

The star-triangle relation \eqref{qq1d1} and connection to $Q1_{(\delta=1)}$ previously appeared in \cite{Bazhanov:2016ajm}.\footnote{See footnote \ref{footnotetypo}.}  The star-triangle relation \eqref{qq1d1} is also up to a change of variables equivalent to Barnes's second Lemma \cite{Barnes1910}, which may in turn be interpreted \cite{GR} as an integral analogue of the Saalsch\"{u}tz summation formula.

Define
\begin{align}
\label{q1d1int}
I_{6,Q,r}(\theta_1,\theta_3\,|\,\spn_1,\spn_2,\spn_3)=\int_{\mathcal{C}}d\spn\,\rho_{6,Q,r}(\theta_1,\theta_3\,|\,\spn_1,\spn_2,\spn_3;\spn)\,,
\end{align}
where
\begin{align}
\label{q1d1rho}
\rho_{6,Q,r}(\theta_1,\theta_3\,|\,\spn_1,\spn_2,\spn_3;\spn)=\oW(\theta_1\,|\,\spn_1,\spn)\,W(\theta_1+\theta_3\,|\,\spn_2,\spn)\,\oW(\theta_3\,|\,\spn_3,\spn)\,,
\end{align}
and the contour $\mathcal{C}$ is a deformation of $\mathbb{R}$, separating the points $\ii\theta_1+\spn_1+\ii\mathbb{Z}_{\geq0}$, $\ii\theta_3+\spn_3+\ii\mathbb{Z}_{\geq0}$, from the points $-\ii\theta_1+ \spn_1-\ii\mathbb{Z}_{\geq0}$, $-\ii\theta_3+ \spn_3-\ii\mathbb{Z}_{\geq0}$, $-\ii(\pi-\theta_1-\theta_3)+\spn_2-\ii\mathbb{Z}_{\geq0}$.  Then the star-triangle relation \eqref{qq1d1}, may be obtained from the limit \eqref{ratlim} of \eqref{qq3d0} ($Q3_{(\delta=0)}$ case):

\begin{prop} \label{q1d1prop}
For the values \eqref{q3d1rvals},
\begin{align}
\label{q1d1lim}
\lim_{\bb\to0^+}&\frac{(2\pi\bb)^4}{\bb^2}\,I_{6,Q}(\theta_1\bb,\theta_3\bb\,|\,\spn_1\bb,(\spn_2-\ii\pi)\bb+\ii\eta,\spn_3\bb)=I_{6,Q,r}(\theta_1,\theta_3\,|\,\spn_1,\spn_2,\spn_3)\,,
\end{align}
where $I_{6,Q}(\theta_1,\theta_3\,|\,\spn_1,\spn_2,\spn_3)$ is defined in \eqref{q3d0int}.

\end{prop}

The steps of the proof of Proposition \ref{q1d1prop}, are analogous to the respective steps for Proposition \ref{q2prop}. The star-triangle relation \eqref{qq1d1} follows from Proposition \ref{q1d1prop}, by using \eqref{ratasymp} to take the same limit of the right hand side of \eqref{qq3d0}.

\tocless\subsubsection{Classical integrable equations}

Let the classical variables $x_1,x_2,x_3$, $\alpha_1,\alpha_3$, take the values given in \eqref{ratcvals}.  The Lagrangian functions for this case are defined in terms of \eqref{gammadef} as
\begin{align}
\label{lagq1d1}
\begin{split}
\lag(\alpha\,|\,x_i,x_j)&=\gamma(x_i-x_j+\ii\alpha)-\gamma(x_i-x_j-\ii\alpha)\,, \\
\ol(\alpha\,|\,x_i,x_j)&=\gamma(x_i-x_j-\ii\alpha)+\gamma(x_j-x_i-\ii\alpha)-\gamma(-2\ii\alpha) \,.
\end{split}
\end{align}
These functions satisfy
\begin{align}
\lag(\alpha\,|\,x_i,x_j)=-\lag(-\alpha\,|\,x_i,x_j)\,,\qquad\ol(\alpha\,|\,x_i,x_j)=\ol(\alpha\,|\,x_j,x_i)\,.
\end{align}

Using the asymptotics \eqref{Stirling}, the leading order $O(\hbar^{-1})$ quasi-classical expansion \eqref{ratqcl} of the integrand \eqref{q1d1rho} is
\begin{align}
\label{q1d1qcl}
\Log \left(\hbar^{-2}\rho_{6,r}(\alpha_1\hbar^{-1},\alpha_3\hbar^{-1}\,|\,x_1\hbar^{-1},x_2\hbar^{-1},x_3\hbar^{-1};x_0\hbar^{-1})\right)=\hbar^{-1}\astr{x_0}+O(1)\,,
\end{align}
where
\begin{align}
\astr{x_0}=\ol(\alpha_1\,|\,x_1,x_0)+\lag(\alpha_1+\alpha_3\,|\,x_2,x_0)+\ol(\alpha_3\,|\,x_0,x_3)\,.
\end{align}

The saddle point three-leg equation \eqref{3legsym} is then given by
\begin{align}
\label{Q1d13leg}
\frac{1}{\ii}\hspace{-0.05cm}\left.\frac{\partial \astr{x}}{\partial x}\right|_{x=x_0}\hspace{-0.4cm}=\ovphi(\alpha_1\,|\,x_1,x_0)+\varphi(\alpha_1+\alpha_3\,|\,x_2,x_0)+\ovphi(\alpha_3\,|\,x_0,x_3)=0\,,
\end{align}
where $\varphi(\alpha\,|\,x_i,x_j)$, is defined by
\begin{align}
\varphi(\alpha\,|\,x_i,x_j)=\Log(\ii(x_i-x_j)+\alpha)-\Log(\ii(x_i-x_j)-\alpha)\,,
\end{align}
and $\ovphi(\alpha\,|\,x_i,x_j)$, is defined by
\begin{align}
\ovphi(\alpha\,|\,x_i,x_j)=\Log(\alpha-\ii(x_i-x_j))-\Log(\alpha+\ii(x_i-x_j))\,.
\end{align}

The equation \eqref{Q1d13leg} is a three-leg form of $Q1_{(\delta=1)}$, arising as the equation for the saddle point of the star-triangle relation \eqref{qq1d1} in the limit \eqref{ratqcl}.  This was also previously obtained in \cite{Bazhanov:2016ajm}.  With the following change of variables
\begin{align}
x=x_0,\qquad u=x_1,\qquad y=x_2,\qquad v=x_3,\qquad \alpha=-\ii\alpha_1,\qquad \beta=-\ii(\alpha_1+\alpha_3),
\end{align}
the exponential of the three-leg equation \eqref{Q1d13leg} may be written in the form
\begin{align}
Q(x,u,y,v;\alpha,\beta)=0\,,
\end{align}
where 
\begin{align}
\label{Q1d1}
Q(x,u,y,v;\alpha,\beta)=\alpha(x-y)(u-v)-\beta(x-u)(y-v)+(\alpha-\beta)\alpha\beta\,.
\end{align}
This is identified as $Q1_{(\delta=1)}$ \cite{ABS}, an affine-linear quad equation that satisfies the 3D-consistency condition exactly as defined in Section \ref{sec:overview}.

\subsection{\texorpdfstring{$H2_{(\varepsilon=1)}$}{H2(epsilon=1)} case}

\tocless\subsubsection{Star-triangle relation}

Let the spins $\spn_1,\spn_2,\spn_3$, and spectral parameters $\theta_1,\theta_3$ take values \eqref{ratvals}.  The star-triangle relation is given in this case by
\begin{align}
\label{qh2d1}
\begin{split}
\ds\int_{\mathbb{R}}d\spn_0\;\iS(\spn_0)\,\oV(\theta_1\,|\,\spn_1,\spn_0)\,V({\theta_1+\theta_3}\,|\,\spn_2,\spn_0)\,\oW(\theta_3\,|\,\spn_0,\spn_3)\phantom{\,,}  \\
\ds=R(\theta_3)\,V({\theta_1}\,|\,\spn_2,\spn_3)\,\oV({\theta_1+\theta_3}\,|\,\spn_1,\spn_3)\,W({\theta_3}\,|\,\spn_2,\spn_1)\,,
\end{split}
\end{align}
where the Boltzmann weights are
\begin{align}
\label{h2d1quan}
\begin{split}
\ds V(\theta\,|\,\spn_i,\spn_j)&=\Gamma(\pi+\ii(\spn_i+\spn_j)-\theta)\,\Gamma(\pi+\ii(\spn_i-\spn_j)-\theta)\,,\\
\ds \oV(\theta\,|\,\spn_i,\spn_j)&=V(\pi-\theta\,|\,-\spn_i,\spn_j)\,, \\
\ds\oW(\theta\,|\,\spn_i,\spn_j)&=\Gamma(\theta+\ii(\spn_i+\spn_j))\,\Gamma(\theta+\ii(\spn_i-\spn_j))\,\Gamma(\theta-\ii(\spn_i+\spn_j))\,\Gamma(\theta-\ii(\spn_i-\spn_j))\,, \\
\ds W(\theta\,|\,\spn_i,\spn_j)&=\frac{\Gamma(\pi+\ii(\spn_i-\spn_j)-\theta)}{\Gamma(\pi+\ii(\spn_i-\spn_j)+\theta)}\,, 
\end{split}
\end{align}
and
\begin{align}
\label{h2d1s}
\ds S(\spn)=\frac{1}{4\pi\,\Gamma(2\ii \spn)\,\Gamma(-2\ii \spn)}\,,\qquad \ds R(\theta)=\Gamma(2\theta)\,.
\end{align}
The Boltzmann weights $\oW(\theta\,|\,\spn_i,\spn_j)$ and $W(\theta\,|\,\spn_i,\spn_j)$ appear in \eqref{q2quan}, and \eqref{q1d1quan}, respectively, while the Boltzmann weight $V(\theta\,|\,\spn_i,\spn_j)$ did not appear previously.

From \eqref{Stirling}, the asymptotics of the integrand of \eqref{qh2d1} are
\begin{align}
\label{h2d1xinf}
O(\EXP^{-2\pi\, \left|\spn_0\right|})\,,\qquad \spn_0\to\pm\infty\,,
\end{align}
and particularly the integral in \eqref{qh2d1} is absolutely convergent for the values \eqref{ratvals}.

The star-triangle relation \eqref{qh2d1} did not appear before, but up to a change of variables is equivalent to a gamma function integration formula found independently by de Branges \cite{DeBranges1972}, and Wilson \cite{Wilson1980}.

Define
\begin{align}
\label{h2d1int}
I_{10,r}(\theta_1,\theta_3\,|\,\spn_1,\spn_2,\spn_3)=\int_{\mathcal{C}}d\spn\,\rho_{10,r}(\theta_1,\theta_3\,|\,\spn_1,\spn_2,\spn_3;\spn)\,,
\end{align}
where
\begin{align}
\label{h2d1rho}
\rho_{10,r}(\theta_1,\theta_3\,|\,\spn_1,\spn_2,\spn_3;\spn)=\iS(\spn)\,\oV(\theta_1\,|\,\spn_1,\spn)\,V({\theta_1+\theta_3}\,|\,\spn_2,\spn)\,\oW(\theta_3\,|\,\spn,\spn_3)\,,
\end{align}
and the contour $\mathcal{C}$ is a deformation of $\mathbb{R}$, separating the points $\ii\theta_1+\spn_1+\ii\mathbb{Z}_{\geq0}$, $\ii\theta_3\pm \spn_3+\ii\mathbb{Z}_{\geq0}$, $\ii(\pi-\theta_1-\theta_3)-\spn_2+\ii\mathbb{Z}_{\geq0}$, from their negatives.  Then the star-triangle relation \eqref{qh2d1}, may be obtained from the limit \eqref{ratlim} of \eqref{qh3d1} ($H3_{(\delta=1;\,\varepsilon=1)}$ case):

\begin{prop} \label{h2d1prop}
For the values \eqref{h3d1rvals}
\begin{align}
\label{h2d1lim}
\begin{split}
\lim_{\bb\to0^+}&\frac{(2\pi\bb)^6}{\pi\bb^4 (2\pi\bb^2)^{2(\ii(\spn_2-\spn_1)+\theta_3)}}\, I_{10}(\theta_1\bb,\theta_3\bb\,|\,\spn_1\bb,(\spn_2-\ii\pi)\bb+\ii\eta,\spn_3\bb)=I_{10,r}(\theta_1,\theta_3\,|\,\spn_1,\spn_2,\spn_3)\,,
\end{split}
\end{align}
where $I_{10}(\theta_1,\theta_3\,|\,\spn_1,\spn_2,\spn_3)$ is defined in \eqref{h3d1int}.

\end{prop}

The steps of the proof of Proposition \ref{h2d1prop}, are analogous to the respective steps for Proposition \ref{q2prop}.  The star-triangle relation \eqref{qh2d1} follows from Proposition \ref{h2d1prop}, by using \eqref{ratasymp} to take the same limit of the right hand side of \eqref{qh3d1}.

\tocless\subsubsection{Classical integrable equations}

Let the classical variables $x_1,x_2,x_3$, $\alpha_1,\alpha_3$, take the values given in \eqref{ratcvals}.  The Lagrangian functions for this case are defined in terms of \eqref{gammadef} as
\begin{align}
\label{lagh2d1}
\begin{split}
\Lambda(\alpha\,|\,x_i,x_j)=&\,\ds \gamma(x_i+x_j+\ii\alpha)+\gamma(x_i-x_j+\ii\alpha)\,, \\
\ol(\alpha\,|\,x_i,x_j)=&\,\ds \gamma(x_i+x_j-\ii\alpha)+\gamma(x_i-x_j-\ii\alpha)+\gamma(-x_i+x_j-\ii\alpha) \\
&+\gamma(-x_i-x_j-\ii\alpha) -\gamma(-2\ii\alpha)\,.
\end{split}
\end{align}
The latter function satisfies
\begin{align}
\ol(\alpha\,|\,x_i,x_j)=\ol(\alpha\,|\,x_j,x_i)\,.
\end{align}

Using the asymptotics \eqref{Stirling}, the leading order $O(\hbar^{-1})$ quasi-classical expansion \eqref{ratqcl} of the integrand \eqref{h2d1rho} is
\begin{align}
\label{h2d1qcl}
\begin{split}
\Log \left(\hbar^{2\pi-3}\rho_{10,r}(\alpha_1\hbar^{-1},\alpha_3\hbar^{-1}\,|\,x_1\hbar^{-1},x_2\hbar^{-1},x_3\hbar^{-1};x_0\hbar^{-1})\right)\phantom{\,,\qquad\qquad\quad} \\
+2\ii\hbar^{-1}(x_2-x_1-\ii\alpha_3)(1+\Log\hbar)=\hbar^{-1}\astr{x_0}+O(1)\,,
\end{split}
\end{align}
where
\begin{align}
\astr{x_0}=C(x_0)+\olam(\alpha_1\,|\,x_1,x_0)+\Lambda(\alpha_1+\alpha_3\,|\,x_2,x_0)+\ol(\alpha_3\,|\,x_0,x_3)\,,
\end{align}
and
\begin{align}
\begin{gathered}
\olam(\alpha\,|\,x_i,x_j)=\Lambda(-\alpha\,|\,-x_i,x_j)\,, \\
C(x)=2\ii x(\Log(-2\ii x)-\Log(2\ii x))\,.
\end{gathered}
\end{align}

The saddle point three-leg equation \eqref{3legasym} is then given by
\begin{align}
\label{H2d13leg}
\frac{1}{\ii}\hspace{-0.05cm}\left.\frac{\partial \astr{x}}{\partial x}\right|_{x=x_0}\hspace{-0.4cm}=\ovphib(\alpha_1\,|\,x_1,x_0)+\phi(\alpha_1+\alpha_3\,|\,x_2,x_0)+\vphi(\alpha_3\,|\,x_3,x_0)=0\,,
\end{align}
where $\phi(\alpha\,|\,x_i,x_j)$, and $\vphi(\alpha\,|\,x_i,x_j)$, are defined by
\begin{align}
\begin{split}
\phi(\alpha\,|\,x_i,x_j)&=\Log(\ii(x_i+x_j)-\alpha)-\Log(\ii(x_i-x_j)-\alpha)\,, \\
\vphi(\alpha\,|\,x_i,x_j)&=\Log(\alpha+\ii(x_i+x_j))+\Log(\alpha-\ii(x_i-x_j)) \\
&-\Log(\alpha-\ii(x_i+x_j))-\Log(\alpha+\ii(x_i-x_j))\,,
\end{split}
\end{align}
and
\begin{align}
\ovphib(\alpha\,|\,x_i,x_j)=\phi(-\alpha\,|\,-x_i,x_j)\,.
\end{align}

Equation \eqref{H2d13leg} is a three-leg form of $H2_{(\varepsilon=1)}$, arising as the equation for the saddle point of the star-triangle relation \eqref{qh2d1} in the limit \eqref{ratqcl}.  With the following change of variables
\begin{align}
x=x_0^2,\qquad u=x_1,\qquad y=x_2,\qquad v=x_3^2,\qquad \alpha=-\ii\alpha_1,\qquad \beta=-\ii(\alpha_1+\alpha_3),
\end{align}
the exponential of the three-leg equation \eqref{H2d13leg} may be written in the form
\begin{align}
H(x,u,y,v;\alpha,\beta)=0\,,
\end{align}
where 
\begin{align}
\label{H2d1}
H(x,u,y,v;\alpha,\beta)=(x-v)(u-y)+(\alpha-\beta)(x+v+(u+y)(\alpha+\beta)-2uy-\alpha^2-\beta^2)\,.
\end{align}
This is identified as $H2_{(\varepsilon=1)}$ \cite{ABS2}, an affine-linear quad equation that satisfies the 3D-consistency condition exactly as defined in Section \ref{sec:overview}.

\subsection{\texorpdfstring{$H2_{(\varepsilon=1)}$}{H2(epsilon=1)} case (alternate form)}

\tocless\subsubsection{Star-triangle relation}

Let the spins $\spn_1,\spn_2,\spn_3$, and spectral parameters $\theta_1,\theta_3$ take values \eqref{ratvals}.  The star-triangle relation is given in this case by
\begin{align}
\label{qh2d1alt}
\begin{split}
\ds\int_{\mathbb{R}}d\spn_0\,\oV({\theta_1}\,|\,\spn_1,\spn_0)\,V({\theta_1+\theta_3}\,|\,\spn_2,\spn_0)\,\oW({\theta_3}\,|\,\spn_0,\spn_3)\phantom{\quad\,.}  \\
\ds=R(\theta_3)\,V({\theta_1}\,|\,\spn_2,\spn_3)\,\oV({\theta_1+\theta_3}\,|\,\spn_1,\spn_3)\,W({\theta_3}\,|\,\spn_2,\spn_1)\,,
\end{split}
\end{align}
where the Boltzmann weights are
\begin{align}
\label{h2d1quanalt}
\begin{split}
\ds V(\theta\,|\,\spn_i,\spn_j)&=\frac{\Gamma(\pi+\ii(\spn_i+\spn_j)-\theta)}{\Gamma(\pi+\ii(\spn_i-\spn_j)+\theta)}\,, \\
\ds\oV(\theta\,|\,\spn_i,\spn_j)&=\Gamma(\theta+\ii(\spn_i-\spn_j))\,\Gamma(\theta-\ii(\spn_i+\spn_j))\,, \\[0.1cm]
\ds\oW(\theta\,|\,\spn_i,\spn_j)&=\Gamma(\theta+\ii(\spn_i-\spn_j))\,\Gamma(\theta- \ii(\spn_i-\spn_j))\,, \\
\ds W(\theta\,|\,\spn_i,\spn_j)&=\frac{\Gamma(\pi+\ii(\spn_i+\spn_j)-\theta)\,\Gamma(\pi+\ii(\spn_i-\spn_j)-\theta)}{\Gamma(\pi+\ii(\spn_i+\spn_j)+\theta)\,\Gamma(\pi+\ii(\spn_i-\spn_j)+\theta)}\,,
\end{split}
\end{align}
and
\begin{align}
R(\theta_3)=2\pi\,\Gamma(2\theta_3)\,.
\end{align}

From \eqref{Stirling}, the asymptotics of the integrand of \eqref{qh2d1alt} are
\begin{align}
\label{h2d12xinf}
O(\EXP^{-2\pi\, \left|\spn_0\right|})\,,\qquad \spn_0\to\pm\infty\,,
\end{align}
and particularly the integral in \eqref{qh2d1alt} is absolutely convergent for the values \eqref{ratvals}.

Similarly to the preceding case of $H2_{(\varepsilon=1)}$ in \eqref{h2d1quan}, $W(\theta\,|\,\spn_i,\spn_j)$, and $\oW(\theta\,|\,\spn_i,\spn_j)$ appear in \eqref{q2quan}, and \eqref{q1d1quan}, respectively, while the Boltzmann weight $V(\theta\,|\,\spn_i,\spn_j)$ is not symmetric in the exchange of spins, and generally has non-vanishing imaginary component.  In the star-triangle relation \eqref{qh2d1alt}, the Boltzmann weights $W(\theta\,|\,\spn_i,\spn_j)$ and $\oW(\theta\,|\,\spn_i,\spn_j)$ are exchanged compared to \eqref{qh2d1}, which will result in a different three-leg form to \eqref{H2d13leg}, that corresponds to the same quad equation $H2_{(\varepsilon=1)}$ (up to relabelling of the variables).

The star-triangle relation \eqref{qh2d1alt} is up to a change of variables equivalent to Barnes's second lemma \cite{Barnes1910} (the same integral formula corresponding to \eqref{qq1d1} for the $Q1_{(\delta=1)}$ case, but with a slightly different change of variables).

Define
\begin{align}
\label{h2d1intalt}
I_{6,H,r}(\theta_1,\theta_3\,|\,\spn_1,\spn_2,\spn_3)=\int_{\mathcal{C}}d\spn\,\rho_{6,H,r}(\theta_1,\theta_3\,|\,\spn_1,\spn_2,\spn_3;\spn)\,,
\end{align}
where
\begin{align}
\label{h2d1rhoalt}
\rho_{6,H,r}(\theta_1,\theta_3\,|\,\spn_1,\spn_2,\spn_3;\spn)=\oV({\theta_1}\,|\,\spn_1,\spn)\,V({\theta_1+\theta_3}\,|\,\spn_2,\spn)\,\oW({\theta_3}\,|\,\spn,\spn_3)\,,
\end{align}
and the contour $\mathcal{C}$ is a deformation of $\mathbb{R}$, separating the points $\ii\theta_3+ \spn_3+\ii\mathbb{Z}_{\geq0}$, $\ii(\pi-\theta_1-\theta_3)-\spn_2+\ii\mathbb{Z}_{\geq0}$, from the points $-\ii\theta_3+\spn_3-\ii\mathbb{Z}_{\geq0}$, $-\ii\theta_1\pm \spn_1-\ii\mathbb{Z}_{\geq0}$.  Then the star-triangle relation \eqref{qh2d1alt}, may be obtained from the limit \eqref{ratlim} of \eqref{qh3d1alt} ($H3_{(\delta=1;\,\varepsilon=1)}$ case):

\begin{prop} \label{h2d12prop}
For the values \eqref{h3d12rvals},
\begin{align}
\label{h2d12lim}
\lim_{\bb\to0^+}&\frac{(2\pi\bb)^4}{\bb^2}\,I_{6,H}(\theta_1\bb,\theta_3\bb\,|\,\spn_1\bb,(\spn_2-\ii\pi)\bb+\ii\eta,\spn_3\bb)=I_{6,H,r}(\theta_1,\theta_3\,|\,\spn_1,\spn_2,\spn_3)
\end{align}
where $I_{6,H}(\theta_1,\theta_3\,|\,\spn_1,\spn_2,\spn_3)$ is defined in \eqref{h3d12int}.

\end{prop}

The steps of the proof of Proposition \ref{h2d12prop}, are analogous to the respective steps for Proposition \ref{q2prop}.  The star-triangle relation \eqref{qh2d1alt} follows from Proposition \ref{h2d12prop}, by using \eqref{ratasymp} to take the same limit of the right hand side of \eqref{qh3d1alt}.

\tocless\subsubsection{Classical integrable equations}

Let the classical variables $x_1,x_2,x_3$, $\alpha_1,\alpha_3$, take the values given in \eqref{ratcvals}.  The Lagrangian functions for this case are defined in terms of \eqref{gammadef} as
\begin{align}
\label{lagh2d1alt}
\begin{split}
\Lambda(\alpha\,|\,x_i,x_j)&=\ds \gamma(x_i+x_j+\ii\alpha)-\gamma(x_i-x_j-\ii\alpha)\,,  \\
\olam(\alpha\,|\,x_i,x_j)&=\ds \gamma(x_i-x_j-\ii\alpha)+\gamma(-(x_i+x_j+\ii\alpha))\,, \\
\ol(\alpha\,|\,x_i,x_j)&=\ds \gamma(x_i-x_j-\ii\alpha)+\gamma(x_j-x_i-\ii\alpha)-\gamma(-2\ii\alpha)\,.
\end{split}
\end{align}
The latter function satisfies
\begin{align}
\ol(\alpha\,|\,x_i,x_j)=\ol(\alpha\,|\,x_j,x_i)\,.
\end{align}

Using the asymptotics \eqref{Stirling}, the leading order $O(\hbar^{-1})$ quasi-classical expansion \eqref{ratqcl} of the integrand \eqref{h2d1rhoalt} is
\begin{align}
\label{h2d1qclalt}
\Log \left(\hbar^{-2}\!\rho_{6,H,r}(\alpha_1\hbar^{-1},\alpha_3\hbar^{-1}\,|\,x_1\hbar^{-1},x_2\hbar^{-1},x_3\hbar^{-1};x_0\hbar^{-1})\right)=\hbar^{-1}\!\astr{x_0}+O(1),
\end{align}
where
\begin{align}
\astr{x_0}=\olam(\alpha_1\,|\,x_1,x_0)+\Lambda(\alpha_1+\alpha_3\,|\,x_2,x_0)+\ol(\alpha_3\,|\,x_0,x_3)\,.
\end{align}

The saddle point three-leg equation \eqref{3legasym} is then given by
\begin{align}
\label{H2d13legalt}
\frac{1}{\ii}\hspace{-0.05cm}\left.\frac{\partial \astr{x}}{\partial x}\right|_{x=x_0}\hspace{-0.4cm}=\ovphib(\alpha_1\,|\,x_1,x_0)+\phi(\alpha_1+\alpha_3\,|\,x_2,x_0)+\vphi(\alpha_3\,|\,x_3,x_0)=0\,,
\end{align}
where 
\begin{align}
\begin{split}
\phi(\alpha\,|\,x_i,x_j)&=\Log(\ii(x_i+x_j)-\alpha)+\Log(\ii(x_i-x_j)+\alpha)\,, \\
\ovphib(\alpha\,|\,x_i,x_j)&=-\Log(\alpha-\ii(x_i+x_j))-\Log(\alpha+\ii(x_i-x_j))\,, \\
\vphi(\alpha\,|\,x_i,x_j)&=\Log(\alpha-\ii(x_i-x_j))-\Log(\alpha+\ii(x_i-x_j))\,.
\end{split}
\end{align}

Equation \eqref{H2d13legalt} is a three-leg form of $H2_{(\varepsilon=1)}$, arising as the equation for the saddle point of the star-triangle relation \eqref{qh2d1alt} in the limit \eqref{ratqcl}.  With the following change of variables
\begin{align}
x=x_0,\qquad u=x_1^2,\qquad y=x_2^2,\qquad v=x_3,\qquad \alpha=-\ii\alpha_1,\qquad \beta=-\ii(\alpha_1+\alpha_3),
\end{align}
the exponential of the three-leg equation \eqref{H2d13legalt} may be written in the form
\begin{align}
H(x,u,y,v;\alpha,\beta)=0\,,
\end{align}
where 
\begin{align}
\label{H2d12}
H(x,u,y,v;\alpha,\beta)=(x-v)(u-y)+(\alpha-\beta)(u+y+(x+v)(\alpha+\beta)-2xv-\alpha^2-\beta^2)\,.
\end{align}
This is identified as $H2_{(\varepsilon=1)}$ \cite{ABS2}, an affine-linear quad equation that satisfies the 3D-consistency condition exactly as defined in Section \ref{sec:overview}.

\subsection{\texorpdfstring{$H2_{(\varepsilon=0)}$}{H2(epsilon=0)} case}

\tocless\subsubsection{Star-triangle relation}

Let the spins $\spn_1,\spn_2,\spn_3$, and spectral parameters $\theta_1,\theta_3$ take values \eqref{ratvals}.  The star-triangle relation is given in this case by
\begin{align}
\label{qh2d0}
\begin{split}
\ds\int_{\mathbb{R}} d\spn_0\,\oV({\theta_1}\,|\,\spn_1,\spn_0)\,V({\theta_1+\theta_3}\,|\,\spn_2,\spn_0)\,\oW({\theta_3}\,|\,\spn_0,\spn_3)\phantom{\quad\,,}  \\
\ds=R(\theta_3)\,V({\theta_1}\,|\,\spn_2,\spn_3)\,\oV({\theta_1+\theta_3}\,|\,\spn_1,\spn_3)\,W({\theta_3}\,|\,\spn_2,\spn_1)\,,
\end{split}
\end{align}
where the Boltzmann weights are
\begin{align}
\label{h2d0quan}
\begin{split}
V(\theta\,|\,\spn_i,\spn_j)&\ds=\Gamma(\pi+\ii(\spn_i+\spn_j)-\theta)\,, \\
\oV(\theta\,|\,\spn_i,\spn_j)&\ds= V(\pi-\theta\,|\,-\spn_i,-\spn_j)\,, \\
\oW(\theta\,|\,\spn_i,\spn_j)&\ds=\Gamma(\theta+\ii(\spn_i-\spn_j))\,\Gamma(\theta-\ii(\spn_i-\spn_j))\,, \\
W(\theta\,|\,\spn_i,\spn_j)&\ds=\frac{\Gamma(\pi+\ii(\spn_i-\spn_j)-\theta)}{\Gamma(\pi+\ii(\spn_i-\spn_j)+\theta)}\,,
\end{split}
\end{align}
and
\begin{align}
R(\theta)=2\pi\,\Gamma(2\theta)\,.
\end{align}
The Boltzmann weights satisfy
\begin{align}
V(\theta\,|\,\spn_i,\spn_j)=V(\theta\,|\,\spn_j,\spn_i)\,,\quad \oW(\theta\,|\,\spn_i,\spn_j)=\oW(\theta\,|\,\spn_j,\spn_i)\,,\quad W(\theta\,|\,\spn_i,\spn_j)\,W(-\theta\,|\,\spn_i,\spn_j)=1\,.
\end{align}

From \eqref{Stirling}, the asymptotics of the integrand of \eqref{qh2d0} are
\begin{align}
\label{h2d0xinf}
O(\EXP^{-2\pi |\spn_0|})\,,\qquad \spn_0\to\pm\infty\,,
\end{align}
and particularly the integral in \eqref{qh2d0} is absolutely convergent for the values \eqref{ratvals}.

The star-triangle relation \eqref{qh2d0} did not appear before, but up to a change of variables is equivalent to Barnes's first lemma \cite{Barnes1908}.

Define
\begin{align}
\label{h2d0int}
I_{4,r}(\theta_1,\theta_3\,|\,\spn_1,\spn_2,\spn_3)=\int_{\mathcal{C}}d\spn\,\rho_{4,r}(\theta_1,\theta_3\,|\,\spn_1,\spn_2,\spn_3;\spn)\,,
\end{align}
where
\begin{align}
\label{h2d0rho}
\rho_{4,r}(\theta_1,\theta_3\,|\,\spn_1,\spn_2,\spn_3;\spn)=\oV({\theta_1}\,|\,\spn_1,\spn)\,V({\theta_1+\theta_3}\,|\,\spn_2,\spn)\,\oW({\theta_3}\,|\,\spn,\spn_3)\,,
\end{align}
and the contour $\mathcal{C}$ is a deformation of $\mathbb{R}$, separating the points $\ii\theta_3+\spn_3+\ii\mathbb{Z}_{\geq0}$, $\ii(\pi-\theta_1-\theta_3)-\spn_2+\ii\mathbb{Z}_{\geq0}$, from the points $-\ii\theta_1-\spn_1-\ii\mathbb{Z}_{\geq0}$, $-\ii\theta_3+\spn_3-\ii\mathbb{Z}_{\geq0}$.  Then the star-triangle relation \eqref{qh2d0}, may be obtained from the limit \eqref{ratlim} of \eqref{qh3d0} ($H3_{(\delta=0,1;\,\varepsilon=1-\delta)}$ case):

\begin{prop} \label{h2d0prop}
For the values \eqref{h3d0rvals},
\begin{align}
\label{h2d0lim}
\lim_{\bb\to0^+}\frac{\sqrt{2\pi}(2\pi\bb)^3}{\bb^2(2\pi\bb^2)^{\ii(\spn_2-\spn_1)+\theta_3}}\, I_{4}(\theta_1\bb,\theta_3\bb\,|\,\spn_1\bb,(\spn_2-\ii\pi)\bb+\ii\eta,\spn_3\bb)=I_{4,r}(\theta_1,\theta_3\,|\,\spn_1,\spn_2,\spn_3)
\end{align}
where $I_{4}(\theta_1,\theta_3\,|\,\spn_1,\spn_2,\spn_3)$ is defined in \eqref{h3d0int}.

\end{prop}

The steps of the proof of Proposition \ref{h2d0prop}, are analogous to the respective steps for Proposition \ref{q2prop}.  The star-triangle relation \eqref{qh2d0} follows from Proposition \ref{h2d0prop}, by using \eqref{ratasymp} to take the same limit of the right hand side of \eqref{qh3d0}.

\tocless\subsubsection{Classical integrable equations}

Let the classical variables $x_1,x_2,x_3$, $\alpha_1,\alpha_3$, take the values given in \eqref{ratcvals}.  The Lagrangian functions for this case are defined in terms of \eqref{gammadef} as
\begin{align}
\label{lagh2d0}
\begin{split}
\ds\Lambda(\alpha\,|\,x_i,x_j)&=\gamma(x_i+x_j+\ii\alpha) \,, \\
\ds{\ol}(\alpha\,|\,x_i,x_j)&=\gamma(x_i-x_j-\ii\alpha)+\gamma(x_j-x_i-\ii\alpha)-\gamma(-2\ii\alpha)\,.
\end{split}
\end{align}
These functions satisfy \eqref{lagh2d0} satisfy
\begin{align}
\Lambda(\alpha\,|\,x_i,x_j)=\Lambda(\alpha\,|\,x_j,x_i)\,,\quad\ol(\alpha\,|\,x_i,x_j)=\ol(\alpha\,|\,x_j,x_i)\,,\quad{\lag}(\alpha\,|\,x_i,x_j)=-{\lag}(-\alpha\,|\,x_i,x_j)\,.
\end{align}

Using the asymptotics \eqref{Stirling}, the leading order $O(\hbar^{-1})$ quasi-classical expansion \eqref{ratqcl} of the integrand \eqref{h2d0rho} is
\begin{align}
\label{h2d0qcl}
\begin{split}
\ii\hbar^{-1}(x_2-x_1-\ii\alpha_3)(1+\Log\hbar)+\Log \left(\hbar^{\pi-2}\!\rho_{4,r}(\alpha_1\hbar^{-1},\alpha_3\hbar^{-1}\,|\,x_1\hbar^{-1},x_2\hbar^{-1},x_3\hbar^{-1};x_0\hbar^{-1})\right)\phantom{\,,} \\
=\hbar^{-1}\!\astr{x_0}+O(1)\,,
\end{split}
\end{align}
where
\begin{align}
\astr{x_0}=\olam(\alpha_1\,|\,x_1,x_0)+\Lambda(\alpha_1+\alpha_3\,|\,x_2,x_0)+\ol(\alpha_3\,|\,x_0,x_3)\,,
\end{align}
and
\begin{align}
\olam(\alpha\,|\,x_i,x_j)=\Lambda(-\alpha\,|\,-x_i,-x_j)\,.
\end{align}

The saddle point three-leg equation \eqref{3legasym} is then given by
\begin{align}
\label{H2d03leg}
\frac{1}{\ii}\hspace{-0.05cm}\left.\frac{\partial \astr{x}}{\partial x}\right|_{x=x_0}\hspace{-0.4cm}=\ovphib(\alpha_1\,|\,x_1,x_0)+\phi(\alpha_1+\alpha_3\,|\,x_2,x_0)+\vphi(\alpha_3\,|\,x_3,x_0)=0\,,
\end{align}
where $\phi(\alpha\,|\,x_i,x_j)$, and $\vphi(\alpha\,|\,x_i,x_j)$, are defined by
\begin{align}
\begin{split}
\phi(\alpha\,|\,x_i,x_j)&=\Log(\ii(x_i+ x_j)-\alpha)\,, \\
\vphi(\alpha\,|\,x_i,x_j)&=\Log(\alpha-\ii(x_i-x_j))-\Log(\alpha+\ii(x_i-x_j))\,,
\end{split}
\end{align}
and
\begin{align}
\ovphib(\alpha\,|\,x_i,x_j)=-\phi(-\alpha\,|\,-x_i,-x_j)\,.
\end{align}

Equation \eqref{H2d03leg} is a three-leg form of $H2_{(\varepsilon=0)}$, arising as the equation for the saddle point of the star-triangle relation \eqref{qh2d0} in the limit \eqref{ratqcl}.  With the following change of variables
\begin{align}
x=x_0,\qquad u=x_1,\qquad y=x_2,\qquad v=x_3,\qquad \alpha=-\ii\alpha_1,\qquad \beta=-\ii(\alpha_1+\alpha_3),
\end{align}
the exponential of the three-leg equation \eqref{H2d03leg} may be written in the form
\begin{align}
H(x,u,y,v;\alpha,\beta)=0\,,
\end{align}
where 
\begin{align}
\label{H2d0}
H(x,u,y,v;\alpha,\beta)=(x-v)(u-y)+(\alpha-\beta)(x+u+y+v-\alpha-\beta)\,.
\end{align}
This is identified as $H2_{(\varepsilon=0)}$ \cite{ABS2}, an affine-linear quad equation that satisfies the 3D-consistency condition exactly as defined in Section \ref{sec:overview}.

\section{Algebraic cases}
\label{sec:alglim}

The algebraic cases\footnote{The star-triangle relations for the different algebraic cases will be seen to correspond to various ``classical'' hypergeometric integrals (with the exception of \eqref{qh2d1alt} for an alternate form of $H1_{(\varepsilon=1)}$, which corresponds to a special case of Barnes's first lemma).  The term ``algebraic'' is used here rather than ``classical'', since the latter term is already used when referring to the quasi-classical expansion, and the resulting classical integrable equations.} of the star-triangle relation, may be obtained as certain limits of the hyperbolic cases in Section \ref{sec:hyplim}.  The limits are less straightforward than the limits of the preceding sections, as the exact form of a limit needs to be determined case by case.  The different algebraic cases of the star-triangle relation that are obtained here as limits of the hyperbolic cases, are summarised in Figure \ref{alglimsfig}.

\begin{figure}[tbh]
\centering
\begin{tikzpicture}[scale=1]
\draw[white!] (-7,3.5) circle (0.01pt)
node[above=1pt]{\color{black} Hyperbolic};
\draw[white!] (7,3) circle (0.01pt);
\draw[white!] (-7,0) circle (0.01pt)
node[above=1pt]{\color{black} Algebraic};
\draw[white!] (7,0) circle (0.01pt);

\draw[white!] (-4.2,3) circle (0.01pt)
node[above=1pt]{\color{black} $ \begin{array}{cc} Q3_{(\delta=1)} \\[0.1cm] \eqref{qq3d1} \end{array} $};
\draw[white!] (0,3) circle (0.01pt)
node[above=1pt]{\color{black} $ \begin{array}{cc} H3_{(\delta=1;\,\varepsilon=1)} \\[0.1cm] \eqref{qh3d1},\eqref{qh3d1alt} \end{array} $ };
\draw[white!] (2.4,3) circle (0.01pt)
node[above=1pt]{\color{black} $ \begin{array}{cc} H3_{(\delta=0;\,\varepsilon=0)} \\[0.1cm] \eqref{qh3d02} \end{array} $ };
\draw[white!] (5.3,3) circle (0.01pt)
node[above=1pt]{\color{black} $ \begin{array}{cc} H3_{(\delta=0,1;\,\varepsilon=1-\delta)} \\[0.1cm] \eqref{qh3d0} \end{array} $ };
\draw[white!] (-2.2,3) circle (0.01pt)
node[above=1pt]{\color{black} $ \begin{array}{cc} Q3_{(\delta=0)} \\[0.1cm] \eqref{qq3d0} \end{array} $  };
\draw[thick,->] (5.4,3)--(5.4,1);
\draw[thick,->] (2.4,3)--(1.4,1);
\draw[thick,->] (-2.2,3)--(-3.0,1);
\draw[thick,->] (-4.2,3)--(-3.4,1);
\draw[thick,->] (0,3)--(1.0,1);
\draw[white!] (5.3,-0.5) circle (0.01pt)
node[above=1pt]{\color{black} $ \begin{array}{cc} H1_{(\varepsilon=1)} \\[0.1cm] \eqref{qh1e0b} \end{array} $ };
\draw[white!] (-3.2,-0.5) circle (0.01pt)
node[above=1pt]{\color{black} $ \begin{array}{cc} Q1_{(\delta=0)} \\[0.1cm] \eqref{qq1d0} \end{array} $};
\draw[white!] (1.2,-0.5) circle (0.01pt)
node[above=1pt]{\color{black} $ \begin{array}{cc} H1_{(\varepsilon=0)},H1_{(\varepsilon=1)} \\[0.1cm] \eqref{qh1e0a} \end{array} $ };
\end{tikzpicture}
\caption{Star-triangle relations corresponding to integrable quad equations $Q1_{(\delta=0)}$, $H1_{(\varepsilon=0)}$, $H1_{(\varepsilon=1)}$, obtained in this section as the algebraic limit of the hyperbolic cases of Section \ref{sec:hyplim}.}
\label{alglimsfig}
\end{figure}
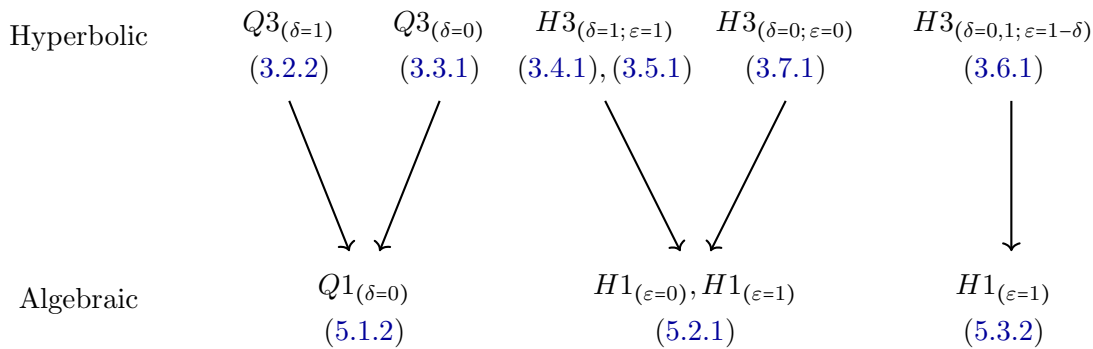

The algebraic limit from the hyperbolic level, may be thought of as a combination of the quasi-classical expansion for the hyperbolic cases in \eqref{HGFqcl}, and the rational limit in \eqref{ratasymp}.  Specifically, for
\begin{align}
\bb\to0^+\,,
\end{align}
the relevant asymptotics of the hyperbolic gamma function in this limit are  given by\cite{Rains2009}
\begin{align}
\label{algasym}
\Log\frac{\Gamma_h\left(\frac{x}{2\pi\bb}+\ii \theta_i\bb;\bb\right)}{\Gamma_h\left(\frac{x}{2\pi\bb}+\ii \theta_j\bb;\bb\right)}=(\theta_j-\theta_i)\Log\left(2\cosh\left(\frac{x}{2}\right)\right)+O(\bb)\,,
\end{align}
for $x\in\mathbb{R}$, $0<\theta_i,\theta_j<\eta$.  Unlike the quasi-classical expansions at the hyperbolic \eqref{hypqcl} and rational \eqref{ratqcl} levels respectively, the exact form of the quasi-classical expansion for the algebraic cases considered in this section will be seen to vary from case to case.

\subsection{\texorpdfstring{$Q1_{(\delta=0)}$}{Q1(delta=0)} case}

\tocless\subsubsection{Star-triangle relation}

Let the spins $\spn_1,\spn_2,\spn_3$, and spectral parameters $\theta_1$, $\theta_3$, take values
\begin{align}
\label{algvals}
\spn_1,\spn_2,\spn_3\in\mathbb{R}\,,\qquad 0<\theta_1,\theta_3,\theta_1+\theta_3<\tfrac{1}{2}\,.
\end{align}
The star-triangle relation is given in this case by \cite{Zamolodchikov:1980mb}
\begin{align}
\label{qq1d0}
\begin{split}
\ds\int_\mathbb{R}d\spn_0\,\oW({\theta_1}\,|\,\spn_1,\spn_0)\,W({\theta_1+\theta_3}\,|\,\spn_2,\spn_0)\,\oW({\theta_3}\,|\,\spn_3,\spn_0)\phantom{\qquad\,.} \\
=R(\theta_1,\theta_3)\,W({\theta_1}\,|\,\spn_2,\spn_3)\,\oW({\theta_1+\theta_3}\,|\,\spn_1,\spn_3)\,W({\theta_3}\,|\,\spn_2,\spn_1)\,,
\end{split}
\end{align}
where the Boltzmann weights are 
\begin{align}
\label{quanq1d0}
\begin{split}
W(\theta\,|\,\spn_i,\spn_j)&=|\spn_i-\spn_j|^{-2\theta}\,, \\[0.1cm]
\oW(\theta\,|\,\spn_i,\spn_j)&=W(\tfrac{1}{2}-\theta\,|\,\spn_i,\spn_j)\,,
\end{split}
\end{align}
and
\begin{align}
\label{qq1d0r}
R(\theta_1,\theta_3)=\sqrt{\pi}\frac{\Gamma(\theta_1)\,\Gamma(\frac{1}{2}-(\theta_1+\theta_3))\,\Gamma(\theta_3)}{\Gamma(\frac{1}{2}-\theta_1)\,\Gamma(\theta_1+\theta_3)\,\Gamma(\frac{1}{2}-\theta_3)}\,.
\end{align}

Equation \eqref{qq1d0} is the star-triangle relation for the $D=1$ Zamolodchikov fishnet model \cite{Zamolodchikov:1980mb}.  The connection to $Q1_{(\delta=0)}$ was previously given in \cite{Bazhanov:2016ajm,Kels:2017fyt}.  

The star-triangle relation \eqref{qq1d0} may be obtained as a limit of \eqref{qq3d1} ($Q3_{(\delta=1)}$ case), or \eqref{qq3d0} ($Q3_{(\delta=0)}$ case).  To see this, first the asymptotics \eqref{algasym} may be used to obtain the following integral from the star-triangle relation \eqref{q3d1int} ($Q3_{(\delta=1)}$ case):  

\begin{prop} \label{Selbergprop1}
Let the variables $\spn_i$, $\theta_1$, $\theta_3$, take values given in \eqref{hypvals}, and without loss of generality, choose $|\spn_3|>|\spn_1|$.  Then
\begin{align}
\label{Selberglim1}
\begin{split}
\lim_{\bb\to0^+}(8\pi\bb)\EXP^{|\spn_3|\bb^{-2}}I_{14}(\theta_1\bb,\theta_3\bb\,|\,\spn_1(2\pi\bb)^{-1},\spn_2(2\pi\bb)^{-1},\spn_3(2\pi\bb)^{-1})\phantom{\,,} \\[0.1cm]
=\int_{[\spn_1,\spn_3]}d\spn_0\,S(\spn_0)\,W_2(\theta_1\,|\,\spn_1,\spn_0)\,W_1(\theta_1+\theta_3\,|\,\spn_2,\spn_0)\,W_2(\theta_3\,|\,\spn_3,\spn_0)\,,
\end{split}
\end{align}
where $I_{14}(\theta_1,\theta_3\,|\,\spn_1,\spn_2,\spn_3)$ is defined in \eqref{q3d1int}, and the Boltzmann weights on the right hand side are defined by
\begin{align}
\label{quan1d01}
\begin{split}
W_1(\theta\,|\,\spn_i,\spn_j)&=\left|\cosh(\spn_i)+\cosh(\spn_j)\right|^{-2\theta}\,, \\[0.1cm]
W_2(\theta\,|\,\spn_i,\spn_j)&=\left|\cosh(\spn_i)-\cosh(\spn_j)\right|^{2\theta-1}\,,
\end{split}
\end{align}
and
\begin{align}
\label{quan1d01s}
S(\spn)=\left|\sinh(\spn)\right|\,.
\end{align}
\end{prop}

\begin{proof}

First the change of integration variable $\spn_0\to\spn_0(2\pi\bb)^{-1}$ should be applied to the integral \eqref{q3d1int}. Using \eqref{algasym}, and the functional equations \eqref{HGFidents} for the hyperbolic gamma function, the relevant asymptotics of the Boltzmann weights \eqref{qq3d1} for $\bb\to0$ are
\begin{align}
\begin{split}
\Log W\left(\theta\bb\,\Big|\,\frac{\spn_i}{2\pi\bb},\frac{\spn_j}{2\pi\bb}\right)&=-2\theta\Log\left(2\left(\cosh(\spn_i)+\cosh(\spn_j)\right)\right)+O(\bb)\,, \\[0.2cm]
\Log \oW\left(\theta\bb\,\Big|\,\frac{\spn_i}{2\pi\bb},\frac{\spn_j}{2\pi\bb}\right)&=\frac{-\max(|\spn_i|,|\spn_j|)}{\bb^2}+(2\theta-1)\Log\left(2\left|\cosh(\spn_i)-\cosh(\spn_j)\right|\right)+O(\bb)\,.
\end{split}
\end{align}
For the given case of $|\spn_3|>|\spn_1|$, the left hand side of \eqref{Selberglim1} decays exponentially outside $\spn_0\in[-|\spn_3|,-|\spn_1|]$, and $\spn_0\in[|\spn_1|,|\spn_3|]$, as $\bb\to0^+$, and by evenness of \eqref{q3d1rho}, the contribution to the integral from both of the respective regions are the same.  The result \eqref{Selberglim1} then follows by dominated convergence.

\end{proof}

A similar limit may also be taken for the star-triangle relation \eqref{q3d0int} ($Q3_{(\delta=0)}$ case).

\begin{prop} \label{Selbergprop2}
Let the variables $\spn_i$, $\theta_1$, $\theta_3$, take values given in \eqref{hypvals}, and without loss of generality, choose $\spn_3>\spn_1$.  Then
\begin{align}
\label{Selberglim2}
\begin{split}
\lim_{\bb\to0^+}(4\pi\bb)\EXP^{(\spn_3-\spn_1)\bb^{-2}}I_{6,Q}(\theta_1\bb,\theta_3\bb\,|\,\spn_1(\pi\bb)^{-1},\spn_2(\pi\bb)^{-1},\spn_3(\pi\bb)^{-1})\phantom{\,,} \\[0.1cm]
=\int_{[\spn_1,\spn_3]}d\spn_0\,W_2(\theta_1\,|\,\spn_1,\spn_0)\,W_1(\theta_1+\theta_3\,|\,\spn_2,\spn_0)\,W_2(\theta_3\,|\,\spn_3,\spn_0)\,,
\end{split}
\end{align}
where $I_{6,Q}(\theta_1,\theta_3\,|\,\spn_1,\spn_2,\spn_3)$ is defined in \eqref{q3d0int}, and the Boltzmann weights on the right hand side are defined by
\begin{align}
\label{quan1d02}
\begin{split}
W_1(\theta\,|\,\spn_i,\spn_j)&=\left|\cosh\left(\spn_i-\spn_j\right)\right|^{-2\theta}\,, \\[0.2cm]
W_2(\theta\,|\,\spn_i,\spn_j)&=\left|\sinh\left(\spn_i-\spn_j\right)\right|^{2\theta-1}\,.
\end{split}
\end{align}
\end{prop}

The proof of Proposition \ref{Selbergprop2}, follows analogously to Proposition \ref{Selbergprop1}.

The following Selberg-type integral \cite{Selberg} is obtained following a simple change of variables after taking either the same limit \eqref{Selberglim1}, on both sides of \eqref{qq3d1} ($Q3_{(\delta=1)}$ case), or the same limit \eqref{Selberglim2}, on both sides of \eqref{qq3d0} ($Q3_{(\delta=0)}$ case).

\begin{cor}

For the values $\spn_2<\spn_1<\spn_3$,
\begin{align}
\label{qybeq1d0}
\begin{split}
\int_{[\spn_1,\spn_3]}d\spn_0\,\oW({\theta_1}\,|\,\spn_1,\spn_0)\,W({\theta_1+\theta_3}\,|\,\spn_2,\spn_0)\,\oW({\theta_3}\,|\,\spn_3,\spn_0)\phantom{\,.} \\
=R(\theta_1,\theta_3)\,W({\theta_1}\,|\,\spn_2,\spn_3)\,\oW({\theta_1+\theta_3}\,|\,\spn_1,\spn_3)\,W({\theta_3}\,|\,\spn_2,\spn_1)\,,
\end{split}
\end{align}
where the Boltzmann weights are defined in \eqref{quanq1d0}, and
\begin{align}
R(\theta_1,\theta_3)&=\frac{\Gamma(2\theta_1)\Gamma(2\theta_3)}{\Gamma(2(\theta_1+\theta_3))}\,.
\end{align}

\end{cor}

The equation \eqref{qybeq1d0} does not quite have the desired form of the star-triangle relation \eqref{YBEsym}, since the integration is taken over $[\spn_1,\spn_3]$, rather than $\mathbb{R}$ (or some subset independent of $\spn_1,\spn_2,\spn_3$).  For the purposes here of deriving 3D-consistent quad equations from a quasi-classical expansion, this is not a problem, and the $Q1_{(\delta=0)}$ equation may be derived directly from \eqref{qybeq1d0}.  However at the quantum level this is problematic, since the star-triangle relation represents a partition function involving integration over the values of the central spin variable $\spn_0$, which cannot depend on the values of other spin variables. 

To obtain the desired integration over $\mathbb{R}$ a sum of three copies of \eqref{qybeq1d0} can be taken, each with different choices of the respective variables, to arrive at the star-triangle relation \eqref{qq1d0}.  First, the equation \eqref{qybeq1d0} directly gives the contribution to \eqref{qq1d0} over the integration region $\spn_0\in [\spn_1,\spn_3]$.  Second, the equation \eqref{qybeq1d0} with the substitution
\begin{align}
\spn_1\to \spn_2\,,\qquad \spn_2\to \spn_3\,,\qquad \spn_3\to \spn_1\,,\qquad \theta_1\to\theta_3\,,\qquad\theta_3\to \tfrac{1}{2}-\theta_1-\theta_3\,,
\end{align}
gives a contribution to \eqref{qq1d0} over the integration region $\spn_0\in [\spn_2,\spn_1]$.  The final equation uses a substitution to obtain the contribution to \eqref{qq1d0} over the integration region $\spn_0\in (-\infty,\spn_2]\cup [\spn_3,\infty)$.  The exact form of the substitution depends on whether the values $\spn_1,\spn_2,\spn_3$, are respectively greater or less than zero.  For example, for $\spn_2<\spn_1<0<\spn_3$, a substitution for the region $\spn_0\in (-\infty,\spn_2]\cup [\spn_3,\infty)$ is
\begin{align}
\spn\to \spn^{-1}\,,\quad \spn_1\to (\spn_2)^{-1}\,,\quad \spn_2\to (\spn_3)^{-1}\,,\quad \spn_3\to {\spn_1}^{-1}\,,\quad \theta_1\to \tfrac{1}{2}-\theta_1-\theta_3\,,\quad \theta_3\to\theta_1\,.
\end{align}
For other values, a uniform shift of the variables $\spn,\spn_1,\spn_2,\spn_3$, can be applied to \eqref{qybeq1d0} in order to make use of the above substitution (the equation \eqref{qybeq1d0} is invariant under such a shift).  While the above considerations are for the case of $\spn_2<\spn_1<\spn_3$, other cases follow from this case by simple changes of variables.  For example setting $\spn_2\leftrightarrow \spn_1$, $\theta_1\to\frac{1}{2}-\theta_1-\theta_3$, results in the same equation \eqref{qq1d0} for $\spn_1<\spn_2<\spn_3$, while setting $\spn_1\to \spn_3, \spn_2\to \spn_1, \spn_3\to \spn_2$, $\theta_1\to\theta_3, \theta_3\to \frac{1}{2}-\theta_1-\theta_3$, results in the same equation \eqref{qq1d0} for $\spn_1<\spn_3<\spn_2$, and other cases are similar.

\tocless\subsubsection{Classical integrable equations}

The quasi-classical expansion for this case involves setting
\begin{align}
\label{algqclq1}
\im(\theta_1)=\alpha_1\hbar^{-1}\,,\qquad\im(\theta_3)=\alpha_3\hbar^{-1}\,,
\end{align}
as well as relabelling $(\spn_0,\spn_1,\spn_2,\spn_3)\to(x_0,x_1,x_2,x_3)$, and then considering the quasi-classical expansion for $\hbar\rightarrow0^+$.  In this limit the classical variables are 
\begin{align}
x_1,x_2,x_3\in\mathbb{R}\,,\qquad\alpha_1,\alpha_3\in\mathbb{R}\,.
\end{align}

The Lagrangian function for this case is defined by
\begin{align}
\label{lagq1d0}
\mathcal{L}(\alpha\,|\,x_i,x_j)=\alpha\Log\left|x_i-x_j\right|-\frac{\alpha}{2}\Log\left|\alpha\right|\,.
\end{align}
and satisfies
\begin{align}
\lag(\alpha\,|\,x_i,x_j)=\lag(\alpha\,|\,x_j,x_i)\,,\qquad\lag(\alpha\,|\,x_i,x_j)=-\lag(-\alpha\,|\,x_i,x_j)\,.
\end{align}

The leading order $O(\hbar^{-1})$ quasi-classical expansion \eqref{algqclq1} of the integrand of \eqref{qq1d0} is
\begin{align}
\label{q1d0qcl}
\begin{split}
\Log \left(\oW(\alpha_1\hbar^{-1}\,|\,x_1,x_0)\,W((\alpha_1+\alpha_3)\hbar^{-1}\,|\,x_2,x_0)\,\oW(\alpha_3\hbar^{-1}\,|\,x_3,x_0)\right)+\phantom{\,,} \\[0.1cm]
\frac{\alpha_1\Log|\alpha_1|+\alpha_3\Log|\alpha_3|-(\alpha_1+\alpha_3)\Log|\alpha_1+\alpha_3|}{2\hbar}=\hbar^{-1}\astr{x_0}+O(1)\,,
\end{split}
\end{align}
where
\begin{align}
\astr{x_0}=\ol(\alpha_1\,|\,x_1,x_0)+\lag(\alpha_1+\alpha_3\,|\,x_2,x_0)+\ol(\alpha_3\,|\,x_3,x_0)\,,
\end{align}
and
\begin{align}
\ol(\alpha\,|\,x_i,x_j)=\lag(-\alpha\,|\,x_i,x_j)\,.
\end{align}

The saddle point three-leg equation \eqref{3legsym} is then given by
\begin{align}
\label{Q1d03leg}
\left.\frac{\partial \astr{x}}{\partial x}\right|_{x=x_0}\hspace{-0.4cm}=\ovphi(\alpha\,|\,x_1,x_0)+\varphi(\alpha_1+\alpha_3\,|\,x_2,x_0)+\ovphi(\alpha_3\,|\,x_3,x_0)=0\,,
\end{align}
where $\varphi(\alpha\,|\,x_i,x_j)$, is defined by
\begin{align}
\varphi(\alpha\,|\,x_i,x_j)=\frac{\alpha}{x_i-x_j}\,,
\end{align}
and
\begin{align}
\ovphi(\alpha\,|\,x_i,x_j)=\vphi(-\alpha\,|\,x_i,x_j)\,.
\end{align}

The equation \eqref{Q1d03leg} is a three-leg form of $Q1_{(\delta=0)}$, arising as the equation for the saddle point of the star-triangle relation \eqref{qq1d0} in the limit \eqref{algqclq1}.  This was also previously obtained in \cite{Bazhanov:2016ajm,Kels:2017fyt}.  With the following change of variables
\begin{align}
x=x_0,\qquad u=x_1,\qquad y=x_2,\qquad v=x_3,\qquad \alpha=-\alpha_1,\qquad\beta=-(\alpha_1+\alpha_3),
\end{align}
the three-leg equation \eqref{Q1d03leg} may be written in the form
\begin{align}
Q(x,u,y,v;\alpha,\beta)=0\,,
\end{align}
where 
\begin{align}
\label{Q1d0}
Q(x,u,y,v;\alpha,\beta)=\alpha(x-y)(u-v)-\beta(x-u)(y-v)\,.
\end{align}
This is identified as $Q1_{(\delta=0)}$ \cite{ABS}, an affine-linear quad equation that satisfies the 3D-consistency condition exactly as defined in Section \ref{sec:overview}.

\subsection{\texorpdfstring{$H1_{(\varepsilon=1)}$}{H1(epsilon=1)} case}

\tocless\subsubsection{Star-triangle relation}

Let the spins $\spn_1,\spn_2,\spn_3$, and spectral parameters $\theta_1$, $\theta_3$, take values given in \eqref{algvals}.  The star-triangle relation is given in this case by
\begin{align}
\label{qh1e0a}
\begin{split}
\ds\int_\mathbb{R} d\spn_0\,V({\theta_1}\,|\,\spn_1,\spn_0)\,V({\theta_1+\theta_3}\,|\,\spn_2,\spn_0)\,\oW({\theta_3}\,|\,\spn_0,\spn_3)\phantom{\quad\,,}  \\
\ds=R(\theta_3)\,V({\theta_1}\,|\,\spn_2,\spn_3)\,V({\theta_1+\theta_3	}\,|\,\spn_1,\spn_3)\,W({\theta_3}\,|\,\spn_2,\spn_1)\,,
\end{split}
\end{align}
where the Boltzmann weights are
\begin{align}
\label{quanh1e0a}
\begin{split}
V(\theta\,|\,\spn_i,\spn_j)&\ds=\EXP^{\ii \spn_i\spn_j}\,, \\
\oW(\theta\,|\,\spn_i,\spn_j)&\ds=\left|2\sinh\frac{\spn_i-\spn_j}{2}\right|^{2\theta-1}, \\
W(\theta\,|\,\spn_i,\spn_j)&\ds=\frac{\Gamma(\frac{1}{2}+\ii(\spn_i+\spn_j)-\theta)\,\Gamma(\frac{1}{2}-\ii(\spn_i+\spn_j)-\theta)}{\Gamma(\frac{1}{2}+\ii(\spn_i+\spn_j))\,\Gamma(\frac{1}{2}-\ii(\spn_i+\spn_j))\,\Gamma(\frac{1}{2}+\theta)\,\Gamma(\frac{1}{2}-\theta)}\,, 
\end{split}
\end{align}
and
\begin{align}
\label{quanh1e0ar}
R(\theta)=2\pi\,\Gamma(2\theta)\,.
\end{align}
Each of the Boltzmann weights in \eqref{quanh1e0a} are symmetric upon the exchange of spins $\spn_i\leftrightarrow \spn_j$.

The integrand is
\begin{align}
O(\EXP^{\,|\spn_0-\spn_3|(\theta_3-\frac{1}{2})}\EXP^{\ii \spn_0(\spn_1+\spn_2))})\,,\qquad \spn_0\to\pm\infty\,,
\end{align}
and particularly \eqref{qh1e0a} is absolutely convergent for the values \eqref{algvals}.

The star-triangle relation \eqref{qh1e0a}, is related to the following limit of \eqref{qh3d02} ($H3_{(\delta=0;\varepsilon=0)}$ case)

\begin{prop} \label{h1e0aprop}

Let the variables $\spn_1,\spn_2,\spn_3$, $\theta_1$, $\theta_3$, take values given in \eqref{hypvals}.  Then
\begin{align}
\label{h1e0alim}
\begin{split}
\lim_{\bb\to0^+}(2\pi\bb)\EXP^{\spn_3(2\bb^2)^{-1}}I_{2}(\theta_1\bb,\theta_3\bb\,|\,\spn_1\bb,\spn_2\bb+\ii(2\bb)^{-1},\spn_3(2\pi\bb)^{-1})\phantom{\,,} \\[0.1cm]
=\int_{(-\infty,\spn_3]} d\spn_0\,V({\theta_1}\,|\,-\spn_1,\spn_0)\,V({\theta_1+\theta_3}\,|\,\spn_2,\spn_0)\,\oW({\theta_3}\,|\,\spn_0,\spn_3)\,,
\end{split}
\end{align}
where $I_{2}(\theta_1,\theta_3\,|\,\spn_1,\spn_2,\spn_3)$ is defined in \eqref{h3d02int} ($H3_{(\delta=0;\varepsilon=0)}$ case), and the Boltzmann weights on the right hand side are defined in \eqref{quanh1e0a}.
\end{prop}

\begin{proof}
Following a change of variables $\spn_0\to\spn_0(2\pi\bb)^{-1}$, with the use of the asymptotic formula \eqref{algasym} it is seen that the left hand side of \eqref{h1e0alim} decays exponentially fast outside $\spn_0\in(-\infty,\spn_3]$.  The result \eqref{h1e0alim} then follows by dominated convergence.
\end{proof}

The following equation is obtained by taking the same limit \eqref{h1e0alim}, on both sides of \eqref{qh3d02}.

\begin{cor}

\begin{align}
\label{betafunca}
\begin{split}
\ds\int_{(-\infty,\spn_3]} d\spn_0\,V({\theta_1}\,|\,-\spn_1,\spn_0)\,V({\theta_1+\theta_3}\,|\,\spn_2,\spn_0)\,\oW({\theta_3}\,|\,\spn_0,\spn_3)\phantom{\,,}  \\[0.1cm]
\ds=R(\theta_3)\,V({\theta_1}\,|\,\spn_2,\spn_3)\,V({\theta_1+\theta_3	}\,|\,-\spn_1,\spn_3)\,W({\theta_3}\,|\,\spn_2,\spn_1)\,,
\end{split}
\end{align}
where $\oW(\theta\,|\,\spn_i,\spn_j)$, and $V(\theta\,|\,\spn_i,\spn_j)$, are defined in \eqref{quanh1e0a}, and
\begin{align}
W(\theta\,|\,\spn_i,\spn_j)&\ds=\frac{\Gamma(\frac{1}{2}+\ii(\spn_i-\spn_j)-\theta)}{\Gamma(\frac{1}{2}+\ii(\spn_i-\spn_j)+\theta)}\,,
\end{align}
\begin{align}
R(\theta)=\Gamma(2\theta)\,.
\end{align}

\end{cor}

The star-triangle relation \eqref{qh1e0a} is equivalent to a sum of two copies of \eqref{betafunca}, with different choices of variables for the two respective cases.  First, the equation \eqref{betafunca} directly gives the contribution to \eqref{qh1e0a} over the integration region $\spn_0\in (-\infty,\spn_3]$.  The remaining contribution to $\spn_0\in [\spn_3,\infty)$ comes from \eqref{betafunca} with the change of variables
\begin{align}
\spn_0\to-\spn_0\,,\qquad \spn_1\to-\spn_1\,,\qquad \spn_2\to-\spn_2\,,\qquad \spn_3\to-\spn_3\,.
\end{align}

Similarly, the star-triangle relation \eqref{qh1e0a}, is related to the following limit of \eqref{qh3d1} ($H3_{(\delta=1;\varepsilon=1)}$ case)

\begin{prop} \label{h1e0aprop2}
For the values \eqref{hypvals},
\begin{align}
\label{h1e0alim2}
\begin{split}
\lim_{\bb\to0^+}(2\pi\bb)\EXP^{|\spn_3|\bb^{-2}} I_{10}(\theta_1\bb,\theta_3\bb\,|\,\spn_1\bb,\spn_2\bb+\ii(2\bb)^{-1},\spn_3(2\pi\bb)^{-1})\phantom{\,,} \\
=\int_{[0,\,|\spn_3|]} d\spn_0\,S(\spn_0)\,\oV({\theta_1}\,|\,\spn_1,\spn_0)\,V({\theta_1+\theta_3}\,|\,\spn_2,\spn_0)\,\oW({\theta_3}\,|\,\spn_0,\spn_3)\,,
\end{split}
\end{align}
where $I_{10}(\theta_1,\theta_3\,|\,\spn_1,\spn_2,\spn_3)$ is defined in \eqref{h3d1int} ($H3_{(\delta=1;\varepsilon=1)}$ case), and the Boltzmann weights on the right hand side are 
\begin{align}
\begin{split}
V(\theta\,|\,\spn_i,\spn_j)&\ds=\left|2\sinh(\tfrac{\spn_j}{2})\right|^{2(\ii \spn_i-\theta)}\,, \\
\oV(\theta\,|\,\spn_i,\spn_j)&\ds=\left|2\sinh(\tfrac{\spn_j}{2})\right|^{2(\theta-\ii \spn_i)-1}\,, \\
\oW(\theta\,|\,\spn_i,\spn_j)&\ds=\left|2(\cosh(\spn_i)-\cosh(\spn_j))\right|^{2\theta-1}\,, 
\end{split}
\end{align}
and
\begin{align}
S(\spn)=|\sinh(\spn)|\,.
\end{align}

\end{prop}

The steps of the proof of Proposition \ref{h1e0aprop2}, are analogous to the respective steps for Proposition \ref{h1e0aprop}. Then taking the same limit \eqref{h1e0alim2}, on both sides of \eqref{qh3d1}, results in \eqref{betafunca} (after a straightforward change of variables).

Finally, note that the equation \eqref{betafunca} is equivalent to the standard Euler beta function.  Indeed, after a change of variables $\EXP^{\spn_0/2}= x$, $\EXP^{\spn_3/2}= x_3$, $\spn_1=x_1$, $\spn_2=x_2$, \eqref{betafunca} may be written as the following integral identity\footnote{Note that \eqref{betafunc0} is also obtained from the $\spn_1=0$, $\spn_2\to\infty$ limit of the star-triangle relation \eqref{qq1d0} (for the $Q1_{(\delta=0)}$ case), with $\theta_1\to\frac{\ii(\spn_1+\spn_2)-\theta_3+1}{2}$.}
\begin{align}
\label{betafunc0}
\int_{[0,x_3]}dx\,\rho_a(\theta_3\,|\,x_1,x_2,x_3;x)=x_3^{2\ii(x_1+x_2)}\,\Gamma(2\theta_3)\frac{\Gamma(\frac{1}{2}+\ii(x_1+x_2)-\theta_3)}{\Gamma(\frac{1}{2}+\ii(x_1+x_2)+\theta_3)}\,,
\end{align}
where
\begin{align}
\label{h1e0arho}
\rho_a(\theta_3\,|\,x_1,x_2,x_3;x)=x^{2\ii(x_1+x_2)}(xx_3)^{1-2\theta_3}\,|x^2-x^2_3|^{2\theta_3-1}\,.
\end{align}
The variable $x_3$ is eliminated with the substitution $x\to x_3\sqrt{x}$, and \eqref{betafunc0} gets the form of the standard Euler beta function
\begin{align}
\label{EBF}
\int_{[0,1]}dx\, x^{a-1}(1-x)^{b-1}=\frac{\Gamma(a)\,\Gamma(b)}{\Gamma(a+b)}\,,\qquad \re(a)\,,\re(b)>0\,,
\end{align}
where $a=\ii(x_1+x_2)-\theta_3+\frac{1}{2}$, and $b=2\theta_3$.  Consequently the star-triangle relation \eqref{qh1e0a} is equivalent to a sum of two copies of the beta function \eqref{betafunc0}, with one copy having spin variables of the opposite sign, with respect to the other copy.

\tocless\subsubsection{Classical integrable equations}

In this limit the classical variables $x_1,x_2,x_3$, $\alpha_1,\alpha_3$, are
\begin{align}
x_1,x_2,x_3\in\mathbb{R}\,,\qquad \alpha_1,\alpha_3\in\mathbb{R}\,.
\end{align}

The Lagrangian functions for this case are defined by
\begin{align}
\label{lagh1e0a}
\begin{split}
\ds\Lambda(\alpha\,|\,x_i,x_j)&=\ii x_ix_j\,, \\
\ds{\ol}(\alpha\,|\,x_i,x_j)&=2\ii\alpha\Log\left|\sinh\left(\frac{x_i-x_j}{2}\right)\right|.
\end{split}
\end{align}

The Lagrangians are each symmetric upon the exchange of spins $x_i\leftrightarrow x_j$, and also satisfy
\begin{align}
{\ol}(\alpha\,|\,x_i,x_j)=-{\ol}(-\alpha\,|\,x_i,x_j)\,.
\end{align}

The quasi-classical expansion of \eqref{qh1e0a}, involves setting
\begin{align}
\label{algqclh1}
\begin{gathered}
\spn_0= x,\qquad \spn_3= x_3,\qquad \spn_1= x_1\hbar^{-1},\qquad \spn_2= x_2\hbar^{-1},\qquad \im(\theta_3)=\alpha_3\hbar^{-1},
\end{gathered}
\end{align}
for $\hbar\to0$.  Then the leading order $O(\hbar^{-1})$ quasi-classical expansion \eqref{algqclh1} of the integrand of \eqref{qh1e0a} becomes
\begin{align}
\label{h1e0aqcl}
\Log \rho_a(\alpha_3\hbar^{-1}\,|\,x_1\hbar^{-1},x_2\hbar^{-1},\EXP^{\frac{x_3}{2}};\EXP^{\frac{x}{2}})=\hbar^{-1}\astr{x_0}+O(1)\,,
\end{align}
where $\rho_a(\alpha_3\,|\,x_1,x_2,x_3;x)$ is given in \eqref{h1e0arho}, and
\begin{align}
\astr{x_0}=\Lambda(\alpha_1\,|\,x_1,x_0)+\Lambda(\alpha_1+\alpha_3\,|\,x_2,x_0)+\ol(\alpha_3\,|\,x_0,x_3)\,.
\end{align}

The saddle point three-leg equation \eqref{3legasym} is then given by
\begin{align}
\label{h1e0a3leg}
\frac{1}{\ii}\hspace{-0.05cm}\left.\frac{\partial \astr{x}}{\partial x}\right|_{x=x_0}\hspace{-0.4cm}=\phi(\alpha_1\,|\,x_1,x_0)+\phi(\alpha_1+\alpha_3\,|\,x_2,x_0)+\vphi(\alpha_3\,|\,x_3,x_0)=0\,,
\end{align}
where $\phi(\alpha\,|\,x_i,x_j)$, and $\vphi(\alpha\,|\,x_i,x_j)$, are defined by
\begin{align}
\phi(\alpha\,|\,x_i,x_j)=x_i\,,\qquad\vphi(\alpha\,|\,x_i,x_j)=-\alpha\coth\left(\frac{x_i-x_j}{2}\right)\,.
\end{align}

The equation \eqref{h1e0a3leg} is a three-leg form of $H1_{(\varepsilon=1)}$, arising as the equation for the saddle point of the star-triangle relation \eqref{qh1e0a} in the limit \eqref{algqclh1}. With the following change of variables
\begin{align}
x=\EXP^{x_0},\qquad u=-x_1,\qquad y=x_2,\qquad v=\EXP^{x_3},\qquad \alpha=-\alpha_1,\qquad \beta=-(\alpha_1+\alpha_3),
\end{align}
the three-leg equation \eqref{h1e0a3leg} may be written in the form
\begin{align}
H(x,u,y,v;\alpha,\beta)=0\,,
\end{align}
where 
\begin{align}
\label{H1d1}
H(x,u,y,v;\alpha,\beta)=(u-y)(x-v)+(\alpha-\beta)(x+v)\,.
\end{align}
This is identified as $H1_{(\varepsilon=1)}$ \cite{ABS2}, an affine-linear quad equation that satisfies the 3D-consistency condition exactly as defined in Section \ref{sec:overview}.

\subsection{\texorpdfstring{$H1_{(\varepsilon=1)}$}{H1(epsilon=1)} case (alternate form)}

\tocless\subsubsection{Star-triangle relation}

Let the spins $\spn_1,\spn_2,\spn_3$ and spectral parameters $\theta_1$, $\theta_3$, take values
\begin{align}
\label{algvalsalt}
\spn_1,\spn_2,\spn_3\in\mathbb{R}\,,\qquad 0<\theta_1,\theta_3\,.
\end{align} 
The star-triangle relation is given in this case by
\begin{align}
\label{qh1e0b}
\begin{split}
\ds\int_{\mathbb{R}} d\spn_0\,\oV({\theta_1}\,|\,\spn_1,\spn_0)\,V({\theta_1+\theta_3}\,|\,\spn_2,\spn_0)\,\oW({\theta_3}\,|\,\spn_0,\spn_3)\phantom{\quad\,,}  \\
\ds=R(\theta_3)\,V({\theta_1}\,|\,\spn_2,\spn_3)\,\oV({\theta_1+\theta_3	}\,|\,\spn_1,\spn_3)\,W({\theta_3}\,|\,\spn_2,\spn_1)\,,
\end{split}
\end{align}
where the Boltzmann weights are
\begin{align}
\label{quanh1e0b}
\begin{split}
V(\theta\,|\,\spn_i,\spn_j)&\ds=|\spn_i|^{\ii \spn_j-\theta}\,, \\
\oV(\theta\,|\,\spn_i,\spn_j)&\ds=|\spn_i|^{\theta-\ii \spn_j}\,, \\
\oW(\theta\,|\,\spn_i,\spn_j)&\ds=\Gamma(\theta+\ii(\spn_i-\spn_j))\,\Gamma(\theta-\ii(\spn_i-\spn_j))\,, \\
W(\theta\,|\,\spn_i,\spn_j)&\ds=\left(|\spn_i|+|\spn_j|\right)^{-2\theta}\,,
\end{split}
\end{align}
and
\begin{align}
\label{quanh1e0br}
R(\theta)=2\pi\,\Gamma(2\theta)\,.
\end{align}
The Boltzmann weights satisfy
\begin{align}
W(\theta\,|\,\spn_i,\spn_j)=W(\theta\,|\,\spn_j,\spn_i)\,,\quad \oW(\theta\,|\,\spn_i,\spn_j)=\oW(\theta\,|\,\spn_j,\spn_i)\,,\quad W(\theta\,|\,\spn_i,\spn_j)W(-\theta\,|\,\spn_i,\spn_j)=1\,.
\end{align}

The asymptotics of the integrand of \eqref{qh1e0b} are
\begin{align}
O(\EXP^{-\pi|\spn_0|}\,\EXP^{\ii \spn_0(\Log |\spn_2|-\Log |\spn_1|)}))\,,\qquad \spn\to\pm\infty\,,
\end{align}
and particularly, \eqref{qh1e0b} is absolutely convergent for the values \eqref{algvalsalt}.

The above star-triangle relation \eqref{qh1e0b}, may be obtained in the following limit of \eqref{qh3d1alt} ($H3_{(\delta=1;\varepsilon=1)}$ case):

\begin{prop} \label{h1e0bprop}
For the values \eqref{hypvals},
\begin{align}
\label{h1e0blim}
\begin{split}
\lim_{\bb\to0^+}\frac{4\pi^2\bb}{(2\pi\bb^2)^{2\theta_3}}\EXP^{\frac{\spn_1}{2\bb^2}}I_{6,H}(\theta_1\bb,\theta_3\bb\,|\,\spn_1(2\pi\bb)^{-1},\spn_2(2\pi\bb)^{-1},\spn_3\bb)\phantom{\,,} \\[0.1cm]
=\int_\mathbb{R} d\spn_0\,\oV({\theta_1}\,|\,\spn_1,\spn_0)\,V({\theta_1+\theta_3}\,|\,\spn_2,\spn_0)\,\oW({\theta_3}\,|\,\spn_0,\spn_3)\,,
\end{split}
\end{align}
where $I_{6,H}(\theta_1,\theta_3\,|\,\spn_1,\spn_2,\spn_3)$ is defined in \eqref{h3d12int} ($H3_{(\delta=1;\varepsilon=1)}$ case), and the Boltzmann weights on the right hand side are
\begin{align}
\begin{split}
V(\theta\,|\,\spn_i,\spn_j)&\ds=(2\cosh(\tfrac{\spn_i}{2}))^{2(\ii \spn_j-\theta)}\,, \\
\oV(\theta\,|\,\spn_i,\spn_j)&\ds=\left|2\sinh(\tfrac{\spn_i}{2})\right|^{2(\theta-\ii \spn_j)-1}\,, \\[0.1cm]
\oW(\theta\,|\,\spn_i,\spn_j)&\ds=\Gamma(\theta+\ii(\spn_i-\spn_j))\,\Gamma(\theta-\ii(\spn_i-\spn_j))\,.
\end{split}
\end{align}

\end{prop}

The steps of the proof are analogous to the steps for the rational limits taken in Section \ref{sec:ratlim}.  The star-triangle relation \eqref{qh1e0b} is obtained by taking the same limit \eqref{h1e0blim}, on both sides of \eqref{qh3d1alt}, and making a straightforward change of variables to get the form of the Boltzmann weights \eqref{quanh1e0b}.

Similarly, the star-triangle relation \eqref{qh1e0b}, may be obtained in a similar limit from \eqref{qh3d0} ($H3_{(\delta=0,1;\varepsilon=1-\delta)}$ case).  For this case, to make use of the asymptotic formula \eqref{algasym}, a ratio of Boltzmann weights coming from the right hand side of \eqref{qh3d0} is required as follows:

\begin{prop} \label{h1e0bprop2}
For  the values \eqref{hypvals}, and $\spn_1<0$,
\begin{align}
\label{h1e0blim2}
\begin{split}
\lim_{\bb\to0^+}\left(\frac{4\pi^2\bb}{(2\pi\bb^2)^{2\theta_3}}\frac{V(\theta_1\,|\,\spn_2,\spn_3)\,\oV((\theta_1+\theta_3)\,|\,\spn_1,\spn_3)}{W(\theta_1\bb\,|\,\spn_2(2\pi\bb)^{-1},\spn_3\bb)\,\oW((\theta_1+\theta_3)\bb\,|\,\spn_1(2\pi\bb)^{-1},\spn_3\bb)}\right.\phantom{\,,}  \\[0.2cm]
\times \left. I_{4}(\theta_1\bb,\theta_3\bb\,|\,\spn_1(2\pi\bb)^{-1},\spn_2(2\pi\bb)^{-1},\spn_3\bb)\right)\phantom{\,,} \\[0.2cm]
=\int_\mathbb{R} d\spn_0\,\oV({\theta_1}\,|\,\spn_1,\spn_0)\,V({\theta_1+\theta_3}\,|\,\spn_2,\spn_0)\,W_1({\theta_3}\,|\,\spn_0,\spn_3)\,,
\end{split}
\end{align}
where $I_{4}(\theta_1,\theta_3\,|\,\spn_1,\spn_2,\spn_3)$ is defined in \eqref{h3d0int} ($H3_{(\delta=0,1;\varepsilon=1-\delta)}$ case), and the Boltzmann weights in the denominator of the first line of \eqref{h1e0blim2} are defined in \eqref{h3d0quans}. The remaining Boltzmann weights are
\begin{align}
\begin{split}
V(\theta\,|\,\spn_i,\spn_j)&\ds=(\EXP^{-\spn_i}+1)^{\ii \spn_j-\theta}\,, \\
\oV(\theta\,|\,\spn_i,\spn_j)&\ds=(\EXP^{-\spn_i}-1)^{\theta-\ii \spn_j}\,, \\
W_1(\theta\,|\,\spn_i,\spn_j)&\ds=\Gamma(\theta+\ii(\spn_i-\spn_j))\,\Gamma(\theta-\ii(\spn_i-\spn_j))\,.
\end{split}
\end{align}

\end{prop}

The condition $\spn_1<0$, is neccessary for the convergence of \eqref{h1e0blim2}, and the rest of the proof is analogous to the limits taken for the rational cases of the Section \ref{sec:ratlim}. Then taking the same limit \eqref{h1e0blim2}, on the right hand side of \eqref{qh3d0}, results in the star-triangle relation \eqref{qh1e0b} (up to a straightforward change of variables).

Besides the exponential factors, the star-triangle relation \eqref{qh1e0b} resembles the form of the star-triangle relations given for the rational cases in Section \ref{sec:ratlim}, and indeed it may be seen to be a particular case of Barnes's integral formula \cite{Barnes1908} for the hypergeometric function $\!\!~_2F_1$ \cite{WW} (the same integral formula is equivalent to the star-triangle relation for the $H2_{(\varepsilon=0)}$ case \eqref{qh2d0}).  

To see this, first note that the star-triangle relation \eqref{qh1e0b}, is explicitly written as
\begin{align}
\label{qh1e0c}
\int_\mathbb{R}d\spn_0\,\rho_{b}(\theta_1,\theta_3\,|\,\spn_1,\spn_2,\spn_3;\spn_0)=2\pi\,\Gamma(2\theta_3)\,|\spn_1|^{\theta_1+\theta_3-\ii \spn_3} |\spn_2|^{\ii \spn_3-\theta_1} (|\spn_1|+|\spn_2|)^{-2\theta_3}\,,
\end{align}
where
\begin{align}
\label{h1e0brho}
\rho_{b}(\theta_1,\theta_3\,|\,\spn_1,\spn_2,\spn_3;\spn_0)=\Gamma(\theta_3+\ii(\spn_3-\spn_0))\,\Gamma(\theta_3-\ii(\spn_3-\spn_0))\,|\spn_1|^{\theta_1-\ii \spn_0} |\spn_2|^{\ii \spn_0-\theta_1-\theta_3}\,.
\end{align}
This equation may be written in a form that is independent of $\theta_1$, since the $\theta_1$ dependency trivially cancels from both sides of \eqref{qh1e0c}.  Similarly, the variable $\spn_3$ may be eliminated with a change of variable of $\spn_0\to \spn_0-\spn_3$, while also the expression \eqref{qh1e0c} only depends on the ratio $\mu=\frac{\spn_1}{\spn_2}$.  Consequently, after the corresponding change of variables, the star-triangle relation \eqref{qh1e0b} can be written explicitly as the following integral identity for a product of gamma functions:
\begin{align}
\label{EBFI}
\int_{\mathbb{R}}dz\,\Gamma(\theta_3+\ii z)\,\Gamma(\theta_3-\ii z)\,|\mu|^{-(\theta_3+\ii z)}=2\pi\,\Gamma(2\theta_3)(1+|\mu|)^{-2\theta_3}\,.
\end{align}
In this form, it is seen that this is a particular integral evaluation formula involving the standard Beta function \eqref{EBF}, with the choice $a=\theta_1+\ii z$, $b=\theta_3-\ii z$.  As mentioned above, the integral \eqref{EBFI} is a special case of Barnes's integral formula \cite{Barnes1908} for the hypergeometric function $\!\!~_2F_1$ \cite{WW}.  The integral \eqref{EBFI} may be directly evaluated as a sum over residues in the upper half plane, or also by using the fact that it is the inverse Mellin transform of the beta function \eqref{EBF}.

\tocless\subsubsection{Classical integrable equations}

In this limit the classical variables $x_1,x_2,x_3$, $\alpha_1,\alpha_3$, are 
\begin{align}
x_1,x_2,x_3\in\mathbb{R}\,,\qquad 0<\alpha_1,\alpha_3\,.
\end{align}

The Lagrangian functions for this case are defined by
\begin{align}
\label{lagh1e0b}
\begin{split}
\ds\Lambda(\alpha\,|\,x_i,x_j)&=(\ii x_j-\alpha)\Log|x_i| \,, \\
\ds{\ol}(\alpha\,|\,x_i,x_j)&=\gamma(x_i-x_j-\ii\alpha)+\gamma(x_j-x_i-\ii\alpha)-\gamma(-2\ii\alpha) \,,
\end{split}
\end{align}
where $\gamma(z)$ is defined in \eqref{gammadef}.

The latter function satisfies
\begin{align}
{\ol}(\alpha\,|\,x_i,x_j)={\ol}(\alpha\,|\,x_j,x_i)\,.
\end{align}

The quasi-classical expansion of \eqref{qh1e0b}, involves setting
\begin{align}
\label{algqclh12}
\begin{gathered}
\spn_0=x,\qquad\spn_3=x_3,\qquad\spn_1=x_1\hbar^{-1},\qquad \spn_2= x_2\hbar^{-1},\qquad \theta_3=\alpha_3\hbar^{-1},
\end{gathered}
\end{align}
and then taking $\hbar\to0$.  Then the leading order $O(\hbar^{-1})$ quasi-classical expansion \eqref{algqclh12} of the integrand of \eqref{qh1e0b} is
\begin{align}
\label{h1e0aqcl2}
\begin{split}
(1+\Log\hbar)2\alpha_3\hbar^{-1}-\Log\hbar+\Log \rho_{b}(\alpha_3\hbar^{-1}\,|\,x_1\hbar^{-1},x_2\hbar^{-1},x_3;x)\phantom{\,,} \\[0.1cm]
=\hbar^{-1}\astr{x_0}+O(1)\,,
\end{split}
\end{align}
where $\rho_{b}(\alpha_3\,|\,x_1,x_2,x_3;x)$ is given in \eqref{h1e0brho},
\begin{align}
\astr{x_0}=\olam(\alpha_1\,|\,x_1,x_0)+\Lambda(\alpha_1+\alpha_3\,|\,x_2,x_0)+\ol(\alpha_3\,|\,x_0,x_3)\,,
\end{align}
and
\begin{align}
\olam(\alpha\,|\,x_i,x_j)=-\Lambda(\alpha\,|\,x_i,x_j)\,.
\end{align}

The saddle point three-leg equation \eqref{3legasym} is then given by
\begin{align}
\label{h1e0b3leg}
\frac{1}{\ii}\hspace{-0.05cm}\left.\frac{\partial \astr{x}}{\partial x}\right|_{x=x_0}\hspace{-0.4cm}=\ovphib(\alpha_1\,|\,x_1,x_0)+\phi(\alpha_1+\alpha_3\,|\,x_2,x_0)+\vphi(\alpha_3\,|\,x_3,x_0)=0\,,
\end{align}
where $\phi(\alpha\,|\,x_i,x_j)$, and $\vphi(\alpha\,|\,x_i,x_j)$, are defined by
\begin{align}
\begin{split}
\phi(\alpha\,|\,x_i,x_j)&=\Log|x_i|\,, \\
\vphi(\alpha\,|\,x_i,x_j)&=\Log(\alpha-\ii(x_i-x_j))-\Log(\alpha+\ii(x_i-x_j))\,,
\end{split}
\end{align}
and
\begin{align}
\ovphib(\alpha\,|\,x_i,x_j)=-\phi(\alpha\,|\,x_i,x_j)\,.
\end{align}

The equation \eqref{h1e0b3leg} is a three-leg form of $H1_{(\varepsilon=1)}$, arising as the equation for the saddle point of the star-triangle relation \eqref{qh1e0b} in the limit \eqref{algqclh12}. With the following change of variables
\begin{align}
x=x_0,\qquad u=-|x_1|,\qquad y=|x_2|,\qquad v=x_3,\qquad \alpha=-\ii\alpha_1,\qquad \beta=-\ii(\alpha_1+\alpha_3),
\end{align}
the exponential of the three-leg equation \eqref{h1e0b3leg} may be written in the form
\begin{align}
H(x,u,y,v;\alpha,\beta)=0\,,
\end{align}
where 
\begin{align}
\label{H1d12}
H(x,u,y,v;\alpha,\beta)=(u-y)(x-v)+(\alpha-\beta)(u+y)\,.
\end{align}
This is identified as $H1_{(\varepsilon=1)}$ \cite{ABS2}, an affine-linear quad equation that satisfies the 3D-consistency condition exactly as defined in Section \ref{sec:overview}.

\subsection{\texorpdfstring{$H1_{(\varepsilon=0)}$}{H1(epsilon=0)} case}

\tocless\subsubsection{Star-triangle relation}

The star-triangle relation \eqref{qh1e0a}, in fact also provides a quantum counterpart of $H1_{(\varepsilon=0)}$, by using a slightly different choice of quasi-classical expansion.

\tocless\subsubsection{Classical integrable equations}

For this case the classical variables $x_1,x_2,x_3$, $\alpha_1,\alpha_3$, are 
\begin{align}
x_1,x_2,x_3\in\mathbb{R}\,,\qquad \alpha_1,\alpha_3\in\mathbb{R}\,.
\end{align}

The Lagrangian functions for this case are defined by
\begin{align}
\label{lagh1e0b3}
\begin{split}
\ds\Lambda(\alpha\,|\,x_i,x_j)=\ii x_ix_j\,,\qquad
\ds\ol(\alpha\,|\,x_i,x_j)=2\ii\alpha\Log|x_i-x_j|\,.
\end{split}
\end{align}
The Lagrangians are each symmetric upon the exchange of spins $x_i\leftrightarrow x_j$, and ${\ol}(\alpha\,|\,x_i,x_j)$, also satisfies
\begin{align}
{\ol}(\alpha\,|\,x_i,x_j)=-{\ol}(-\alpha\,|\,x_i,x_j)\,.
\end{align}

The quasi-classical expansion of \eqref{qh1e0a} (different to \eqref{algqclh1}), involves setting
\begin{align}
\label{algqclh13}
\spn_0= x_0\hbar,\qquad \spn_1= x_1\hbar^{-2},\qquad \spn_2= x_2\hbar^{-2},\qquad \spn_3= x_3\hbar, \qquad\im(\theta_3)\to \alpha_3\hbar^{-1},
\end{align}
for $\hbar\to0$.  Then the leading order $O(\hbar^{-1})$ quasi-classical expansion \eqref{algqclh13} of the integrand of \eqref{qh1e0a} becomes
\begin{align}
\label{h1e0b3qcl}
\begin{split}
-2(\hbar^{-1}\ii\alpha_3+\re(\theta_3)-\tfrac{1}{2})\Log\hbar+\Log \rho_a(\alpha_3\hbar^{-1}\,|\,x_1\hbar^{-2},x_2\hbar^{-2},x_3\hbar;x\hbar)\phantom{\,,}
\\[0.1cm]
=\hbar^{-1}\astr{x_0}+O(1)\,,
\end{split}
\end{align}
where $\rho_a(\alpha_3\,|\,x_1,x_2,x_3;x)$ is given in \eqref{h1e0arho}, and
\begin{align}
\astr{x_0}=\Lambda(\alpha_1\,|\,x_1,x_0)+\Lambda(\alpha_1+\alpha_3\,|\,x_2,x_0)+\ol(\alpha_3\,|\,x_0,x_3)\,.
\end{align}

The saddle point three-leg equation \eqref{3legasym} is then given by
\begin{align}
\label{h1e0b3leg3}
\frac{1}{\ii}\hspace{-0.05cm}\left.\frac{\partial \astr{x}}{\partial x}\right|_{x=x_0}\hspace{-0.4cm}=\phi(\alpha_1\,|\,x_1,x_0)+\phi(\alpha_1+\alpha_3\,|\,x_2,x_0)+\vphi(\alpha_3\,|\,x_3,x_0)=0\,,
\end{align}
where $\phi(\alpha\,|\,x_i,x_j)$, and $\vphi(\alpha\,|\,x_i,x_j)$, are defined by
\begin{align}
\phi(\alpha\,|\,x_i,x_j)=x_i\,,\qquad\vphi(\alpha\,|\,x_i,x_j)=\frac{2\alpha}{x_j-x_i}\,.
\end{align}

The equation \eqref{h1e0b3leg3} is a three-leg form of $H1_{(\varepsilon=0)}$, arising as the equation for the saddle point of the star-triangle relation \eqref{qh1e0a} in the limit \eqref{algqclh13}. With the following change of variables
\begin{align}
x=x_0,\qquad u=-x_1,\qquad y=x_2,\qquad v=x_3,\qquad\alpha=-\alpha_1,\qquad \beta=-(\alpha_1+\alpha_3),
\end{align}
the three-leg equation \eqref{h1e0b3leg3} may be written in the form
\begin{align}
H(x,u,y,v;\alpha,\beta)=0\,,
\end{align}
where 
\begin{align}
\label{h1e0b}
H(x,u,y,v;\alpha,\beta)=(u-y)(x-v)+2(\beta-\alpha)\,.
\end{align}
This is identified as $H1_{(\varepsilon=0)}$ \cite{ABS2}, an affine-linear quad equation that satisfies the 3D-consistency condition exactly as defined in Section \ref{sec:overview}.

\section{Conclusion}

This paper shows how all 3D-consistent $H$-type quad equations in the ABS classification, may be derived from the quasi-classical expansion of respective counterpart solutions of an asymmetric form of the star-triangle relation \eqref{YBEasym}.  The results here provide an almost systematic way to derive 3D-consistent quad equations from the star-triangle relation; once a solution of the latter is known, the major computations only involve determining the right choice of the quasi-classical expansion which gives a suitable asymptotic series in $\hbar$ of the form \eqref{qclybe}. In all cases an affine-linear 3D-consistent quad equation is then obtained via a change of variables from the equation of the saddle point.  The star-triangle relations \eqref{YBEsym}, \eqref{YBEasym}, thus provide natural path integral quantizations of symmetric $Q$-type, and asymmetric $H$-type ABS equations respectively.  The results of this paper extend the Yang-Baxter/3D-consistency correspondence that was previously given for the $Q$-type ABS equations \cite{Bazhanov:2016ajm}, to the entire ABS list.  

Each solution of the asymmetric form of the star-triangle relation \eqref{YBEasym} appearing in this paper is new (apart from the case \eqref{qh3d0} for $H3_{(\delta=0,1;\,\varepsilon=1-\delta)}$ \cite{BAZHANOV2018509}).  Each of the star-triangle relations were also shown to have an equivalent interpretation as a transformation formula for a univariate hypergeometric integral, as summarised in Appendix \ref{app:hyper}.  The different solutions of both forms of the star-triangle relations \eqref{YBEsym}, \eqref{YBEasym}, were derived in this paper from certain limits of the top level hyperbolic star-triangle relation \eqref{qq3d1} (for the $Q3_{(\delta=1)}$ case), resulting in many new limits of hyperbolic hypergeometric integrals.  Table \ref{hyptable} highlights some interesting new interpretations in terms of integrability for some relatively old classical hypergeometric integrals, such as Barnes's integral formulas (equivalently quantum $H2_{(\varepsilon=1)}$, $H2_{(\varepsilon=0)}$, $H1_{(\varepsilon=1)}$) which first appeared in the 1900's, and the Euler beta function (quantum $H1_{(\varepsilon=0)}$, $H1_{(\varepsilon=1)}$) which goes back much further.  Such hypergeometric integral formulas are the basis of the link between the two types of integrable systems based on the YBE, and 3D-consistency respectively.

There still remains much more that can be done in developing further aspects of the Yang-Baxter/3D-consistency correspondence presented in this paper.  Some examples include determining statistical properties of quantum integrable systems that are associated to the continuous spin solutions of the star-triangle relation  \eqref{YBEsym}, \eqref{YBEasym}, determining the relevance of higher orders of a quasi-classical expansion to classical integrability based on 3D-consistency, interpreting other forms of hypergeometric integrals in terms of integrability, and understanding the appearance of classical integrable equations for supersymmetric gauge theories \cite{Yamazaki:2012cp,Kels:2017vbc}.  These are some of the important topics that are planned for future research.

\section*{Acknowledgements}

The author is an overseas researcher under Postdoctoral Fellowship of Japan Society for the Promotion of Science (JSPS).  The author thanks Vladimir Bazhanov, and Sergey Sergeev, for helpful discussions related to this work.

\begin{appendices}

\section{Summary of connection to hypergeometric integrals}\label{app:hyper}

Each of the star-triangle relations in this paper, after a change of variables can be seen to be equivalent to transformation formulas for univariate hypergeometric integrals.  The connection with the hypergeometric integral theory is stated throughout this paper for each case, and is also summarised in Table \ref{hyptable}, along with citations for where the formulas originally appeared for hypergeometric integrals (to the best of the authors knowledge).  It appears that the star-triangle relations for the ``hyperbolic Barnes's $\!\!~_2F_1$  integral''  \eqref{qh3d02}, and ``hyperbolic Barnes's first lemma'' \eqref{qh3d0}, have not previously been related to hypergeometric integrals, however they did previously appear in the context of integrable lattice models in \cite{Bazhanov:2007mh} (as a self-duality relation), and \cite{BAZHANOV2018509} respectively.

{\renewcommand{\arraystretch}{1.1}
\begin{table}[htbp]
\begin{center}
\medskip
\begin{tabular}{c|c|c}
  Hypergeometric Integral & STR & Quad Equation\\
 \hline
 \hline
 Elliptic beta integral \cite{SpiridonovEBF} &   \cite{Bazhanov:2010kz}  & $Q4$ \cite{Bazhanov:2010kz} \\
   \hline
   \hline
 Hyperbolic beta integral \cite{STOKMAN2005119} & \eqref{qq3d1}   & $Q3_{(\delta=1)}$ \eqref{Q3d1} \\ 
   \hline
 Hyperbolic Saalsch\"{u}tz integral \cite{BultThesis} & \eqref{qq3d0}   & $Q3_{(\delta=0)}$ \eqref{Q3d0} \\ 
   \hline
 Hyperbolic Askey-Wilson integral \cite{STOKMAN2005119,Ruijsenaars2003} & \eqref{qh3d1}   & $H3_{(\delta=1;\,\varepsilon=1)}$ \eqref{H3d1} \\ 
   \hline
 Hyperbolic Saalsch\"{u}tz integral \cite{BultThesis} & \eqref{qh3d1alt}  & $H3_{(\delta=1;\,\varepsilon=1)}$ (alt.) \eqref{H3d12} \\ 
   \hline
 Hyperbolic Barnes's first lemma & \eqref{qh3d0}  & $H3_{(\delta=0,1;\,\varepsilon=1-\delta)}$ \eqref{H3d01} \\ 
 \hline
 Hyperbolic Barnes's $\!\!~_2F_1$  integral & \eqref{qh3d02}   & $H3_{(\delta=0;\,\varepsilon=0)}$ \eqref{H3d02} \\ 
   \hline
   \hline
 Askey integral \cite{Askey1989} & \eqref{qq2}   & $Q2$ \eqref{Q2} \\
  \hline
 Barnes's second lemma \cite{Barnes1910} & \eqref{qq1d1}   & $Q1_{(\delta=1)}$ \eqref{Q1d1} \\
  \hline
 de Branges-Wilson integral \cite{DeBranges1972,Wilson1980} & \eqref{qh2d1}   & $H2_{(\varepsilon=1)}$ \eqref{H2d1} \\
  \hline
 Barnes's second lemma \cite{Barnes1910} & \eqref{qh2d1alt}   & $H2_{(\varepsilon=1)}$ (alt.) \eqref{H2d12} \\ 
   \hline
 Barnes's first lemma \cite{Barnes1908} & \eqref{qh2d0}  & $H2_{(\varepsilon=0)}$ \eqref{H2d0} \\ 
   \hline
   \hline
 Selberg integral \cite{Selberg} &  \eqref{qq1d0}  & $Q1_{(\delta=0)}$ \eqref{Q1d0} \\ 
 \hline
 Euler beta integral &  \eqref{qh1e0a}  & $H1_{(\varepsilon=1)}$  \eqref{H1d1} \\ 
 \hline
 Barnes's $\!\!~_2F_1$ integral \cite{Barnes1908,WW} & \eqref{qh1e0b}   & $H1_{(\varepsilon=1)}$ (alt.) \eqref{H1d12} \\ 
 \hline
 Euler beta integral & \eqref{qh1e0a}   & $H1_{(\varepsilon=0)}$ \eqref{h1e0b}
\end{tabular}
\caption{Relation between various solutions of the star-triangle relation (STR) \eqref{YBEsym}, \eqref{YBEasym}, hypergeometric integrals, and 3D-consistent quad equations of \cite{ABS,ABS2} in a quasi-classical expansion.}
\label{hyptable}
\end{center}
\end{table}
}

\section{Examples of non-integrable cases}
\label{app:notYBE}

In Section \ref{sec:hyplim}, five limits were given of the top level hyperbolic star-triangle relation \eqref{qq3d1}, which resulted in different hyperbolic solutions of the star-triangle relation.  There are many other limits of \eqref{q3d1int} that result in new identities, which may not necessarily take the form of the star-triangle relation \eqref{YBEsym}, or \eqref{YBEasym}.  These cases still result in affine-linear classical equations in a quasi-classical expansion, however these classical equations do not satisfy the 3D-consistency integrability condition.  Thus these examples are considered to be non-integrable, and this suggests that the 3D-consistent equations will only arise from the specific form of the star-triangle relations given in \eqref{YBEsym}, \eqref{YBEasym}.  Four non-integrable hyperbolic cases are considered in this Appendix.

\subsection{Examples of non-integrable cases 1: triangle identities}

Let the variables $\spn_1,\spn_2,\spn_3$, and $\theta_1,\theta_3$, take values \eqref{hypvals}.  The triangle identity is given by
\begin{align}
\label{qtri}
\int_{\mathbb{R}}d\spn\,S(\spn_0)\,\oW(\theta_1\,|\,\spn_1,\spn_0)\,\oW(\theta_3\,|\,\spn_3,\spn_0)=R(\theta_1,\theta_3)\,\oW(\theta_1+\theta_3\,|\,\spn_1,\spn_3)\,,
\end{align}
where the Boltzmann weight is the same as in \eqref{genfadvol}
\begin{align}
\label{qtriBW1}
\begin{split}
\iW(\theta\,|\,\spn_i,\spn_j)&=\ds\frac{\Gamma_h(\spn_i+\spn_j+\ii\theta;\bb)}{\Gamma_h(\spn_i+\spn_j-\ii\theta;\bb)}\frac{\Gamma_h(\spn_i-\spn_j+\ii\theta;\bb)}{\Gamma_h(\spn_i-\spn_j-\ii\theta;\bb)}\,, \\
\oW(\theta\,|\,\spn_i,\spn_j)&=W(\eta-\theta\,|\,\spn_i,\spn_j)\,,
\end{split}
\end{align}
$S(\spn)$ is
\begin{align}
\label{qtris1}
\begin{split}
S(\spn)&=\ds\frac{1}{2}\Gamma_h(-2\spn-\ii\eta;\bb)\,\Gamma_h(2\spn-\ii\eta;\bb) \\
&=\ds2\sinh(2\pi \spn\bb)\sinh(2\pi \spn/\bb)\,,
\end{split}
\end{align}
and
\begin{align}
R(\theta_1,\theta_3)=\frac{\Gamma_h(\ii(\eta-2\theta_1);\bb)\,\Gamma_h(\ii(\eta-2\theta_3);\bb)}{\Gamma_h(\ii(\eta-2(\theta_1+\theta_3));\bb)}\,.
\end{align}
From \eqref{HGFinflim}, the asymptotics of the integrand of \eqref{qtri}, with \eqref{qtriBW1} are
\begin{align}
\label{qtri1xinf}
O(\EXP^{-4\pi(\eta-(\theta_1+\theta_3))|\spn_0|})\,,\qquad \spn_0\to\pm\infty\,,
\end{align}
and particularly the integral in \eqref{qtri}, with \eqref{qtriBW1}, is absolutely convergent for the values \eqref{hypvals}.

\begin{prop} \label{tri1prop}
For the values \eqref{q3d1rvals},
\begin{align}
\label{tri1lim}
\begin{split}
\lim_{\kappa\to\infty}\EXP^{4\pi(\theta_1+\theta_3)\kappa+4\pi(\theta_1+\theta_3)\spn_2}I_{14}(\theta_1,\theta_3\,|\,\spn_1,\spn_2+\kappa,\spn_3)\phantom{\,,} \\
=\int_{\mathbb{R}}d\spn_0\,S(\spn_0)\,\oW(\theta_1\,|\,\spn_1,\spn_0)\,\oW(\theta_3\,|\,\spn_3,\spn_0)\,,
\end{split}
\end{align}
where $I_{14}(\theta_1,\theta_3\,|\,\spn_1,\spn_2,\spn_3)$ is defined in \eqref{q3d1int}, and the Boltzmann weights on the right hand side are defined in \eqref{qtriBW1}, and \eqref{qtris1}.

\end{prop}

The steps of the proof of Proposition \ref{tri1prop}, are analogous to the respective steps for Proposition \ref{h3d1prop}.  The triangle identity \eqref{qtri}, with \eqref{qtriBW1}, follows from Proposition \ref{tri1prop}, by using \eqref{HGFinflim} to take the same limit of the right hand side of \eqref{qq3d1}.

The triangle identity \eqref{qtri} is also satisfied by the Boltzmann weight (related to \eqref{fadvol})
\begin{align}
\label{qtriBW2}
\begin{split}
\iW(\theta\,|\,\spn_i,\spn_j)&=\ds\EXP^{-2\pi(\eta-\theta)(\spn_i+\spn_j)}\,\frac{\Gamma_h(\spn_i-\spn_j+\ii\theta;\bb)}{\Gamma_h(\spn_i-\spn_j-\ii\theta;\bb)}\,, \\
\oW(\theta\,|\,\spn_i,\spn_j)&=W(\eta-\theta\,|\,\spn_i,\spn_j)\,,
\end{split}
\end{align}
where now
\begin{align}
\label{qtris2}
S(\spn)=1\,.
\end{align}
From \eqref{HGFinflim}, the asymptotics of the integrand of \eqref{qtri}, with \eqref{qtriBW2} are
\begin{align}
\label{qtri2xinf}
\begin{split}
O(\EXP^{4\pi(\eta-(\theta_1+\theta_3))\spn_0})\,,&\qquad \spn_0\to-\infty\,, \\
O(\EXP^{-4\pi\eta \spn_0})\,,&\qquad \spn_0\to+\infty\,,
\end{split}
\end{align}
and particularly the integral in \eqref{qtri}, with \eqref{qtriBW2}, is absolutely convergent for the values \eqref{hypeta}, \eqref{hypvals}.

\begin{prop} \label{tri2prop}
For the values \eqref{q3d1rvals},
\begin{align}
\label{tri2lim}
\begin{split}
\lim_{\kappa\to\infty}\EXP^{4\pi\eta\kappa+2\pi(\eta-2\theta_1)\spn_1+2\pi(\eta-2\theta_3)\spn_3}I_{14}(\theta_1,\theta_3\,|\,\spn_1+\kappa,\spn_2,\spn_3+\kappa)\phantom{\,,} \\
=\int_{\mathbb{R}}d\spn_0\,S(\spn_0)\,\oW(\theta_1\,|\,\spn_1,\spn_0)\,\oW(\theta_3\,|\,\spn_3,\spn_0)\,,
\end{split}
\end{align}
where $I_{14}(\theta_1,\theta_3\,|\,\spn_1,\spn_2,\spn_3)$ is defined in \eqref{q3d1int}, and the Boltzmann weights on the right hand side are defined in \eqref{qtriBW2}, and \eqref{qtris2}.

\end{prop}

The steps of the proof of Proposition \ref{tri2prop}, are analogous to the respective steps for Proposition \ref{q3d0prop}.  The triangle identity \eqref{qtri}, with \eqref{qtriBW2}, follows from Proposition \ref{tri2prop}, by using \eqref{HGFinflim} to take the same limit of the right hand side of \eqref{qq3d1}.

The triangle identity \eqref{qtri} did not appear before.  The case of \eqref{qtri} with \eqref{qtriBW1}, is equivalent to the hyperbolic analogue of the Askey-Wilson integral \cite{STOKMAN2005119,Ruijsenaars2003} (the same as for \eqref{qh3d1} for the $H3_{(\delta=1,\varepsilon=1)}$ case), while the case of \eqref{qtri} with \eqref{qtriBW2}, is equivalent to the hyperbolic analogue of the Barnes integral formula (the same as for \eqref{qh3d0} for the $H3_{(\delta=0,1;\,\varepsilon=1-\delta)}$ case).  The triangle identity \eqref{qtri} may be considered to be a special case of the star-triangle relation \eqref{YBEsym}, where $W(\theta\,|\,\spn_i,\spn_j)=1$, and has the graphical representation given in Figure \ref{trianglefig}.

\begin{figure}[h]
\centering

\begin{tikzpicture}[scale=1.9]

\draw[-,thick] (-2,0)--(-2,1);
\draw[-,thick] (-2,0)--(-2.87,-0.5);
\draw[<-,thick,dotted] (-2.9,-0.25)--(-1.8,-0.25);
\fill[white!] (-1.8,-0.25) circle (0.5pt)
node[right=1.5pt]{\color{black}\small $p$};
\draw[<-,thick,dotted] (-2.7,-0.71)--(-1.7,1.02);
\fill[white!] (-1.7,1.02) circle (0.5pt)
node[right=1.5pt]{\color{black}\small $q$};
\draw[<-,thick,dotted] (-1.75,0.06)--(-2.3,1.02);
\fill[white!] (-2.3,1.02) circle (0.5pt)
node[left=1.5pt]{\color{black}\small $r$};
\fill (-2,0) circle (1.5pt)
node[below=1.5pt]{\color{black} $\spn_0$};
\filldraw[fill=black,draw=black] (-2,1) circle (1.5pt)
node[above=1.5pt] {\color{black} $\spn_1$};
\filldraw[fill=black,draw=black] (-2.87,-0.5) circle (1.5pt)
node[below=1.5pt] {\color{black} $\spn_3$};

\fill[white!] (0.05,0.3) circle (0.01pt)
node[left=0.05pt] {\color{black}$=$};

\draw[-,thick] (2,1)--(1.13,-0.5);
\draw[<-,thick,,dotted] (1.1,0.25)--(2.05,0.25);
\fill[white!] (2.05,0.25) circle (0.5pt)
node[right=1.5pt]{\color{black}\small $p$};
\draw[<-,thick,dotted] (1.83,-0.23)--(1.32,0.67);
\fill[white!] (1.32,0.67) circle (0.5pt)
node[above=1.5pt]{\color{black}\small $r$};
\filldraw[fill=black,draw=black] (2,1) circle (1.5pt)
node[above=1.5pt]{\color{black} $\spn_1$};
\filldraw[fill=black,draw=black] (1.13,-0.5) circle (1.5pt)
node[left=1.5pt]{\color{black} $\spn_3$};
\end{tikzpicture}
\caption{Graphical representation of the triangle identity \eqref{qtri}, where edges correspond to Figure \ref{2boltzmannweights} with $\theta_1=q-r$, and $\theta_3=p-q$.}
\label{trianglefig}
\end{figure}
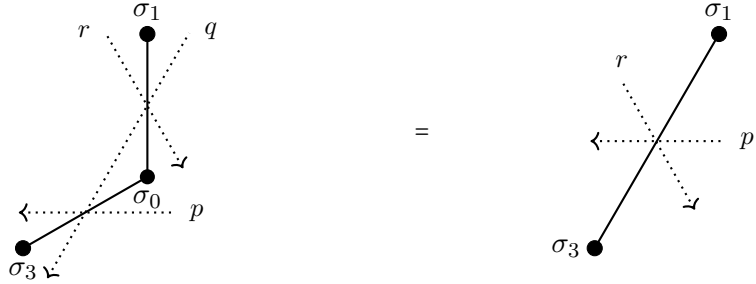

\subsubsection{Classical equations}

The triangle identity \eqref{qtri} also has a quasi-classical expansion for \eqref{hypqcl}, which will lead to some ``two-leg'' equations of motion, and associated affine-linear equations.  The latter will obviously not be a 3D-consistent quad equation, since \eqref{qtri} does not contain a factor that involves interaction between the two variables $\spn_2$, and $\spn_0$.

The Lagrangian function for the case \eqref{qtriBW1} is
\begin{align}
\label{lagtri1}
\begin{split}
\ds\mathcal{L}(\alpha\,|\,x_i,x_j)&=\ds\lie(\EXP^{x_i+x_j-\ii\alpha})+\lie(\EXP^{x_i-x_j-\ii\alpha})+\lie(\EXP^{-x_i+x_j-\ii\alpha})+\lie(\EXP^{-x_i-x_j-\ii\alpha}) \\
&\quad -\lie(\EXP^{-2\ii\alpha})+x_i^2+x_j^2+\pi\ii(x_i+x_j)-\frac{\pi^2}{2}\,,
\end{split}
\end{align}
while the Lagrangian function for the case \eqref{qtriBW2} is
\begin{align}
\label{lagtri2}
\begin{split}
\ds\mathcal{L}(\alpha\,|\,x_i,x_j)&=\ds+\lie(\EXP^{x_i-x_j-\ii\alpha})+\lie(\EXP^{-x_i+x_j-\ii\alpha})-\lie(\EXP^{-2\ii\alpha}) \\[0.1cm]
&\quad +\frac{(x_i-x_j)^2+(\pi-\alpha)^2}{2}-\ii\alpha(x_i+x_j)-\frac{2\pi^2}{3}\,.
\end{split}
\end{align}

Using the asymptotics \eqref{HGFqcl}, the leading order $O(\hbar^{-1})$ quasi-classical expansion \eqref{hypqcl} of the integrand of \eqref{qtri}, for both of the cases \eqref{lagtri1}, or \eqref{lagtri2}, takes the form
\begin{align}
\begin{split}
\frac{(\pi-\alpha_1-\alpha_3)^2-2(\alpha_1\alpha_3+\frac{\pi^2}{3})-\lie(\EXP^{-2\ii\alpha_1})-\lie(\EXP^{-2\ii\alpha_3})}{\ii\hbar}+\phantom{\,,} \\
\Log \rho(\tfrac{\alpha_1}{\sqrt{2\pi\hbar}},\tfrac{\alpha_3}{\sqrt{2\pi\hbar}}\,|\,\tfrac{x_1}{\sqrt{2\pi\hbar}},\tfrac{x_2}{\sqrt{2\pi\hbar}}\tfrac{x_3}{\sqrt{2\pi\hbar}};\tfrac{x_0}{\sqrt{2\pi\hbar}})
=(\ii\hbar)^{-1}\astr{x_0}+O(1)\,,
\end{split}
\end{align}
where 
\begin{align}
\rho(\theta_1,\theta_3\,|\,\spn_1,\spn_2,\spn_3;\spn_0)=S(\spn_0)\,\oW(\theta_1\,|\,\spn_1,\spn_0)\,\oW(\theta_3\,|\,\spn_3,\spn_0)\,,
\end{align}
and
\begin{align}
\astr{x_0}&=C(x_0)+\lag(\alpha_1\,|\,x_0,x_1)+\lag(\alpha_3\,|\,x_0,x_3)\,, 
\end{align}
where respectively for \eqref{lagtri1}:
\begin{align}
\label{ctri1}
\ds\iC(x)=2\pi\ii x\,,
\end{align}
and for \eqref{lagtri2}:
\begin{align}
\label{ctri2}
\ds\iC(x)=0\,.
\end{align}

The saddle point ``two-leg'' equation is then given by
\begin{align}
\label{ctri2leg}
\left.\frac{\partial \astr{x}}{\partial x}\right|_{x=x_0}=\varphi(\alpha_1\,|\,x_1,x_0)+\varphi(\alpha_3\,|\,x_3,x_0)=0\,,
\end{align}
where for the case of \eqref{qtriBW1} with \eqref{lagtri1}, $\varphi(\alpha\,|\,x_i,x_j)$ is defined by
\begin{align}
\label{ctri2leg1}
\begin{split}
\varphi(\alpha\,|\,x_i,x_j)=&\Log(1-\EXP^{x_i-x_j-\ii\alpha})-\Log(1-\EXP^{x_i+x_j-\ii\alpha})+ \\
&\Log(1-\EXP^{-x_i-x_j-\ii\alpha})-\Log(1-\EXP^{-x_i+x_j-\ii\alpha}) +\pi\ii+2x_j\,,
\end{split}
\end{align}
while for the case of \eqref{qtriBW2} with \eqref{lagtri2}, $\varphi(\alpha\,|\,x_i,x_j)$ is defined by
\begin{align}
\label{ctri2leg2}
\varphi(\alpha\,|\,x_i,x_j)=\Log(1-\EXP^{x_i-x_j-\ii\alpha})-\Log(1-\EXP^{x_j-x_i-\ii\alpha})-x_i+x_j-\ii\alpha\,.
\end{align}

With the following change of variables
\begin{align}
x=-\cosh(x_0),\qquad u=\cosh(x_1),\qquad v=\cosh(x_3),\qquad\alpha=\EXP^{-\ii\alpha_1},\qquad\beta=\EXP^{-\ii(\alpha_1+\alpha_3)},
\end{align}
the exponential of the two-leg equation \eqref{ctri2leg} with \eqref{ctri2leg1}, is seen to correspond to the equation
\begin{align}
\label{quadtri1}
x(\beta^2-1)\alpha+u(\beta^2-\alpha^2)+v(\alpha^2-1)\beta=0\,,
\end{align}
while for the change of variables
\begin{align}
x=-\EXP^{x_0},\qquad u=\EXP^{x_1},\qquad v=\EXP^{x_3},\qquad\alpha=\EXP^{-\ii\alpha_1},\qquad\beta=\EXP^{-\ii(\alpha_1+\alpha_3)},
\end{align}
the exponential of the two-leg equation \eqref{ctri2leg} with \eqref{ctri2leg2}, corresponds to the equation
\begin{align}
\label{quadtri2}
uv(\beta^2-1)\alpha+xv(\beta^2-\alpha^2)+xu(\alpha^2-1)\beta=0\,.
\end{align}
The equations \eqref{quadtri1}, \eqref{quadtri2}, are independent of $y$, and cannot satisfy the 3D-consistency condition defined in Section \ref{sec:overview}.  The equation \eqref{quadtri1} may be thought of as $Q3_{(\delta=1)}$ in \eqref{Q3d1} with $y\to\infty$, while the equation \eqref{quadtri2} may be thought of as $Q3_{(\delta=0)}$ in \eqref{Q3d0} with $y=0$.

\subsection{Examples of non-integrable cases 2}

Another example is the equation
\begin{align}
\label{notYBE1}
\begin{split}
\ds\int_{\mathbb{R}}d\spn_0\,S(\spn_0)\,\oW_1({\eta-\theta_1}\,|\,\spn_1,\spn_0)\,\oW_2({\theta_1+\theta_3}\,|\,\spn_2,\spn_0)\,W_1({\eta-\theta_3}\,|\,\spn_3,\spn_0)\phantom{\,.} \ds \\
\ds =R(\theta_1,\theta_3)\,W_1({\theta_1}\,|\,\spn_3,\spn_2)\,W_2(\eta-(\theta_1+\theta_3)\,|\,\spn_3,\spn_1)\,\oW_3({\theta_3}\,|\,\spn_1,\spn_2)\,,
\end{split}
\end{align}
where the variables are as defined in \eqref{hypvals}.  The equation \eqref{notYBE1} does not have the form of the star-triangle relation given in either \eqref{YBEsym}, or \eqref{YBEasym}.  Here
\begin{align}
\label{notYBE1quan}
\begin{array}{rclrcl}
\ds\oW_1(\theta\,|\,\spn_i,\spn_j)&\hspace{-0.2cm}=\hspace{-0.2cm}&\ds\EXP^{-B_+(\spn_j,-\spn_i,\theta)}\,\oW(\theta\,|-\spn_i,\spn_j)\,,\; &\ds W_1(\theta\,|\,\spn_i,\spn_j)&\hspace{-1cm}=\hspace{-1cm}&\ds\EXP^{B_+(\spn_i,\spn_j,\theta)}\,W(\theta\,|\,\spn_i,\spn_j), \\[0.2cm]
\ds\oW_2(\theta\,|\,\spn_i,\spn_j)&\hspace{-0.2cm}=\hspace{-0.2cm}&\ds\EXP^{B_+(\spn_j,\spn_i,\theta)}\,\oW(\theta\,|\,\spn_i,\spn_j)\,,\; &&&  \\[0.2cm]
\ds\oW_3(\theta\,|\,\spn_i,\spn_j)&\hspace{-0.2cm}=\hspace{-0.2cm}&\ds\EXP^{B_-(\spn_j,\spn_i,\theta)}\,\oW(\theta\,|\,\spn_i,\spn_j)\,,\; & \multicolumn{3}{l}{\smash{\raisebox{1\normalbaselineskip}{$\ds W_2(\theta\,|\,\spn_i,\spn_j)=\frac{\EXP^{-B_+(\spn_i,-\spn_j,\theta)}}{W(-\theta\,|\,\spn_i,\spn_j)}\,,$}}}
\end{array}
\end{align}
\begin{align}
\begin{split}
\oW(\theta\,|\,\spn_i,\spn_j)&=\ds\EXP^{-\overline{B}(\spn_i+\spn_j+\ii\theta)}\Gamma_h(\spn_i-\spn_j+\ii\theta;\bb)\,, \\[0.1cm]
W(\theta\,|\,\spn_i,\spn_j)&=\ds\frac{\Gamma_h(\spn_i+\spn_j+\ii\theta;\bb)}{\Gamma_h(\spn_i-\spn_j-\ii\theta;\bb)}\,,
\end{split}
\end{align}
and
\begin{align}
\label{notYBE1quans}
S(\spn)=\EXP^{4\pi\eta \spn},\qquad R(\theta_1,\theta_3)=\EXP^{\frac{\pi\ii}{2}\left(B_{2,2}(2\theta_3;\bb,\bb^{-1})+B_{2,2}(2(\theta_1+\theta_3);\bb,\bb^{-1})-B_{2,2}(2\theta_1;\bb,\bb^{-1})\right)}\,,
\end{align}
where $\eta$ is defined in \eqref{hypeta}, $B_{2,2}(z;\bb,\bb^{-1})$ is defined in \eqref{b22def}, $\overline{B}(z)$ is defined in \eqref{obdef}, and
\begin{align}
\label{bpmdef}
B_{\pm}(\theta;\spn_i,\spn_j)=\overline{B}(-\spn_i-\spn_j+\ii\theta)\pm\overline{B}(\spn_i-\spn_j+\ii\theta)\,.
\end{align}

From \eqref{HGFinflim}, the asymptotics of the integrand of \eqref{notYBE1} are
\begin{align}
\label{notYBE1xinf}
\begin{split}
O(\EXP^{4\pi\eta \spn_0})\,,\qquad \spn_0\to-\infty\,, \\
O(\EXP^{-2\pi(\eta-\theta_3+\ii(\spn_1-\spn_2))\spn_0})\,,\qquad \spn_0\to+\infty\,,
\end{split}
\end{align}
and particularly the integral in \eqref{notYBE1}, is absolutely convergent for the values \eqref{hypvals}. 

The equation \eqref{notYBE1} arises in the following limit of the star-triangle relation \eqref{qq3d1}:

\begin{prop} \label{notYBE1prop}
For the values \eqref{q3d1rvals},
\begin{align}
\label{notYBE1lim}
\begin{split}
\lim_{\kappa\to\infty}&\EXP^{4\pi(2\eta+\theta_3-\ii(\spn_1-\spn_2))\kappa} I_{14}(\theta_1+2\ii\kappa,\theta_3-\ii\kappa\,|\,\spn_1+\kappa,\spn_2+2\kappa,\spn_3) \\
&=\int_{\mathbb{R}}d\spn_0\,S(\spn_0)\,\oW_1({\eta-\theta_1}\,|\,\spn_1,\spn_0)\,\oW_2({\theta_1+\theta_3}\,|\,\spn_2,\spn_0)\,W_1({\eta-\theta_3}\,|\,\spn_3,\spn_0)\,,
\end{split}
\end{align}
where $I_{14}(\theta_1,\theta_3\,|\,\spn_1,\spn_2,\spn_3)$ is defined in \eqref{q3d1int}, and the Boltzmann weights on the right hand side are defined in \eqref{notYBE1quan}, \eqref{notYBE1quans}.

\end{prop}

The steps of the proof of Proposition \ref{notYBE1prop}, are analogous to the respective steps for Proposition \ref{q3d0prop}.  Equation \eqref{notYBE1}, follows from Proposition \ref{notYBE1prop}, by using \eqref{ratasymp} to take the same limit of the right hand side of \eqref{qq3d1}.

\subsubsection{Classical equations}

The Lagrangian functions for this case are defined by
\begin{align}
\label{lagnotYBE1}
\begin{split}
\Lambda(\alpha\,|\,x_i,x_j)&=\ds \lie(-\EXP^{x_i-x_j+\ii\alpha})+\frac{(x_i-x_j)^2}{2}\,,   \\
\lag(\alpha\,|\,x_i,x_j)&=\ds\lie(-\EXP^{x_i+x_j+\ii\alpha})+\lie(-\EXP^{-x_i+x_j+\ii\alpha})+\frac{(x_j+\ii\alpha)^2}{2}\,.
\end{split}
\end{align}

Using the asymptotics \eqref{HGFqcl}, the leading order $O(\hbar^{-1})$ quasi-classical expansion \eqref{hypqcl} of the integrand of \eqref{notYBE1} is
\begin{align}
\begin{split}
\frac{x_1^2-x_3^2+\ii(\pi-\alpha_1)x_1+(\alpha_1+\alpha_3)(\ii x_2+\alpha_3)+\pi(\alpha_1-\alpha_3)-\frac{\pi^2}{3}}{\ii\hbar}+\phantom{\,,} \\
\Log \rho(\tfrac{\alpha_1}{\sqrt{2\pi\hbar}},\tfrac{\alpha_3}{\sqrt{2\pi\hbar}}\,|\,\tfrac{x_1}{\sqrt{2\pi\hbar}},\tfrac{x_2}{\sqrt{2\pi\hbar}}\tfrac{x_3}{\sqrt{2\pi\hbar}};\tfrac{x_0}{\sqrt{2\pi\hbar}})
=(\ii\hbar)^{-1}\astr{x_0}+O(1)\,,
\end{split}
\end{align}
where 
\begin{align}
\rho(\theta_1,\theta_3\,|\,\spn_1,\spn_2,\spn_3;\spn_0)=S(\spn_0)\,\oW_1({\eta-\theta_1}\,|\,\spn_1,\spn_0)\,\oW_2({\theta_1+\theta_3}\,|\,\spn_2,\spn_0)\,W_1({\eta-\theta_3}\,|\,\spn_3,\spn_0)\,,
\end{align}
and
\begin{align}
\astr{x_0}&=\Lambda(\pi-\alpha_1\,|\,-x_1,x_0)+\Lambda(\alpha_1+\alpha_3\,|\,x_2,x_0)+\ol(\pi-\alpha_3\,|\,x_3,x_0)\,.
\end{align}

The saddle point three-leg equation is then
\begin{align}
\label{notYBE13leg}
\left.\frac{\partial \astr{x}}{\partial x}\right|_{x=x_0}\hspace{-0.45cm}=\ovphib(\alpha_1\,|\,x_1,x_0)+\phi(\alpha_1+\alpha_3\,|\,x_2,x_0)+\vphi(\alpha_3\,|\,x_3,x_0)=0\,,
\end{align}
where $\phi(\alpha\,|\,x_i,x_j)$, and $\vphi(\alpha\,|\,x_i,x_j)$, are defined by
\begin{align}
\begin{split}
\phi(\alpha\,|\,x_i,x_j)&=\Log(1+\EXP^{x_i-x_j+\ii\alpha})-x_i+x_j\,, \\
\vphi(\alpha\,|\,x_i,x_j)&=-\Log(1-\EXP^{x_i+x_j-\ii\alpha})-\Log(1-\EXP^{x_j-x_i-\ii\alpha})+\ii(\pi-\alpha)\,,
\end{split}
\end{align}
and
\begin{align}
\ovphib(\alpha\,|\,x_i,x_j)=\phi(\pi-\alpha\,|\,-x_i,x_j)\,.
\end{align}

With the following change of variables
\begin{align}
x=\EXP^{x_0},\qquad u=\EXP^{x_1},\qquad y=\EXP^{x_2},\qquad v=\cosh(x_3),\qquad\alpha=\EXP^{\ii\alpha_1},\quad\beta=\EXP^{\ii\alpha_3}\,,
\end{align}
the three-leg equation \eqref{notYBE13leg} is seen to correspond to the quad equation
\begin{align}
\label{not3D1}
x(u\alpha+y\beta)\beta-y(u-2v\beta)\alpha+\beta\alpha^2=0\,.
\end{align}
The quad equation \eqref{not3D1} does not satisfy the 3D-consistency condition defined in Section \ref{sec:overview}.

\subsection{Examples of non-integrable cases 3}

The final example is the equation
\begin{align}
\label{notYBE2}
\begin{split}
\ds\int_{\mathbb{R}}d\spn_0\,S(\spn_0)\,\oW_1(\theta\,|\,\spn_1,\spn_0)\,W_2({\theta_1+\theta_3}\,|\,\spn_0,\spn_2)\,\oW_3({\theta_3}\,|\,\spn_3,\spn_0)\phantom{\,,} \\
\ds =R(\theta_1,\theta_3)\, W_2({\theta_1}\,|\,\spn_2,\spn_3)\,\oW_4(\theta_1+\theta_3\,|\,\spn_1,\spn_3)\,W_5(\theta_3\,|\,\spn_1,\spn_2)\,,
\end{split}
\end{align}
where the variables are as defined in \eqref{hypvals}.  The equation \eqref{notYBE2} does not have the form of the star-triangle relation given in either \eqref{YBEsym}, or \eqref{YBEasym}. Here
\begin{align}
\label{notYBE2quan}
\begin{split}
\ds W_1(\theta\,|\,\spn_i,\spn_j)&=\EXP^{\overline{B}(-\spn_i-\spn_j+\ii\theta)-\overline{B}(\spn_i+\spn_j+\ii\theta)-\overline{B}(\spn_i-\spn_j+\ii\theta)-\overline{B}(-\spn_i+\spn_j+\ii\theta)}\,, \\[0.1cm]
\ds W_2(\theta\,|\,\spn_i,\spn_j)&=\EXP^{\overline{B}(-\spn_i-\spn_j+\ii\theta)+\overline{B}(\spn_i-\spn_j+\ii\theta)+\overline{B}(-\spn_i+\spn_j+\ii\theta)}\,\Gamma_h(\spn_i+\spn_j+\ii\theta;\bb)\,, \\[0.1cm]
\ds W_3(\theta\,|\,\spn_i,\spn_j)&=\EXP^{\overline{B}(-\spn_i-\spn_j+\ii\theta)+\overline{B}(-\spn_i+\spn_j+\ii\theta)-\overline{B}(\spn_i+\spn_j+\ii\theta)}\,\Gamma_h(\spn_i-\spn_j+\ii\theta;\bb)\,, \\[0.1cm]
\ds W_4(\theta\,|\,\spn_i,\spn_j)&=\frac{1}{W_2(-\theta\,|\,\spn_i,\spn_j)}\,, \\[0.1cm]
\ds W_5(\theta\,|\,\spn_i,\spn_j)&=\frac{1}{W_3(-\theta\,|\,\spn_j,\spn_i)}\,,
\end{split}
\end{align}
\begin{align}
\oW_i(\theta\,|\,\spn_i,\spn_j)=W_i(\eta-\theta\,|\,\spn_i,\spn_j)\,,\qquad i=1,2,3,4,5\,,
\end{align}
where $\eta$ is defined in \eqref{hypeta}, $\overline{B}(z)$ is defined in \eqref{obdef}, and
\begin{align}
\label{notYBE2quans}
S(\spn)=\EXP^{4\pi\eta \spn},\qquad R(\theta_1,\theta_3)=\EXP^{\frac{\pi\ii}{2}\left(B_{2,2}(2(\theta_1+\theta_3);\bb,\bb^{-1})-B_{2,2}(2\theta_1;\bb,\bb^{-1})-B_{2,2}(2\theta_3;\bb,\bb^{-1})\right)}\,.
\end{align}

From \eqref{HGFinflim}, the asymptotics of the integrand of \eqref{notYBE2} are
\begin{align}
\label{notYBE2xinf}
\begin{split}
O(\EXP^{2\pi(\theta_1+\theta_3-\ii(\spn_1+\spn_3))\spn_0})\,,\qquad \spn_0\to-\infty\,, \\
O(\EXP^{-2\pi(\eta-\theta_3+\ii(\spn_1-\spn_2))\spn_0})\,,\qquad \spn_0\to+\infty\,,
\end{split}
\end{align}
and particularly the integral in \eqref{notYBE2}, is absolutely convergent for the values \eqref{hypvals}.

\begin{prop} \label{notYBE2prop}
For the values \eqref{q3d1rvals},
\begin{align}
\label{notYBE2lim}
\begin{split}
\lim_{\kappa\to\infty}\left(\EXP^{4\pi(3\eta+\theta_3-\ii(\spn_1-\spn_2))\kappa} I_{14}(\theta_1+2\ii\kappa,\theta_3+\ii\kappa\,|\,\spn_1+2\kappa,\spn_2+\kappa,\spn_3+\kappa)\right)\phantom{\,,} \\
=\int_{\mathbb{R}}d\spn_0\,S(\spn_0)\,\oW_1(\theta\,|\,\spn_1,\spn_0)\,W_2({\theta_1+\theta_3}\,|\,\spn_0,\spn_2)\,\oW_3({\theta_3}\,|\,\spn_3,\spn_0)\,,
\end{split}
\end{align}
where $I_{14}(\theta_1,\theta_3\,|\,\spn_1,\spn_2,\spn_3)$ is defined in \eqref{q3d1int}, and the Boltzmann weights on the right hand side are defined in \eqref{notYBE2quan}.

\end{prop}

After a change of integration variable $\spn_0\to\spn_0+2\kappa$, the steps of the proof of Proposition \ref{notYBE2prop}, are analogous to the respective steps for Proposition \ref{q3d0prop}.  Equation \eqref{notYBE2}, follows from Proposition \ref{notYBE2prop}, by using \eqref{ratasymp} to take the same limit of the right hand side of \eqref{qq3d1}.

\subsubsection{Classical equations}

The Lagrangian functions for this case are defined by
\begin{align}
\label{lagnotYBE2}
\begin{split}
\lag_1(\alpha\,|\,x_i,x_j)&=\ds 2x_ix_j-\frac{(x_i+x_j+\ii\alpha)^2}{2}-\frac{\pi^2}{3}\,, \\
\lag_2(\alpha\,|\,x_i,x_j)&=\ds \lie(-\EXP^{x_i+x_j+\ii\alpha})\,.
\end{split}
\end{align}

Using the asymptotics \eqref{HGFqcl}, the leading order $O(\hbar^{-1})$ quasi-classical expansion \eqref{hypqcl} of the integrand of \eqref{notYBE2} is
\begin{align}
\begin{split}
\frac{-x_2^2-x_3^2+\alpha_1(\alpha_1+2\alpha_3)-\frac{\pi^2}{2}}{\ii\hbar}+
\Log \rho(\tfrac{\alpha_1}{\sqrt{2\pi\hbar}},\tfrac{\alpha_3}{\sqrt{2\pi\hbar}}\,|\,\tfrac{x_1}{\sqrt{2\pi\hbar}},\tfrac{x_2}{\sqrt{2\pi\hbar}}\tfrac{x_3}{\sqrt{2\pi\hbar}};\tfrac{x_0}{\sqrt{2\pi\hbar}})\phantom{\,,} \\
=(\ii\hbar)^{-1}\astr{x_0}+O(1)\,,
\end{split}
\end{align}
where 
\begin{align}
\rho(\theta_1,\theta_3\,|\,\spn_1,\spn_2,\spn_3;\spn_0)=S(\spn_0)\,\oW_1(\theta\,|\,\spn_1,\spn_0)\,W_2({\theta_1+\theta_3}\,|\,\spn_0,\spn_2)\,\oW_3({\theta_3}\,|\,\spn_3,\spn_0)\,,
\end{align}
and
\begin{align}
\astr{x_0}&=\ol_1(\alpha_1\,|\,x_0,x_1)+\lag_2(\alpha_1+\alpha_3\,|\,x_0,x_2)+\lag_3(\pi-\alpha_3\,|\,x_3,x_0)\,,
\end{align}
where
\begin{align}
\lag_3(\alpha\,|\,x_i,x_j)&=\ds \lag_2(\alpha\,|\,x_i,-x_j)-x_i^2+x_j^2+2\ii x_j(\pi-\alpha)-(\pi-\alpha)^2\,.
\end{align}

The saddle point equation three-leg equation is then
\begin{align}
\label{notYBE23leg}
\left.\frac{\partial \astr{x}}{\partial x}\right|_{x=x_0}\hspace{-0.45cm}=\vphi(\alpha_1\,|\,x_0,x_1)+\phi(\alpha_1+\alpha_3\,|\,x_2,x_0)+\ovphib(\alpha_3\,|\,x_3,x_0)=0\,,
\end{align}
where $\phi(\alpha\,|\,x_i,x_j)$, and $\vphi(\alpha\,|\,x_i,x_j)$, are defined by
\begin{align}
\phi(\alpha\,|\,x_i,x_j)=-\Log(1+\EXP^{x_i+x_j+\ii\alpha})+\ii\alpha\,,\qquad\vphi(\alpha\,|\,x_i,x_j)=x_i+x_j\,,
\end{align}
and
\begin{align}
\ovphib(\alpha\,|\,x_i,x_j)=-\phi(\pi-\alpha\,|\,x_i,-x_j)\,.
\end{align}

With the following change of variables
\begin{align}
x=-\EXP^{x_0},\qquad u=\EXP^{x_1},\qquad y=\EXP^{x_2},\qquad v=\EXP^{x_3},\qquad\alpha=\EXP^{-\ii\alpha_1},\qquad\beta=\EXP^{-\ii(\alpha_3+\alpha_3)},
\end{align}
the three-leg equation \eqref{notYBE23leg} is seen to correspond to the quad equation
\begin{align}
\label{not3D2}
xv\alpha +(1+xuyv)\beta-uv\beta^2=0\,.
\end{align}
The quad equation \eqref{not3D2} does not satisfy the 3D-consistency condition defined in Section \ref{sec:overview}.

\end{appendices}

\bibliography{QABS}
\bibliographystyle{nb}

\end{document}